\documentclass[11pt]{article}
\usepackage[utf8]{inputenc}

\usepackage{
amsmath,amssymb,amsfonts,amsthm}
\usepackage{bm}
\usepackage{bbold}
\usepackage{colonequals}

\usepackage{comment}
\usepackage{epsfig}
\usepackage{latexsym,nicefrac,bbm}
\usepackage{mathtools}
\usepackage{rotating}
\usepackage[table]{xcolor}
\usepackage[shortlabels]{enumitem}

\def\Vhrulefill{\leavevmode\leaders\hrule height 0.5ex depth \dimexpr0.4pt-0.5ex\hfill\kern0pt}

\usepackage
{xspace}
\usepackage{color,fancybox,graphicx,subfigure,fullpage}
\usepackage[top=1.25in, bottom=1.25in, left=1in, right=1in]{geometry}
\usepackage{tabularx}
\usepackage{hyperref}
\usepackage{cleveref}
\usepackage{pdfsync}
\usepackage[boxruled,linesnumbered]{algorithm2e}
\usepackage{multicol}
\usepackage{enumitem}
\usepackage{float}
\usepackage{tikz}
\usepackage{siunitx}

\newcommand{\Erdos}{Erd\H{o}s}
\newcommand{\Renyi}{R\'{e}nyi}

\newtheorem{theorem}{Theorem}[section]
\newtheorem{fact}{Fact}[section]

\newtheorem{conjecture}[theorem]{Conjecture}
\newtheorem{definition}[theorem]{Definition}
\newtheorem{lemma}[theorem]{Lemma}
\newtheorem{remark}[theorem]{Remark}
\newtheorem{proposition}[theorem]{Proposition}

\newcommand{\EE}{\mathbb{E}}

\newcommand{\NN}{\mathbb{N}}
\newcommand{\PP}{\mathbb{P}}
\newcommand{\QQ}{\mathbb{Q}}
\newcommand{\RR}{\mathbb{R}}

\DeclareSymbolFont{bbold}{U}{bbold}{m}{n}
\DeclareSymbolFontAlphabet{\mathbbold}{bbold}

\newcommand{\sG}{\mathcal{G}}
\newcommand{\sL}{\mathcal{L}}
\newcommand{\sS}{\mathcal{S}}

\newcommand{\One}{\mathbbold{1}}

\newcommand{\bi}{\mathrm{bi}}

\newcommand{\what}{\widehat}

\renewcommand{\epsilon}{\varepsilon}
\renewcommand{\Tilde}{\widetilde}
\newcommand{\LoSt}{\mathrm{LoSt}}

\newcommand{\adv}{\mathrm{adv}}
\newcommand{\rand}{\mathrm{rand}}

\newcommand{\Ex}{\mathop{\mathbb{E}}}  

\newcommand{\MC}{\mathrm{MC}}
\newcommand{\dom}{\mathrm{dom}}

\newif\ifnotes
\notestrue

\newcommand{\eps}{\varepsilon}

\title{Computational hardness of detecting graph lifts and certifying lift-monotone properties of random regular graphs}

\usepackage{authblk}

\author{Dmitriy Kunisky\thanks{Email: \texttt{dmitriy.kunisky@yale.edu}. Partially supported by ONR Award N00014-20-1-2335 and a Simons Investigator Award to Daniel Spielman.}\,}

\author{Xifan Yu\thanks{Email: \texttt{xifan.yu@yale.edu}. Partially supported by a Simons Investigator Award to Daniel Spielman.}\,}

\affil{Department of Computer Science, Yale University}

\date{April 25, 2024}

\begin{document}

\maketitle

\begin{abstract}
    We introduce a new conjecture on the computational hardness of detecting random lifts of graphs: we claim that there is no polynomial-time algorithm that can distinguish between a large random $d$-regular graph and a large random lift of a Ramanujan $d$-regular base graph (provided that the lift is corrupted by a small amount of extra noise), and likewise for bipartite random graphs and lifts of bipartite Ramanujan graphs.
    We give evidence for this conjecture by proving lower bounds against the local statistics hierarchy of hypothesis testing semidefinite programs.
    We then explore the consequences of this conjecture for the hardness of certifying bounds on numerous functions of random regular graphs, expanding on a direction initiated by Bandeira, Banks, Kunisky, Moore, and Wein (2021).
    Conditional on this conjecture, we show that no polynomial-time algorithm can certify tight bounds on the maximum cut of random 3- or 4-regular graphs, the maximum independent set of random 3- or 4-regular graphs, or the chromatic number of random 7-regular graphs.
    We show similar gaps asymptotically for large degree for the maximum independent set and for any degree for the minimum dominating set, finding that naive spectral and combinatorial bounds are optimal among all polynomial-time certificates.
    Likewise, for small-set vertex and edge expansion in the limit of very small sets, we show that the spectral bounds of Kahale (1995) are optimal among all polynomial-time certificates.
\end{abstract}

\thispagestyle{empty}

\clearpage

\tableofcontents

\thispagestyle{empty}

\clearpage

\section{Introduction}
\pagenumbering{arabic}

We study the problem of distinguishing a uniformly random $d$-regular graph, whose law we denote $\sG(n, d)$ for a graph on $n$ vertices ($dn$ must be even), from a graph formed as a \emph{random lift} of a base graph.
We also consider a bipartite variant of the same problem for the uniformly random $d$-regular bipartite graph, whose law we denote $\sG((\frac{n}{2}, \frac{n}{2}), d)$ ($n$ must be even).
Here and throughout, $d \geq 3$ is fixed and we consider the limit $n \to \infty$, with the limit $d \to \infty$ sometimes taken afterwards.

\begin{definition}[Random lift \cite{ALMR-2001-RandomLifts}]
    \label{def:random-lift}
    Suppose $H$ is a $d$-regular multigraph on $k$ vertices.
    We write $\sL_m(H)$ for the law of the random multigraph $G$ on $hm$ vertices formed as follows.
    View the $km$ vertices as divided into $k$ \emph{fibers} in correspondence with the vertices of $H$.
    For each non-loop edge of $H$, insert a uniformly random perfect matching between the corresponding fibers in $G$.
    For each loop of $H$, insert a uniformly random perfect matching of the vertices in the corresponding fiber in $G$ (if $H$ has loops, we assume $m$ is even).
\end{definition}
\noindent
This is a random version of a general notion of \emph{graph lift} which we define in Section~\ref{sec:prelim:lifts}.
This is an adequate working definition for our purposes; to fully agree with other literature on lifts we must be somewhat more careful, especially in the treatment of loops in $H$.
We also describe the more standard abstract framework for lifts in Section~\ref{sec:prelim:lifts}.

The main question that will motivate us is: for a fixed $d$-regular $H$, when is it possible to tell, with high probability as $m \to \infty$, whether a graph has been drawn from $\sG(km, d)$ or from $\sL_m(H)$?
Likewise, since when $H$ is bipartite then $G \sim \sL_m(H)$ is also bipartite, when in this case is it possible to tell whether a graph has been drawn from $\sG((\frac{km}{2}, \frac{km}{2}), d)$ or from $\sL_m(H)$?
Actually, as stated, both problems are easy: all eigenvalues of $H$ are also eigenvalues of any lift of $H$, while a large random general or bipartite $d$-regular graph does not have any particular number (other than the ``trivial'' eigenvalues $d$ and, in the bipartite case, $-d$) as an eigenvalue with high probability.
To fool a simple algorithm that examines the spectrum for these unusual regularities, we consider also adding a small amount of further noise to either graph distribution, say by changing a small fraction of edges while maintaining $d$-regularity (see Section~\ref{sec:prelim:noise-conjectures}).

It is unlikely that brittle algebraic algorithms of the above kind will work once we adopt such a noise model.
Still, there are natural algorithmic ideas that compute more robust spectral quantities.
In particular, it is natural to consider whether not particular eigenvalues but the \emph{extreme eigenvalues} of the graph will detect random lifts.

Let $\lambda_1(G) \geq \cdots \geq \lambda_n(G)$ be the ordered eigenvalues of the adjacency matrix of $G$.
Aside from the trivial eigenvalue(s), the typical width of the remaining spectrum of $G$ is known:
\begin{proposition}[Spectral gap of random regular graphs]
    \label{prop:eigenvalues-random-regular}
    The following hold as $n \to \infty$:
    \begin{itemize}
        \item If $G \sim \sG(n, d)$, then $\lambda_1(G) = d$, and with high probability $\lambda_2(G) = 2\sqrt{d - 1} + o(1)$ and $\lambda_n(G) = -2\sqrt{d - 1} + o(1)$ \cite{Nilli-1991-AlonBoppanaBound,Friedman-2003-AlonConjecture,Bordenave-2015-Friedman}.
        \item If $G \sim \sG((\frac{n}{2}, \frac{n}{2}), d)$, then $\lambda_1(G) = d$, $\lambda_n(G) = -d$, and with high probability $\lambda_2(G) = 2\sqrt{d - 1} + o(1)$ and $\lambda_{n-1}(G) = -2\sqrt{d - 1} + o(1)$ \cite{LS-1996-SpectraRegularGraphsBipartite,FL-1996-SpectraHypergraphs,BDH-2022-SpectralGapRandomBipartiteBiregular}.
    \end{itemize}
\end{proposition}

Since under reasonable noise models the \emph{spectral radius}---$\max\{|\lambda_2(G)|, |\lambda_n(G)|\}$ for non-bipartite graphs and $\max\{|\lambda_2(G)|, |\lambda_{n - 1}(G)|\}$ for bipartite ones---is a robust quantity not changed much by a small amount of noise, it is natural to ask a restriction of our question: for what $H$ does $G \sim \sL_m(H)$ have larger spectral radius than $G \sim \sG(n, d)$ or $\sG \sim \sG((\frac{n}{2}, \frac{n}{2}), d)$?

As mentioned above, with probability 1, every eigenvalue of $H$ is also an eigenvalue of a random lift $G \sim \sL_m(H)$.
Therefore, in the non-bipartite case, if $\max\{|\lambda_2(H)|, |\lambda_n(H)|\} > 2\sqrt{d - 1} + \eps$, then $\max\{|\lambda_2(G)|, |\lambda_n(G)|\} > 2\sqrt{d - 1} + \eps$ for $G \sim \sL_m(H)$ with probability 1 (and with high probability is greater than, say, $2\sqrt{d - 1} + \frac{\eps}{2}$ even if we add a sufficiently small amount of noise).
The analogous statement holds for bipartite graphs as well.
Our spectral radius test statistic thus detects graph lifts whenever $H$ fails to be \emph{Ramanujan}.

\begin{definition}[Ramanujan graph]
    Let $H$ be a $d$-regular graph.
    \begin{itemize}
    \item We call $H$ \emph{Ramanujan} if $\max\{|\lambda_2(G)|, |\lambda_n(G)|\} \leq 2\sqrt{d - 1}$.
    \item We call $H$ \emph{bipartite Ramanujan} if $H$ is bipartite and $\max\{|\lambda_2(G)|, |\lambda_{n-1}(G)|\} \leq 2\sqrt{d - 1}$.
    \end{itemize}
\end{definition}

The remarkable work of Bordenave and Collins on the spectra of graph lifts shows that this simple condition is in fact also sufficient.
\begin{proposition}[Spectral gap of random lifts \cite{BC-2019-EigenvaluesRandomLifts}]
    \label{prop:eigenvalues-lifts}
    Suppose $H$ is $d$-regular and $G \sim \sL_m(H)$.
    The following hold with high probability as $m \to \infty$.
    \begin{itemize}
    \item If $H$ is Ramanujan, then $\max\{|\lambda_2(G)|, |\lambda_n(G)|\} \leq 2\sqrt{d - 1} + o(1)$.
    \item If $H$ is bipartite Ramanujan, then $\max\{|\lambda_2(G)|, |\lambda_{n-1}(G)|\} \leq 2\sqrt{d - 1} + o(1)$.
    \end{itemize}
\end{proposition}
\noindent
Their results say much more---after removing the trivial eigenvalues, the entire spectrum of $G \sim \sL_m(H)$ looks like the spectrum of $H$ superimposed on the typical spectrum of a graph drawn from either $\sG(km, d)$ or $\sG((\frac{km}{2}, \frac{km}{2}), d)$, which in either case has the continuous Kesten-McKay distribution \cite{McKay-1981-EigenvalueRegularGraph} supported on $[-2\sqrt{d - 1}, 2\sqrt{d - 1}]$.
Thus, say in the non-bipartite case, when $H$ is not Ramanujan then, compared to the Kesten-McKay spectrum of $G \sim \sG(km, d)$, the spectrum of $G \sim \sL_m(H)$ will have outlier eigenvalues.
When $H$ is Ramanujan, the eigenvalues of $H$ (a finite number of eigenvalues while $km \to \infty$) lie ``hidden'' within the bulk of the Kesten-McKay distribution, so the spectra of $G \sim \sL_m(H)$ and $G \sim \sG(n, d)$ appear indistinguishable (informally speaking).

Our first goal in this paper is to give new and stronger evidence than the above discussion for the following conjecture, which may be viewed as saying that testing for graph lifts using the spectral radius is optimal and that the indistinguishability alluded to above in fact holds.

\begin{definition}[Strong detection]
    Suppose $\PP_n$ and $\QQ_n$ are sequences of probability measures over $\{0, 1\}^{N}$ for some $N = N(n) \in \NN$.
    We say that a function $f: \{0, 1\}^N \to \{\texttt{p}, \texttt{q}\}$ achieves \emph{strong detection} between these sequences if
    \begin{equation}
        \lim_{n \to \infty} \PP_n[f(G) = \texttt{q}] = \lim_{n \to \infty} \QQ_n[f(G) = \texttt{p}] = 0,
    \end{equation}
    that is, if the Type~I and II errors of $f$, viewed as a hypothesis testing procedure, both tend to zero.
\end{definition}

\begin{conjecture}[Hardness of detecting lifts; informal]
    \label{conj:hardness}
    Suppose $H$ is a $d$-regular graph on $k$ vertices.
    The following hold.
    \begin{itemize}
    \item If $H$ is Ramanujan, then there is no polynomial-time algorithm that achieves strong detection between the distributions $\sL_{m}(H)$ with a small amount of extra noise applied and $\sG(km, d)$ as $m \to \infty$.
    \item If $H$ is bipartite Ramanujan, then there is no polynomial-time algorithm that achieves strong detection between the distributions $\sL_{m}(H)$ with a small amount of extra noise applied and $\sG((\frac{km}{2}, \frac{km}{2}), d)$ as $m \to \infty$.
    \end{itemize}
\end{conjecture}
\noindent
The idea for Conjecture~\ref{conj:hardness} arose in the first author's collaboration with the other authors of \cite{BBKMW-2020-SpectralPlantingColoring}, and we are grateful for their permission and encouragement to pursue evidence for and consequences of the Conjecture independently.

We give more precise statements in Section~\ref{sec:prelim:noise-conjectures}.
We will fulfill our first goal by giving lower bounds in the \emph{local statistics hierarchy} of semidefinite programs (SDP), a natural analog of the sum-of-squares hierarchy for hypothesis testing problems proposed by \cite{banks2021local}.

Our second goal is to explore the consequences of the conjecture, if it holds.
These consequences will be derived through the strategy of \emph{quietly planting} solutions to various optimization problems over random graphs.
This argument will imply lower bounds for numerous algorithmic problems of \emph{certifying bounds} on functions of random $d$-regular graphs.
We sketch the idea and describe its consequences below after explaining our local statistics lower bounds.

\subsection{Main Results: Hardness of Detection for the Local Statistics Hierarchy} \label{sec:intro:local-statistics}

The \emph{local statistics (LoSt) hierarchy}, introduced by \cite{banks2021local}, is a sequence of increasingly powerful SDPs for solving hypothesis testing problems.
It is inspired by the sum-of-squares hierarchy, which is a similar construction for optimization problems \cite{Shor-1987-SumOfSquares,Nesterov-1998-SemidefiniteQuadratic,Parrilo-2000-Thesis,Lasserre-2001-GlobalOptimizationMoments}.
In particular, the LoSt hierarchy is based on the \emph{pseudocalibration} proof technique, first used to prove sum-of-squares lower bounds for the problem of computing the clique number of a random graph \cite{BHKKMP-2019-PlantedClique}.

For the sake of exposition, let us focus on the noiseless version of the problem of testing non-bipartite graphs, $\QQ = \QQ_n = \sG(km, d)$ versus $\PP = \PP_n = \sL_m(H)$, where $n = km$.
Let us view a graph drawn from either distribution as having vertices in $[n]$, drawn uniformly at random, and let $H$ have a fixed labelling of its vertices by $[k]$.
Then, under the lift model $G \sim \PP$, there is an underlying partition $\sigma: [n] \to [k]$ specifying the fibers of the lift.
We may describe this by Boolean variables $x_{u, i} = \One\{\sigma(u) = i\}$ for $u \in [n]$ and $i \in [k]$.
Let us also identify $G$ with its adjacency matrix, $G_{u, v} = \One\{u \sim v\} \in \{0, 1\}$.
We view $(x, G)$ as the output of $\PP$ rather than just $G$.

The main objects figuring in LoSt algorithms are low-degree polynomials jointly in these two families of variables, $p \in \RR[x, G]$.
One may show that sufficiently symmetric such polynomials are concentrated around their typical values, $p(x, G) \approx \mathbb{E}_{(x, G)\sim \mathbb{P} } [p(x, G)]$ (to an extent depending on the degree and structure of $p$).
An algorithm seeking to hypothesis test only observes $G$, not $x$.
The idea of LoSt is that this algorithm may try to find a collection of values of $x$ that are compatible with the values specified above.
If such exist, then the algorithm decides that $G$ was drawn from $\PP$, and if not, then it decides that $G$ was drawn from $\QQ$.
One could define an explicit heuristic for finding the $x$ that ``best fit'' the data $G$, such as maximum likelihood estimation.
But such a problem would likely be computationally intractable, since the $x$ are Boolean.
Instead, LoSt forms an SDP relaxation of the problem of deciding whether ``well-fitting'' values of $x$ exist or not, in a similar fashion to sum-of-squares relaxations.

To formulate this SDP, let us first describe the constraints that an admissible $x$ must satisfy.
We follow the presentation of \cite{BBKMW-2020-SpectralPlantingColoring}.
The variables $x_{u, i}$ and $G_{u, v}$ are Boolean, and exactly one of the $x_{u, i}$ is 1 for each $u$.
Thus:
\begin{align*}
    G^2_{u,v} - G_{u,v} &= 0 & & \text{for all } u \neq v \in [n] \\
    x^2_{u,i} - x_{u,i} &= 0 & & \text{for all } u \in [n], i \in [k], \\
    \sum_{i\in [k]}x_{u,i} - 1 &= 0 & & \text{for all } u \in [n].
\end{align*}
We denote by $\mathcal{I} \subset \RR[x, G]$ the ideal generated by the left-hand sides of these equations.

Moreover, $\mathbb{P}$ has the symmetry that the law of $(x,G) \sim \mathbb{P}$ is invariant under the simultaneous action of the symmetric group $S_n$ on $x$ and $G$ by permuting the vertex labels, where a permutation $\xi \in S_n$ acts on $x$ and $G$ by $x_{u,i} \mapsto x_{\xi(u),i}$ and $G_{u,v} \mapsto G_{\xi(u),\xi(v)}$. As a result, the expectation $\mathbb{E}_{(x, G) \sim \mathbb{P}}[p(x, G)]$ is constant on the orbits of this action, and we will restrict our attention to polynomials fixed by the action.
We refer to such polynomials as \emph{$S_n$-invariant}.

We are now ready to define the LoSt algorithm, with the one informality that we leave the meaning of ``$\approx$'' unspecified---ultimately it will be related to quantitative aspects of the concentration of the polynomials $p$ involved.

\begin{definition}[Local statistics; informal]
    \label{def:local-stats-informal}
    The \emph{degree-$(D_x,D_G)$ local statistics algorithm}, denoted $\mathrm{LoSt}(D_x, D_G)$, is the following feasibility SDP: given an input graph $G_0$, find $\Tilde{\mathbb{E}} : \mathbb{R}[x]_{\le D_x} \to \mathbb{R}$ satisfying:
    \begin{enumerate}
        \item (Positivity) $\Tilde{\mathbb{E}} [p(x)^2] \ge 0$ whenever $\deg p \le D_x / 2$.
        \item (Hard constraints) $\Tilde{\mathbb{E}} [p(x, G_0)] = 0$ for every $p \in \mathcal{I}_k$ with $\deg p \le D_x$.
        \item (Soft moment calibration; informal) $\Tilde{\mathbb{E}} [p(x,G_0)] \approx \mathbb{E}_{(x,G)\sim\PP} [p(x,G)]$ whenever $\deg_G p(x,G) \le D_G$, $\deg_x p(x,G) \le D_x$, and $p$ is $S_n$-invariant.
    \end{enumerate}
\end{definition}
\noindent
This program may be written as a feasibility SDP using the same transformation as is standard for the sum-of-squares hierarchy; see, e.g., \cite{Laurent-2009-SOS}.
Feasibility SDPs in general may be solved with the classical ellipsoid algorithm \cite{GLS-2012-GeometricAlgorithmsCombinatorialOptimization}.

\begin{remark}[O'Donnell's caveat]
    For the sum-of-squares hierarchy, \cite{ODonnell-2017-SOSNotAutomatizable} raised an important issue noting that, even when the SDPs generated are of polynomial size, they might not be solvable efficiently with the ellipsoid algorithm and relatives thereof if their solutions involve very large numbers.
    This was addressed for many instances of sum-of-squares relaxations of combinatorial optimization problems by \cite{RW-2017-BitComplexity}.
    It seems likely that LoSt SDPs suffer from a similar issue and that a parallel theory should be developed for them, but we do not pursue this direction here.
\end{remark}
\noindent
Our main result is as follows.

\begin{theorem}[Local statistics hardness; informal]
    \label{thm:local-statistics-informal}
    Let $H$ be a d-regular multigraph on $k$ vertices.
    Suppose that $H$ is not bipartite.
    Then, the following hold with high probability as $m \to \infty$.
    \begin{itemize}
        \item If $H$ is Ramanujan, then for any $D \geq 2$, $\LoSt(2,D)$ with any fixed error tolerance cannot distinguish $\mathcal{L}_m(H)$ and $\mathcal{G}(km, d)$ with high probability as $m \to \infty$, nor can it distinguish these when any amount of extra noise is added to $\sL_m(H)$.
        \item If $H$ is not Ramanujan, then there exists a constant $D$ and a fixed error tolerance such that the degree $\LoSt(2,D)$ can distinguish $\mathcal{L}_m(H)$ and $\mathcal{G}(km, d)$ with high probability as $m \to \infty$, and can still distinguish these distributions if a sufficiently small amount of extra noise is added to $\sL_m(H)$.
    \end{itemize}
    Suppose now that $H$ is bipartite.
    Then, the following hold with high probability as $m \to \infty$.
    \begin{itemize}
        \item If $H$ is bipartite Ramanujan, then for any $D \geq 2$, $\LoSt(2,D)$ with any fixed error tolerance cannot distinguish $\mathcal{L}_m(H)$ and $\mathcal{G}((\frac{km}{2}, \frac{km}{2}), d)$ with high probability as $m \to \infty$, nor can it distinguish these distributions when any amount of extra noise is added to $\sL_m(H)$.
        \item If $H$ is not bipartite Ramanujan, then there exists a constant $D$ and a fixed error tolerance such that the degree $\LoSt(2,D)$ can distinguish $\mathcal{L}_m(H)$ and $\mathcal{G}((\frac{km}{2}, \frac{km}{2}), d)$ with high probability as $m \to \infty$, and can still distinguish these distributions if a sufficiently small amount of extra noise is added to $\sL_m(H)$.
    \end{itemize}
\end{theorem}

\subsection{Main Results: Applications}

We now describe the consequences of Conjecture~\ref{conj:hardness}.
A valuable task for many interesting functions of a graph---the independence number, chromatic number, maximum cut, minimum dominating set, expansion, and so forth---is to \emph{certify bounds}.

\begin{definition}[Certification]
    Let $f(G)$ be a real-valued function of a $d$-regular graph $G$.
    We say that a function $c(G)$ (in our setting thought of as computed by some algorithm) \emph{certifies an upper bound} on $G$ if $f(G) \leq c(G)$ for all $G$, and \emph{certifies a lower bound} on $G$ if $f(G) \geq c(G)$ for all $G$.
\end{definition}
\noindent
The language may seem convoluted, but is helpful when we discuss the performance of algorithms over random $G$.
If $G \sim \sG(n, d)$, then the typical value of $f(G) \approx \EE f(G)$ may be small, say, but it may be hard to \emph{certify} a tight bound $c(G)$, one such that $\EE c(G) \approx \EE f(G)$, because a certification algorithm cannot just output, say, $c(G) \colonequals \EE f(G) + \eps$, but rather when given $G$ with an unusual structure atypical to $\sG(n, d)$ and making $f(G)$ very large must adapt and output a correspondingly larger bound.
In this sense, certification algorithms that perform well in the average case must perform a kind of \emph{implicit hypothesis testing} and detect whether they are being given typical or atypical instances for the distribution in question.

The prototypical strategy for designing certification algorithms is \emph{convex relaxation}, forming a linear or semidefinite program that provably bounds the independence number (or some other quantity), using a framework such as the aforementioned sum-of-squares hierarchy.
Unfortunately, convex relaxations, especially at higher and more powerful tiers of various hierarchies such as sum-of-squares, are notoriously technical to analyze.
In particular it is a long-outstanding challenge in this literature to prove lower bounds for certification over distributions with complex correlation structure like $\sG(n, d)$.
Our approach of \emph{quiet planting} is a technique for assessing the performance of \emph{all} efficient certification algorithms without directly studying convex optimization itself, conditional on a hardness assumption such as Conjecture~\ref{conj:hardness}.
This strategy has been employed for various other problems by \cite{BKW-2019-ConstrainedPCA,BKW-2020-PositivePCA,BBKMW-2020-SpectralPlantingColoring,Kunisky-2023-OptimalityGlauberDynamicsIsing}.

The relevance of the Conjecture is that all of the quantities above are \emph{lift-monotone}: for instance, the independence number, as a fraction of the number of vertices, can only increase after forming a graph lift.
So, a random lift of a Ramanujan graph $H$ with a large independent set---larger than is typically present in $G \sim \sG(n, d)$---always still has a large independent set.
Conversely, if an algorithm performs certification and usually outputs a tight bound on $G \sim \sG(n, d)$, then it could be used to distinguish $G \sim \sG(n = km, d)$ from $G \sim \sL_m(H)$: in the first case it would usually output a small bound, while in the second case it must output a large bound (since the actual independence number is large).
Thresholding the output of this algorithm would therefore effectively test between $\sG(n, d)$ and $\sL_m(H)$, contradicting the Conjecture.

The consequences of the quiet planting argument together with Conjecture~\ref{conj:hardness} are surprising: the existence of particular finite Ramanujan graphs with various extremal properties---large independent sets, large cuts, low expansion, and so forth---constitutes evidence that certifying a bound on the same quantity for random regular graphs is computationally hard.
This gives a powerful new technique for proving computational lower bounds for certification by merely exhibiting particular finite graphs.
In Section~\ref{sec:certification-theory}, we give some general theory around this strategy; for now, we restrict our attention to concrete examples.
We will discuss numerous quantities $f(G)$; in each case, we present a simple certification algorithm---often a spectral bound---and give results indicating how close to optimal this simple strategy must be.
We summarize these findings in Table~\ref{table:results}.

\paragraph{Maximum $t$-cut}
We first consider the following generalization of the max-cut problem to ``$t$-cuts'' or ``multisections'' for a graph $G = (V, E)$:
\begin{equation}
    \MC_t(G) \colonequals \max_{\kappa : V \to [t]} \frac{|\{\{u, v\} \in E: \kappa(u) \neq \kappa(v)\}|}{|E|} \in [0, 1].
    \label{eq:MC-def}
\end{equation}
$\MC_2(G)$ is just the fraction of edges cut by a maximum cut, for example.
A classical spectral certificate due to Hoffman \cite{Hoffman-1970-Eigenvalues} gives the bound
\begin{equation}
    \label{eq:hoffman-MCk}
    \MC_t(G) \leq \frac{t - 1}{t}\left(1 + \frac{|\lambda_n(G)|}{d}\right).
\end{equation}
The question of what bound polynomial-time certificates can achieve for a given $t$ was studied by \cite{BBKMW-2020-SpectralPlantingColoring}.
In particular, that work showed that, as $d \to \infty$ for fixed $t$, Hoffman's bound is optimal to leading order: by Proposition~\ref{prop:eigenvalues-random-regular}, it certifies $\MC_t(G) \leq \frac{t - 1}{t}(1 + \frac{2}{\sqrt{d}} + o(\frac{1}{\sqrt{d}}))$ with high probability, and, conditional on Conjecture~\ref{conj:hardness} (which is a generalization of their ``equitable stochastic block model conjecture''), no polynomial-time certification algorithm can improve on this asymptotically.
In contrast, the true value scales with high probability as $\MC_t(G) = \frac{t - 1}{t}(1 + \frac{P_*^{(t)}}{\sqrt{d}} + o(\frac{1}{\sqrt{d}}))$ for a number $P_*^{(t)} < 2$ related to the ground state energy of an associated spin glass model from statistical physics \cite{Sen-2018-OptimizationSparseHypergraph}.
Thus, as $d \to \infty$, there is a sizeable \emph{certification gap} between the true value of $\MC_t(G)$ and what polynomial-time certificates can achieve.\footnote{Since an algorithm with no runtime constraints can solve maximum $t$-cut as well as all of the combinatorial problems we discuss by brute force, a certification gap may be viewed a particular kind of \emph{information-computation gap}, describing a computational task that can be solved but only inefficiently.}

On the other hand, as discussed in Appendix A.2 of \cite{BBKMW-2020-SpectralPlantingColoring}, for finite $d$, their techniques are limited because they only work with a small class of Ramanujan graphs whose construction depends on divisibility relations between $d$ and $t$ (see our discussion in Section~\ref{sec:related}).
For instance, for $d = 3$ and $t = 2$ (maximum cut in random cubic graphs), their results only show that no polynomial-time certificate can bound $\MC_2(G) \leq \num{0.67}$ with high probability when $G \sim \sG(n, 3)$.
This result is actually vacuous, because it is known that with high probability $\MC_2(G) \geq \num{0.90}$.

Using our more flexible approach to quiet planting, we are able to resolve this issue and a similar one for $d = 4$, and show that there is a certification gap in both cases.

\begin{remark}[Statistical-to-computational gap in detecting lifts]
    Our proofs of these certification gaps also imply that, in some cases, it is information-theoretically possible to distinguish $\sL_m(H)$ from $\sG(n, d)$ or $\sG((\frac{n}{2}, \frac{n}{2}), d)$ while Conjecture~\ref{conj:hardness} implies that it is impossible to do so in polynomial time; that is, in some cases the problem of detecting lifts exhibits a \emph{statistical-to-computational gap}.
    This is because we may just solve, say, maximum $k$-cut by brute force and threshold the resulting value---if the typical value for a random lift is larger than the typical value for a uniformly random graph, then this test will detect the random lift with high probability.
\end{remark}

\begin{theorem}[Maximum cut, $d \in \{3, 4\}$]
    \label{thm:max-cut}
    For $G \sim \sG(n, 3)$, the following hold:
    \begin{enumerate}
        \item $\MC_2(G) \in [\num{0.906}, \num{0.925}]$ with high probability \cite{GL-2018-MaxCutSparseRandomGraphs,COLMS-2022-MaxCutRandomRegularIsing}.
        \item Hoffman's bound \eqref{eq:hoffman-MCk} certifies $\MC_2(G) \leq \num{0.971}$ with high probability.
        \item If Conjecture~\ref{conj:hardness} holds, then there is no polynomial-time algorithm certifying $\MC_2(G) \leq \num{0.944}$ with high probability.
    \end{enumerate}
    For $G \sim \sG(n, 4)$, the following hold:
    \begin{enumerate}
        \item $\MC_2(G) \in [\num{0.833}, \num{0.869}]$ with high probability \cite{DDSW-2003-MaxMinBisection34Regular,COLMS-2022-MaxCutRandomRegularIsing}.
        \item Hoffman's bound \eqref{eq:hoffman-MCk} certifies $\MC_2(G) \leq \num{0.933}$ with high probability.
        \item If Conjecture~\ref{conj:hardness} holds, then there is no polynomial-time algorithm certifying $\MC_2(G) \leq \num{0.875}$ with high probability.
    \end{enumerate}
\end{theorem}
\noindent
As mentioned earlier, our proof constructs concrete 3- and 4-regular Ramanujan graphs (see Figures~\ref{fig:d3-k2-example} and \ref{fig:d4-k2-example}, respectively).
These were found by computer-assisted experimentation; it seems likely that our methods can be extended to larger $d$.

\paragraph{Chromatic number}
A graph $G$ being $k$-colorable is the same as having $\MC_k(G) = 1$, or admitting a $k$-cut that cuts every edge.
Thus Hoffman's bound \eqref{eq:hoffman-chromatic} on $\MC_k(G)$ also yields a spectral lower bound certificate for $\chi(G)$, which works out to:
\begin{equation}
    \label{eq:hoffman-chromatic}
    \chi(G) \geq 1 + \left\lceil \frac{d}{|\lambda_n(G)|} \right\rceil.
\end{equation}
As above, the results of \cite{BBKMW-2020-SpectralPlantingColoring} imply that this is asymptotically optimal as $d \to \infty$.
The bound scales as $\frac{\sqrt{d}}{2}(1 + o(1))$, conditional on Conjecture~\ref{conj:hardness} no polynomial-time certificate can improve on this behavior, and again there is a large certification gap as the true chromatic number is known to scale as the much larger $\frac{d}{2\log d}$ \cite{FL-1992-IndependenceChromaticRandomRegular}.

Again for finite $d$ the results of \cite{BBKMW-2020-SpectralPlantingColoring} are suboptimal.
For example, Hoffman's bound is tight for $d \in \{3, 4, 5\}$, the analysis of \cite{BBKMW-2020-SpectralPlantingColoring} gives a certification gap for $d = 6$ (Hoffman's bound with high probability certifies $\chi(G) \geq 3$ and the results of \cite{BBKMW-2020-SpectralPlantingColoring} imply this is optimal conditional on Conjecture~\ref{conj:hardness}, while by the results of \cite{AM-2004-ColoringRandomRegular} the true chromatic number is with high probability $\chi(G) \in \{4, 5\}$), but fail already for $d = 7$.
We rectify this:
\begin{theorem}[Chromatic number, $d = 7$]
    \label{thm:chromatic}
    For $G \sim \sG(n, 7)$, the following hold:
    \begin{enumerate}
        \item $\chi(G) \geq 4$ with high probability \cite{AM-2004-ColoringRandomRegular}.
        \item Hoffman's bound \eqref{eq:hoffman-chromatic} certifies $\chi(G) \geq 3$ with high probability.
        \item If Conjecture~\ref{conj:hardness} holds, then there is no polynomial-time algorithm certifying $\chi(G) \geq 4$ with high probability.
    \end{enumerate}
\end{theorem}
\noindent
As we discuss in Section~\ref{sec:pf:chromatic-number}, some extra care is needed in the choice of a formal version of Conjecture~\ref{conj:hardness} to use for this application, because the chromatic number is a less robust quantity than some of the others we consider; if we are not careful, indiscriminately changing edges in $G \sim \sL_m(H)$ can invalidate the coloring ``lifted from'' $H$.

\paragraph{Independence number}

As for cuts, to compare graphs of different sizes we work with the \emph{normalized independence number} $\what{\alpha}(G) \colonequals \alpha(G) / |V(G)|$, the fraction of vertices in a maximum independent set.
There is again a natural spectral upper bound certificate following from Hoffman's methods in \cite{Hoffman-1970-Eigenvalues} (see also \cite{Haemers-2021-HoffmanRatioBound} for this version), showing:
\begin{equation}
    \label{eq:hoffman-ind}
    \what{\alpha}(G) \leq \frac{|\lambda_n(G)|}{d + |\lambda_n(G)|}.
\end{equation}
As $d \to \infty$, Hoffman's bound scales as $\frac{2}{\sqrt{d}}(1 + o(1))$, while the true normalized independence number is known to scale as $\frac{2\log d}{d}(1 + o(1))$ \cite{FL-1992-IndependenceChromaticRandomRegular}.
In this case, there is no prior work that we are aware of on certifying bounds on the independence number (though \cite{JPRTX-2021-SOSSparseIndependentSet} prove sum-of-squares lower bounds for the somewhat similar though denser setting of \Erdos-\Renyi\ graphs with average degree polylogarithmic in $n$), and a certification gap is not known to exist either as $d \to \infty$ or for finite $d$.
We prove that it does in both settings:

\begin{theorem}[Independence number, $d \to \infty$]
    \label{thm:ind-set-asymp}
    If Conjecture~\ref{conj:hardness} holds,  then for any $\eps > 0$, there is no polynomial-time algorithm certifying that $\what{\alpha}(G) \leq \frac{2}{\sqrt{d}}(1 - \eps + o(1))$ with high probability when $G \sim \sG(n, d)$, with $o(1)$ referring to the limit as $d \to \infty$.
\end{theorem}

\begin{theorem}[Independence number, $d \in \{3, 4\}$]
    \label{thm:ind-set}
    For $G \sim \sG(n, 3)$, the following hold:
    \begin{enumerate}
        \item $\what{\alpha}(G) \in [\num{0.445}, \num{0.451}]$ with high probability \cite{Csoka-2016-IndependentSetsCutsLargeGirthRegular,Harangi-2023-ReplicaBoundsIndependenceRandomRegular}.
        \item Hoffman's bound \eqref{eq:hoffman-ind} certifies $\what{\alpha}(G) \leq \num{0.485}$ with high probability.
        \item If Conjecture~\ref{conj:hardness} holds, then there is no polynomial-time algorithm certifying $\what{\alpha}(G) \leq \num{0.458}$ with high probability.
    \end{enumerate}
    For $G \sim \sG(n, 4)$, the following hold:
    \begin{enumerate}
        \item $\what{\alpha}(G) \in [\num{0.404}, \num{0.412}]$ with high probability \cite{Csoka-2016-IndependentSetsCutsLargeGirthRegular,Harangi-2023-ReplicaBoundsIndependenceRandomRegular}.
        \item Hoffman's bound \eqref{eq:hoffman-ind} certifies $\what{\alpha}(G) \leq \num{0.464}$ with high probability.
        \item If Conjecture~\ref{conj:hardness} holds, then there is no polynomial-time algorithm certifying $\what{\alpha}(G) \leq \num{0.428}$ with high probability.
    \end{enumerate}
\end{theorem}

\paragraph{Domination number}
The \emph{domination number} of a graph is size of the smallest dominating set, a set of vertices such that every vertex not contained in the set is adjacent to a member of the set.
The \emph{normalized domination number} of a graph $G$ is the domination number divided by $|V(G)|$, which we denote $\dom(G)$.
The domination number admits a certificate even simpler than the Hoffman-type spectral bounds: in a $d$-regular graph $G$, any vertex in a dominating set can only ``cover'' $d+1$ vertices of $G$, namely itself and its $d$ neighbors.
Therefore, we have the ``trivial'' bound
\begin{equation}
    \label{eq:trivial-dom}
    \dom(G) \geq \frac{1}{d + 1}.
\end{equation}
Yet, our methods show that the trivial bound is optimal non-asymptotically for any fixed $d$ (in contrast to the results on the $d \to \infty$ limit for the problems discussed above, where there is a $o(1)$ error term as $d \to \infty$), and establish a certification gap for all sufficiently large $d$.
\begin{theorem}[Domination number, any $d$]
    \label{thm:dom}
    There exist absolute constants $C_1, C_2 > 0$ such that, for any $d \geq 3$ and $G \sim \mathcal{G}(n,d)$, the following hold:
    \begin{enumerate}
        \item $C_1 \frac{\log d}{d} \leq \dom(G) \leq C_2 \frac{\log d}{d}$ with high probability \cite{alon2010high}.
        \item The trivial bound \eqref{eq:trivial-dom} certifies $\dom(G) \geq \frac{1}{d + 1}$ always.
        \item If Conjecture~\ref{conj:hardness} holds, then, for any $\epsilon > 0$, there is no polynomial-time algorithm certifying $\dom(G) \geq \frac{1}{d + 1} + \eps$ with high probability.
    \end{enumerate}
\end{theorem}
\noindent
We note that our analysis of certification is optimal; establishing a certification gap for, say, $d = 3$ is a matter of more precisely establishing the scaling of the true domination number, which to the best of our knowledge is an open problem.

\paragraph{Vertex expansion} The small-set vertex expansion $\Phi_{\epsilon}^{v}$ of $G = (V, E)$ on $n$ vertices is the minimum vertex expansion among subsets of vertices of size at most $\epsilon n$:
\begin{equation}
    \Phi_{\epsilon}^{v}(G) \colonequals \min_{\substack{S \subseteq V \\ 1 \leq |S| \leq \epsilon n}} \frac{|\{u \in V: u \text{ has a neighbor in } S\}|}{|S|}
\end{equation}
A classical result of Kahale \cite{Kahale-1995-SpectralBoundExpansion} gives the following spectral bound on the vertex expansion. Let $\Tilde{\lambda}(G) \colonequals \max(\lambda_2(G), 2\sqrt{d - 1})$. Then, for an absolute constant $C > 0$,
\begin{equation}
    \Phi_{\epsilon}^{v}(G) \geq \frac{d}{2}\left(1 - \sqrt{1 - \frac{4(d - 1)}{\Tilde{\lambda}(G)^2}}\right)\left(1 - C \cdot \frac{\log d}{\log \frac{1}{\epsilon}}\right). \label{eq:kahale-vertex}
\end{equation}
This remains the state-of-the-art polynomial-time certificate for the vertex expansion of $d$-regular graphs. Using results on number-theoretic constructions of Ramanujan graphs from \cite{kamber2022combinatorics}, we give evidence that Kahale's bound is optimal in the limit of $\varepsilon \to 0$, for an infinite sequence of $d$.

\begin{theorem}[Small-set vertex expansion]
    \label{thm:vertex-exp}
    For $G \sim \mathcal{G}\left(\left(\frac{n}{2}, \frac{n}{2}\right), d\right)$, the following hold:
    \begin{enumerate}
        \item $\Phi_{\epsilon}^{v}(G) \ge d - 1 - O_d\left(\frac{1}{\log \frac{1}{\epsilon}}\right)$ with high probability \cite{hoory2006expander}.
        \item Kahale's bound \eqref{eq:kahale-vertex} certifies $\Phi_{\epsilon}^{v}(G) \geq \frac{d}{2}\left(1 - O_d \left(\frac{1}{\log \frac{1}{\epsilon}}\right)\right) - o(1)$ with high probability. Taking $\epsilon \to 0$, this bound scales as $\Phi_{\epsilon}^{v}(G) \ge \frac{d}{2} - o(1)$.
        \item If Conjecture~\ref{conj:hardness} holds and $d = q+1$ where $q$ is a prime power, then for any $\delta, \eps > 0$, there is no polynomial-time algorithm certifying $\Phi_{\varepsilon}^v(G) \ge \frac{d}{2} + \delta$.
    \end{enumerate}
\end{theorem}

\paragraph{Edge expansion} The similar but distinct quantity of small-set edge expansion $\Phi_{\epsilon}^{e}$ is similarly defined as the minimum edge expansion among subsets of vertices of size at most $\epsilon n$:
\begin{equation}
    \Phi_{\epsilon}^{e}(G) \colonequals \min_{\substack{S \subseteq V \\ 1 \leq |S| \leq \epsilon n}} \frac{|\{e \in E: e \text{ has exactly one endpoint in } S\}|}{|S|}
\end{equation}
Kahale also gives the following spectral bound on this quantity.
For another absolute constant $C > 0$,
\begin{equation}
    \Phi_{\epsilon}^{e}(G) \geq d - \left(1 + \frac{\Tilde{\lambda}(G)}{2} + \sqrt{\frac{\Tilde{\lambda}(G)^2}{4} - (d-1)}\right)\left(1 + C \cdot \frac{\log d}{\log \frac{1}{\epsilon}}\right). \label{eq:kahale-edge}
\end{equation}
Again, we use results from \cite{kamber2022combinatorics} to give evidence that this bound is optimal in the limit of $\varepsilon \to 0$ for an infinite sequence of $d$.

\begin{theorem}[Small-set edge expansion]
    \label{thm:edge-exp}
    For $G \sim \mathcal{G}\left(\left(\frac{n}{2}, \frac{n}{2}\right), d\right)$, the following hold:
    \begin{enumerate}
        \item $\Phi_{\epsilon}^{e}(G) \ge d - 2 - O_d\left(\frac{1}{\log \frac{1}{\epsilon}}\right)$ with high probability \cite{hoory2006expander}.
        \item Kahale's bound \eqref{eq:kahale-edge} certifies $\Phi_{\epsilon}^{e}(G) \geq d - (\sqrt{d-1} + 1)\left(1 + O_d \left(\frac{1}{\log \frac{1}{\epsilon}}\right)\right) - o(1)$ with high probability. Taking $\epsilon \to 0$, this bound scales as $\Phi_{\epsilon}^{e}(G) \ge d - 1 - \sqrt{d-1} - o(1)$.
        \item If Conjecture~\ref{conj:hardness} holds and $d = q^2+1$ where $q$ is a prime power, then for any $\delta, \eps > 0$, there is no polynomial-time algorithm certifying $\Phi_{\varepsilon}^e(G) \ge d - 1 - \sqrt{d-1} + \delta$.
    \end{enumerate}
\end{theorem}

\subsection{Related Work}
\label{sec:related}

\paragraph{Local statistics}
As we have mentioned, the local statistics hierarchy is heavily motivated by the pseudocalibration proof technique for sum-of-squares lower bounds, due to \cite{BHKKMP-2019-PlantedClique} and since extended to other settings by works such as \cite{GJJPR-2020-SK,PR-2020-MachinerySOS,JPRTX-2021-SOSSparseIndependentSet}.
Its original use in \cite{banks2021local} was as a more robust algorithm for detection under a version of the stochastic block model over regular graphs.
The later work \cite{BBKMW-2020-SpectralPlantingColoring} used local statistics to argue for the hardness of a special case of the problem we propose, which can also be viewed as detection in an ``equitable'' variant of the stochastic block model.

\paragraph{Quiet planting for certification lower bounds}
Earlier work, such as that on sum-of-squares relaxations for the planted clique problem culminating in \cite{BHKKMP-2019-PlantedClique}, somewhat conflated certification with hypothesis testing, since in that case a straightforward planted distribution (just adding a clique into a random graph) shows an optimal lower bound on certification.
Arguing as we do about the hardness of average-case certification by more carefully quietly planting unusual structures originates (to the best of our knowledge) in the approach of \cite{BKW-2019-ConstrainedPCA} to certifying bounds on certain Gaussian optimization problems.
This Gaussianity is essential to that approach of \emph{spectral planting}, which involves planting unusual near-eigenvectors in a random matrix by rotating its frame of eigenvectors, in contrast to merely adding a large low-rank perturbation.

One of the innovations of \cite{BBKMW-2020-SpectralPlantingColoring} is to implement a similar idea in the adjacency matrix of a random graph without spoiling the graph structure using random lifts.
In fact, as discussed there, random lifts may be seen as a natural generalization of the previous spectral planting approach because of the property we have mentioned that graph lifts preserve the eigenvalues of the base graph while adding further random eigenvalues around them.
A similar strategy for analyzing certification was used by \cite{BKW-2020-PositivePCA,KVWX-2023-LowDegreeColoringClique}, while \cite{Kunisky-2023-OptimalityGlauberDynamicsIsing} found that a quiet planting argument could also be used to argue about a sampling problem.\footnote{The work \cite{KVWX-2023-LowDegreeColoringClique} also includes an intriguing result, though specific to a different framework involving low-degree polynomial algorithms, showing a kind of ``completeness'' of quiet planting: whenever a certification or refutation task of the kind we are discussing is hard, then this hardness is witnessed by a hard hypothesis testing problem.}

Generally speaking, while it would be encouraging to also have direct lower bounds against certification or sampling algorithms in these situations, it appears that quiet planting arguments allow us to give evidence of computational hardness in situations that are otherwise inaccessible.
In our case of certification, proving lower bounds against the sum-of-squares hierarchy for, e.g., the maximum independent set in a random regular graph of constant degree is a long-standing and seemingly technically challenging open problem (see the discussion of open problems in \cite{JPRTX-2021-SOSSparseIndependentSet}) which our approach allows us to circumvent while still producing evidence of hardness.

\paragraph{Quiet planting with lifts of Ramanujan graphs}
We give a few more details about the specific approach of \cite{BBKMW-2020-SpectralPlantingColoring}.
Their technique may be viewed as the special case of ours that uses Ramanujan graphs that are complete graphs $H$ with every vertex having some number $a$ of loops and every edge between distinct vertices repeated some number $b$ times.
Thus their adjacency matrix is $(a - b)I + bJ$ and it is straightforward to determine when such a graph is Ramanujan.
The drawback of this restriction is that, if $H$ has $k$ vertices, we must have $d = a + (k - 1)b$, so for small $d$ there are very few options to draw from in this class.
Say for $d = 3$, there is just a single 3-regular Ramanujan $H$ of this class, which is $H = K_4$.
In contrast, our approach allows for any of the many known Ramanujan 3-regular graphs, or sporadic ones found by hand or by computer experiment, to be used for quiet planting.

\subsection{Open Problems}

We mention several open problems that our work motivates:
\begin{enumerate}
    \item It is an outstanding technical challenge to prove lower bounds against the local statistics hierarchy for $D_x > 2$.
    This is analogous to the difficulty (now overcome in some cases) for proving higher-degree lower bounds in the sum-of-squares hierarchy. The setting of detecting graph lifts is a natural testbed for techniques for this question.
    \item Similarly, it is an open problem to prove sum-of-squares lower bounds for arbitrary constant degrees of the hierarchy for any of the quantities considered here for random $d$-regular graphs with constant $d$ as $n \to \infty$.
    \item Other forms of evidence for Conjecture~\ref{conj:hardness} would also be encouraging.
    One interesting direction is to show that low-degree polynomials cannot detect lifts of Ramanujan graphs (see, e.g., \cite{KWB-2022-LowDegreeNotes} for a survey of this algorithmic framework). In that context, this raises another intriguing technical challenge because, unlike the vast majority of cases treated in this framework, neither of the distributions $\sG(n, d)$ nor $\sL_m(H)$ is a product measure.
    \item As we discuss in Section~\ref{sec:certification-theory}, our approach gives an abstract lower bound against certification up to the extremal value of a given quantity over Ramanujan graphs.
    This suggests two interesting questions: first, what is the largest cut, largest independent set, smallest chromatic number, and so forth over all Ramanujan graphs of a given degree?
    And second, are there algorithms (which would have to improve upon the Hoffman-type bounds we have discussed) that achieve these thresholds?
    \item In principle, our approach is not restricted to lift-monotone properties.
    As we discuss in Remark~\ref{rem:beyond-lift-monotone}, we may execute our strategy so long as we have some procedure for controlling how a quantity behaves under random lifts.
    While monotonicity is mostly enough for our purposes, prior work has studied this more precise question for independent sets, the chromatic number, and edge expansion \cite{ALM-2002-RandomLiftsIndependenceChromatic,AL-2006-RandomLiftsExpansion}, for instance.
    The applications we propose perhaps give a new motivation to revisit these matters.
    \item There are also many other quantities to which it would be natural to apply our tools and which are lift-monotone or some simple variation thereof.
We have chosen ones for which at least the true value for random regular graphs is relatively well-understood and mostly ones for which there is a known benchmark certification algorithm, but other, less widely studied functions of a graph that seem amenable to our methods include the logarithmic Sobolev constant~\cite{DSC-1996-LSIFiniteMarkovChains,FF-2023-SOSProofsLSI}, the $\lambda$-clustering coefficient~\cite{Bilu-2006-ExtensionsHoffmanBound}, the $\psi$-covering coefficient~\cite{Bilu-2006-ExtensionsHoffmanBound}, the $k$-independence number~\cite{ACT-2016-SpectralBoundsKIndependence}, and the size of the minimum independent dominating set~\cite{DW-2002-IndependentDominatingSetCubic}.
\end{enumerate}

\begin{sidewaystable}
    \begin{center}
    \begin{tabular}{lllll}
    \hline \\[-0.75em]
    \textbf{Problem} & $d$ & \textbf{True Value} & \textbf{Lower Bound} & \textbf{Certificate} \\[0.5em]
    \hline \\[-0.75em]
    Max $t$-cut & $\to \infty$ & $\frac{t - 1}{t}(1 + \frac{P_*^{(t)}}{\sqrt{d}})$ \cite{Sen-2018-OptimizationSparseHypergraph} & $\frac{t - 1}{t}(1 + \frac{2}{\sqrt{d}})$ \cite{BBKMW-2020-SpectralPlantingColoring} & $\frac{t - 1}{t}(1 + \frac{2}{\sqrt{d}})$ \cite{Hoffman-1970-Eigenvalues} \\
    Max 2-cut & 3 & $\in [\num{0.906}, \num{0.925}]$ \cite{GL-2018-MaxCutSparseRandomGraphs,COLMS-2022-MaxCutRandomRegularIsing} & 0.944 $^{(\text{n})}$ & 0.971 \\
    & 4 & $\in [0.833, 0.869]$ \cite{DDSW-2003-MaxMinBisection34Regular,COLMS-2022-MaxCutRandomRegularIsing} & 0.875 $^{(\text{n})}$ & 0.933
    \\[0.5em]
    \hline \\[-0.5em]
    Max independent set & $\to \infty$ & $\frac{2\log d}{d}$ \cite{FL-1992-IndependenceChromaticRandomRegular} & $\frac{2}{\sqrt{d}}$ \,$^{(\text{n})}$ & $\frac{2}{\sqrt{d}}$ \cite{Haemers-2021-HoffmanRatioBound} \\
     & 3 & $\in [\num{0.445}, \num{0.451}]$ \cite{Csoka-2016-IndependentSetsCutsLargeGirthRegular,Harangi-2023-ReplicaBoundsIndependenceRandomRegular} & 0.458 $^{(\text{n})}$ & 0.485 \\
     & 4 & $\in [\num{0.404}, \num{0.412}]$ \cite{Csoka-2016-IndependentSetsCutsLargeGirthRegular,Harangi-2023-ReplicaBoundsIndependenceRandomRegular} & 0.428 $^{(\text{n})}$ & 0.464
    \\[0.5em]
    \hline \\[-0.5em]
    Min coloring & $\to \infty$ & $\frac{d}{2\log d}$ \cite{FL-1992-IndependenceChromaticRandomRegular} & $\frac{\sqrt{d}}{2}$ \cite{BBKMW-2020-SpectralPlantingColoring} & $\frac{\sqrt{d}}{2}$ \cite{Hoffman-1970-Eigenvalues} \\
    & 7 & $\in \{4, 5, 6\}$ \cite{AM-2004-ColoringRandomRegular} & 3 $^{(\text{n})}$ & 3
    \\[0.5em]
    \hline \\[-0.5em]
    Min dominating set & any & $\Theta(\frac{\log d}{d})$ \cite{alon2010high} & $\frac{1}{d + 1}$ $^{(\text{n})}$ & $\frac{1}{d + 1}$ (folklore)
    \\[0.5em]
    \hline \\[-0.5em]
    Vertex expansion & $q+1$ & $d - 1$ \cite{hoory2006expander} & $\frac{d}{2}$ $^{(\text{n},\text{b})}$ & $\frac{d}{2}$ \cite{Kahale-1995-SpectralBoundExpansion} \\
    ($\epsilon n$ small-set) & & & &
    \\[0.5em]
    \hline \\[-0.5em]
    Edge expansion & $q^2+1$ & $d - 2$ \cite{hoory2006expander} & $d-1-\sqrt{d-1}$ $^{(\text{n},\text{b})}$ & $d-1-\sqrt{d-1}$ \cite{Kahale-1995-SpectralBoundExpansion} \\
    ($\epsilon n$ small-set) & & & &
    \\[0.75em]
    \hline
    \end{tabular}
    \end{center}
    \caption{A summary of our results on applications of Conjecture~\ref{conj:hardness} to specific certification problems over random $d$-regular graphs. All entries refer to the $n \to \infty$ limit, with high probability in this limit, and omit an additive $o_{n \to \infty}(1)$ term, and the $d \to \infty$ entries omit a multiplicative factor of $(1 + o_{d \to \infty}(1))$. The entries concerning expansion omit an additive $o_{\epsilon \to 0}(1)$ term. $q$ denotes a prime power. Results that are \textbf{n}ew in this paper are marked $(\text{n})$, and results about \textbf{b}ipartite graphs are marked $(\text{b})$. Numerical entries in the last column without a citation are obtained by evaluating the relevant Hoffman-type bound for specific $d$ and we do not cite the relevant work again.}
    \label{table:results}
\end{sidewaystable}

\section{Preliminaries}

Let us emphasize an important initial point: the \emph{only} time we will discuss multigraphs is when they appear as base graphs of random lifts.
All other distributions will always be assumed to be conditioned on their output being simple.
(Our claims should still hold without this conditioning, but we apply the conditioning to apply our results to classical problems about random simple $d$-regular graphs.)

\subsection{Notation}

We write $I_n$ and $J_n$ for the $n \times n$ identity and all-ones matrices, respectively, and $\boldsymbol{1}_n$ for the all-ones vector of length $n$.
We omit the subscripts for these when the dimension is clear from context.

\subsection{Lifts and Random Lifts}
\label{sec:prelim:lifts}

We review some background about lifts and their random versions, taking care to make a definition that is compatible with the work \cite{BC-2019-EigenvaluesRandomLifts} that will play an important technical role in our results.

\begin{definition}[$m$-lift]
    Let $H$ be a $d$-regular multigraph on $k$ vertices, and let $M$ be its adjacency matrix. A graph $G = (V, E)$ is an \emph{$m$-lift} of $H$ if $|V| = km$ and there exists a balanced partition $\sigma: V \to [k]$ (i.e., having $|\sigma^{-1}(i)| = m$ for each $i \in [k]$) such that:
    \begin{itemize}
        \item For every $i\in [k]$, $\sigma^{-1}(i)$ induces an $M_{i,i}$-regular graph on $G$.
        \item For every pair of $\{i,j\} \in \binom{[k]}{2}$, $\sigma^{-1}(i) \sqcup \sigma^{-1}(j)$ induces an $M_{i,j}$-regular bipartite graph on $G$, whose bipartition is the one into $\sigma^{-1}(i)$ and $\sigma^{-1}(j)$.
    \end{itemize}
\end{definition}

\begin{definition}[Fiber]
    In the above context, the subset $\sigma^{-1}(i) \subset V(G)$ is called the \emph{fiber} of $i \in V(H)$.
\end{definition}

We have given a definition of random lifts in Definition~\ref{def:random-lift} in the Introduction.
Let us outline how this fits into the framework often used for lifts in the literature (especially mathematics and random matrix theory literature).
It is more common to take $H$ to be a directed graph (in works such as \cite{ALMR-2001-RandomLifts}, choosing an arbitrary direction for each edge is mentioned).
The lifting procedure may then be viewed as associating a permutation in $S_m$ with each directed edge and inserting edges into the graph according to that permutation.
This allows random lifts to be described as operations on random permutation matrices, which is the perspective taken in random matrix theory works on random lifts like \cite{BC-2019-EigenvaluesRandomLifts}.

So far this setup is entirely equivalent to ours which did not mention directed graphs.
A nuance arises when dealing with self-loops, however.
To be more precise, the above framework also specifies a direction-reversing involution of the edges, $e \mapsto e^*$ with $e^{**} = e$.
The permutations $\sigma_e$ associated to the edges should then satisfy $\sigma_{e^*} = \sigma_e^{-1}$.
Again, this is helpful for describing the adjacency matrix: one may place the permutation matrix associated to $\sigma_e$ in the block associated to the directed edge $e$, and this involution-inversion rule will ensure that the resulting matrix is symmetric.

If $e$ is a loop in the sense that its source and terminal vertices are the same, but $e$ and $e^*$ are viewed as two different edges, then this protocol will place a random 2-regular graph on the fiber of this single vertex.
That is not what our Definition~\ref{def:random-lift} prescribed for loops, however.
In fact, this permutation framework allows for two different kinds of loops: when the source and terminal vertices of $e$ are the same, $e$ and $e^*$ can be two different directed edges, as we have just discussed, or \emph{the same} directed edge---the latter choice still satisfies the involution rule.
In this second situation, there is just one permutation $\sigma_e$ involved, and the involution-inversion property requires $\sigma_e = \sigma_e^{-1}$, i.e., that $\sigma$ is an involution (as a permutation).
This means that adding the edges corresponding to $\sigma_e$ places a perfect matching, or a random 1-regular graph, on the fiber of the one source and terminal vertex of $e$.

This is all to say that, in the setup of Definition~\ref{def:random-lift}, we treat all self-loops as the latter type of self-loop.
This choice is permitted by the setup of \cite{BC-2019-EigenvaluesRandomLifts} as well, and therefore Proposition~\ref{prop:eigenvalues-lifts} holds as stated.

\begin{remark}[Random regular graph vs.\ union of random matchings]
    The result of several self-loops in the base graph is a union of several uniformly random perfect matchings in a random lift.
    In general, the law of the union of $d$ random perfect matchings is similar but not identical to that of a random $d$-regular graph drawn from the configuration model (if both are conditioned to be simple).
    However, these models are \emph{contiguous} and thus the kinds of results we are working with transfer between them; see \cite{Janson-1995-RandomRegularGraphsContiguity}.
\end{remark}

\subsection{Noise Models and Precise Conjectures}
\label{sec:prelim:noise-conjectures}

\begin{definition}[Graph distance]
    For $d$-regular graphs $G$, $H$ on the same number of vertices, the distance $\Delta(G, H) = \frac{|E(G) \triangle E(H)|}{2n}$ is a metric on $d$-regular graphs of the same size, where $n = |V(G)| = |V(H)|$. This is equal to the minimum number of edges that need to be changed to transform $G$ into $H$, divided by the size of the vertex set.
\end{definition}

\begin{definition}[Random noise]
    Let $\eps \in (0, 1)$, and $G$ be a $d$-regular graph on $n$ vertices.
    We write $\sS_{\epsilon}^{\rand}G$ for the random graph that is formed by deleting $\lfloor \epsilon n \rfloor$ edges uniformly at random, and then adding back $\lfloor \epsilon n \rfloor$ edges uniformly at random conditional on the final graph being $d$-regular.
    We abuse the notation slightly and write $\sS_{\epsilon}^{\rand}\sG$ where $\sG$ is the law of a random graph for the law of the random graph drawn from $\sG$ to which random noise is then applied.
\end{definition}

\begin{definition}[Bipartite random noise]
    Let $\eps \in (0, 1)$, and $G$ be a bipartite $d$-regular graph on $n$ vertices.
    We write $\sS_{\epsilon}^{\rand, \bi}G$ for the random bipartite graph that is formed by deleting $\lfloor \epsilon n \rfloor$ edges uniformly at random, and then adding back $\lfloor \epsilon n \rfloor$ edges uniformly at random between the bipartition of $G$ conditional on the final graph being $d$-regular.
    We again abuse the notation slightly and write $\sS_{\epsilon}^{\rand, \bi}\sG$ where $\sG$ is the law of a random bipartite graph for the law of the random bipartite graph drawn from $\sG$ to which random noise is then applied.
\end{definition}

\begin{definition}[Respectful random noise]
    Let $\eps \in (0, 1)$.
    Let $H$ be a $d$-regular multigraph with adjacency matrix $M$.
    We write $\widetilde{\sS}_{\epsilon}^{\rand} \sL_m(H)$ for the random graph that is formed by drawing $G \sim \sL_m(H)$, and then changing $\lfloor \epsilon n \rfloor$ edges as in the definition of random noise, conditional on none of added edges going between the fibers of $i, j \in V(H)$ if $M_{ij} = 0$.
    We likewise define the respectful bipartite random noise $\widetilde{\sS}_{\epsilon}^{\rand, \bi} \sL_m(H)$.
\end{definition}

\begin{definition}[Adversarial noise]
    Let $\epsilon \in (0, 1)$. Let $H$ be a $d$-regular multigraph.
    We write $\sS_{\epsilon}^{\adv} \sL_m(H)$ for \emph{any} distribution formed by sampling $G \sim \sL_m(H)$, having an adversary (with access to additional independent random bits) draw a $G^{\prime}$ from a distribution $\sS_{\epsilon}^{\adv} G$ of $d$-regular graphs supported on those $G^{\prime}$ with $\Delta(G, G^{\prime}) \le \epsilon$, and outputting $G^{\prime}$.
    Formally speaking, $\sS_{\epsilon}^{\adv} \sL_m(H)$ is a \emph{set} of probability distributions, and when we write $G \sim \sS_{\epsilon}^{\adv} \sL_m(H)$, we mean that $G$ is drawn from any of these.
    In particular, when we say that a hypothesis testing problem is hard under adversarial noise, we mean that \emph{there exists} a choice of adversary that makes the problem hard.
\end{definition}

\begin{definition}[Respectful adversarial noise]
    In the setting of adversarial noise operators, if $H$ is a $d$-regular multigraph with adjacency matrix $M$, we write $\widetilde{\sS}_{\epsilon}^{\adv} \sL_m(H)$ for the same family of probability distributions subject to the additional constraint of the adversary never adding any edges between the fibers of $i, j \in V(H)$ if $M_{ij} = 0$.
\end{definition}

We will formulate several precise variations of our main conjecture allowing for different noise models.
These conjectures are adapted to different applications.

\begin{conjecture}
    \label{conj:hardness-formal}
    Let $H$ be a $d$-regular multigraph on $k$ vertices and $\epsilon > 0$.
    We propose the following conjectures:
    \begin{enumerate}
    \item If $H$ is Ramanujan, then there is no polynomial-time algorithm that achieves strong detection between $\sS_{\epsilon}^{\rand} \sL_m(H)$ and $\sG(n, d)$.
    The same holds with $\sS_{\epsilon}^{\rand}$ replaced by any of $\widetilde{\sS}_{\epsilon}^{\rand}$, $\sS_{\epsilon}^{\adv}$, or $\widetilde{\sS}_{\epsilon}^{\adv}$.
    \item If $H$ is bipartite Ramanujan, then there is no polynomial-time algorithm that achieves strong detection between $\sS_{\epsilon}^{\rand, \bi} \sL_m(H)$ and $\sG((\frac{n}{2}, \frac{n}{2}), d)$.
    The same holds with $\sS_{\epsilon}^{\rand, \bi}$ replaced by $\widetilde{\sS}_{\epsilon}^{\rand, \bi}$, $\sS_{\epsilon}^{\adv}$, or $\widetilde{\sS}_{\epsilon}^{\adv}$.
    \end{enumerate}
\end{conjecture}
\noindent
We note that these conjectures admit an ordering by their strength: adversarial noise models include the corresponding random noise models as special cases of the choice of adversary and thus make weaker claims of computational hardness than any of the random noise models (both unipartite and bipartite), and non-respectful noise models make weaker claims of hardness than respectful noise models.
(At a high level, the more permissive the noise model the more we expect the associated problem to be hard.)

\begin{remark}
    The prior work \cite{BBKMW-2020-SpectralPlantingColoring} instead proposed applying noise by applying a random sequence of \emph{edge switchings} where a pair of edges $\{a, b\}, \{c, d\}$ is replaced by the pair $\{a, c\}, \{b, d\}$. This is a natural and more local procedure, but creates some difficulties in applying respectful versions of the noise operator; for instance, applying ``respectful switchings'' to a lift of a base graph (ones that do not introduce edges between the fibers of $i, j \in V(H)$ when $i$ and $j$ are not adjacent in $H$) will never create a graph that is not again a lift of the same graph $H$ (and thus that does not have eigenvalues exactly equal to those of $H$)  unless the base graph contains a 4-cycle.
\end{remark}

\section{Local Statistics Algorithm: Proof of Theorem~\ref{thm:local-statistics-informal}}

In this section, we will prove the hardness for local statistics (LoSt) algorithms stated in Theorem \ref{thm:local-statistics-informal} as supporting evidence for Conjecture \ref{conj:hardness-formal}.

In the following analysis, $H$ will be a connected $d$-regular multigraph on $k$ vertices (the case of disconnected $H$ is trivial since in this case $H$ has multiple eigenvalues equal to $d$).
For simplicity, we will denote $\mathbb{P}_n \colonequals \mathcal{L}_m(H)$ where $n = km$ and $\mathbb{Q}_n \colonequals \mathcal{G}(n,d)$ or $\mathbb{Q}_n \colonequals \mathcal{G}\left(\left(\frac{n}{2}, \frac{n}{2}\right),d\right)$ depending on whether $H$ is bipartite. Note that $\mathbb{P}_n$ are defined without any noise.
At the end of this section, we will argue that the LoSt algorithm is robust under the noise operators defined in the previous section.
In this section we will abbreviate ``with high probability'' to ``w.h.p.'', which always refers to the limit $n \to \infty$.

\subsection{Non-Backtracking Walks}
Let $A_G$ be the adjacency matrix for a $d$-regular multigraph $G$. A \emph{length-$s$ non-backtracking walk} on $G$ is a sequence $(v_0, v_1, \dots, v_s)$ of vertices that forms a walk on $G$ such that for every $i \in [s-1]$, $v_{i-1} \ne v_{i+1}$, i.e., the walk does not ``backtrack.''

The non-backtracking matrices are the matrices $A^{(s)}_G$ for $s \in \mathbb{N}$, whose $u,v$ entry counts the number of non-backtracking walks on $G$ from vertex $u$ to vertex $v$ of length $s$. For $d$-regular graphs, these matrices satisfy the following simple recurrence:
\begin{align*}
    A^{(0)}_G &= I, \\
    A^{(1)}_G &= A_G, \\
    A^{(2)}_G &= A^2_G - dI, \\
    A^{(s+1)}_G &= A_GA^{(s)}_G - (d-1)A^{(s-1)}_G \text{ for } s \ge 2.
\end{align*}

In particular, $A^{(s)}_G = q_s(A_G)$, for a sequence of monic univariate polynomials $q_s \in \mathbb{R}[z]$ called \emph{the non-backtracking polynomials}, which also form a sequence of orthogonal polynomials with respect to the Kesten-McKay measure:

\begin{definition}[Kesten-McKay measure]
    The \emph{Kesten-McKay measure} with parameter $d \geq 2$, denoted $\mu_{\mathrm{KM}(d)}$, is the probability measure supported on $[-2\sqrt{d - 1}, 2\sqrt{d - 1}]$ with density
    \begin{equation}
        \frac{d}{2\pi} \frac{\sqrt{4(d-1) - x^2}}{d^2 - x^2} \One\{|x| \leq 2\sqrt{d-1}\}.
    \end{equation}
\end{definition}

The following result is useful when dealing with inner products between non-backtracking matrices, whose proof can be found in Lemma 4.5 in \cite{banks2021local}.
\begin{proposition}
    Let $G$ be drawn from $\mathbb{P}$ or $\mathbb{Q}$. For any $s, t = O(1)$, with high probability,
    \begin{equation}
        \langle A_G^{(s)}, A_G^{(t)} \rangle = n \cdot \Ex_{\lambda \sim \mu_{\mathrm{KM}(d)}}[q_s(\lambda) q_t(\lambda)] + O(\log n).
    \end{equation}
\end{proposition}

The proof uses the idea that sparse random graphs do not have many short cycles, and most vertices are far from any short cycle. The precise definition and statement are as follows.

\begin{definition}\label{def:bad}
    Fix constants $L$ and $C$. We call a vertex in $G$ \emph{bad} if it is at most $L$ steps away from a cycle of length at most $C$. A vertex is \emph{good} if it is not bad.
\end{definition}

\begin{proposition}\label{prop:bad-v}
    Let $G$ be drawn from $\mathbb{P}$ or $\mathbb{Q}$. Then, with high probability, there are fewer than $O(\log n)$ bad vertices in $G$.
\end{proposition}

In the analysis of the LoSt SDP, we actually need to work with occurrences of simple paths in the input graph, and thus we also need to define self-avoiding matrices.

A length-$s$ walk $(v_0, v_1, \dots, v_s)$ is self-avoiding if it corresponds to a simple path, i.e., all the vertices $v_i$ are distinct. The self-avoiding matrix $A_G^{\langle s \rangle}$ where $s \in \mathbb{N}$, is a matrix whose $u,v$ entries counts the number of self-avoiding walks on $G$ from vertex $u$ to vertex $v$ of length $s$. Clearly, $A_G^{\langle s \rangle}$ are different from $A_G^{(s)}$. Yet, as an easy corollary of Proposition \ref{prop:bad-v} above, these matrices are not that different, and for every fixed $d$ and $s$, w.h.p.~for $G$ drawn from $\mathbb{P}$ or $\mathbb{Q}$ they satisfy
\begin{equation}
    \| A_G^{\langle s \rangle} - A_G^{(s)}\|_{F}^2 \le O(\log n).
\end{equation}

\subsection{Path Statistics SDP}
We first study a simplified version of the Local Statistics SDP, which will be useful for analyzing the full Local Statistics SDP later on.

Recall the definitions in Section \ref{sec:intro:local-statistics} and Section \ref{sec:prelim:lifts}. Let $(x,G)$ be a random lift drawn from $\mathcal{L}_m(H)$ of $H$, where the set of variables $x_{u,i}, u\in [n], i\in [k]$ encodes the underlying partition or labelling $\sigma: [n] \to [k]$ associated to the fibers of the lift. We may alternatively regard $x$ as a collection of $k$ vectors $x_1, \dots, x_k \in \{0,1\}^n$, where $x_i$ is the indicator vector of $\sigma^{-1}(i)$. We define the partition matrix\footnote{We remark here that this definition of partition matrix is slightly different from that in \cite{BBKMW-2020-SpectralPlantingColoring}, and is used to simplify some expressions involved in the computation.} for the planted labelling as
\begin{align*}
    P \colonequals \sum_{i\in [k]} x_ix_i^\top.
\end{align*}

Recall $\sigma: [n] \to [k]$ is a balanced partition: every vertex belongs to one label class, and every label class has size $\frac{n}{k}$. The following facts are easy to see as a consequence.
\begin{fact}\label{fact:partition-matrix-1}
    The partition matrix $P$ satisfies the following properties:
    \begin{align*}
        P &\succeq \frac{1}{k}J, \\
        P_{u, u} &= 1 \text{ for all } u \in [n], \\
        \langle P, J \rangle &= \frac{n^2}{k}.
    \end{align*}
\end{fact}

If moreover the base graph $H$ is bipartite, the lifted graph $G$ is again bipartite, and $\sigma: [n] \to [k]$ naturally induces a balanced bipartition of the vertices in $G$ from the bipartition of $H$. On the other hand, if the lifted graph $G$ is connected, there is a unique balanced bipartition of its vertex set so that $G$ is a bipartite graph on the bipartition.\footnote{To see this, note that if $G$ is connected then it has a spanning tree, so any bipartition of $G$ is also a bipartition of this tree. A simple inductive argument shows that any tree has a unique bipartition.} In this case, let us assume we have arranged the vertices according to the underlying bipartition, such that all the edges of $G$ are between vertices $\{1, 2, \dots, \frac{n}{2}\}$ and $\{\frac{n}{2} +1, \frac{n}{2} + 2, \dots, n\}$. This leads to another property of the partition matrix $P$.

\begin{fact}\label{fact:partition-matrix-2}
    If $H$ is a connected bipartite graph, arrange the vertices of $G$ so that all the edges of $G$ are between vertices $\{1, 2, \dots, \frac{n}{2}\}$ and $\{\frac{n}{2} +1, \frac{n}{2} + 2, \dots, n\}$. Then, the partition matrix $P$ satisfies the following properties:
    \begin{align*}
        P &\succeq \frac{1}{k} \begin{bmatrix}
        J_{n/2} & -J_{n/2}\\
        -J_{n/2} & J_{n/2}
    \end{bmatrix},\\
    P_{u,v} &= P_{v,u} = 0 \text{ for all } u \in \left\{1, 2, \dots, \frac{n}{2}\right\} \text{ and } v\in \left\{\frac{n}{2} +1, \frac{n}{2} + 2, \dots, n\right\}.
    \end{align*}
\end{fact}

In the Path Statistics SDP, we will focus on the \emph{path statistics} given by the polynomials that count the number of non-backtracking walks on $G$ of a certain length with the same labels on the two endpoints. Note that this is exactly given by $\langle P, A^{(s)}_G\rangle = \langle P, q_s(A_G) \rangle$, where $A^{(s)}_G$ is the $s$th non-backtracking matrix of $G$.

Let $M$ be the adjacency matrix of $H$, and let $M = \sum_{i=1}^k \lambda_i v_iv_i^\top$ be its spectral decomposition. The following states that the path statistics in the planted model are strongly concentrated, which follows from Lemmas~\ref{lem:variance-partially-labelled-graph}, and~\ref{lem:path-M} which we will show later, and that, in the notation of that section, \[\langle P, A_G^{(s)} \rangle = \langle P, A_G^{\langle s\rangle} \rangle + O(\log n) = \sum_{i\in [k]} p_{(P_s, \{0,s\}, \{i,i\})}(x, G) + O(\log n).\]

\begin{lemma}\label{lem:concentration-PS}
    Suppose $G \sim \mathcal{L}_m(H)$, and let $P$ be the partition matrix associated with $G$. For every $s \in \mathbb{N}$ and increasing, nonnegative function $\Delta(n)$,
    \begin{align*}
        \mathbb{P}\left( \left|\langle P, A_G^{(s)}\rangle - \frac{n}{k}  \sum_{i=1}^k q_s(\lambda_i)\right| > \Delta(n) \right) \leq O\left( \frac{n}{\Delta(n)^2}\right),
    \end{align*}
    where $\lambda_1, \dots, \lambda_k$ are the eigenvalues of the adjacency matrix $M$ of $H$.
\end{lemma}

Now we are ready ready to define the Path Statistics SDP.
Given an input graph $G_0$, the Path Statistics SDP attempts to find a ``pseudo-partition matrix" $\Tilde{P}$ satisfying Fact \ref{fact:partition-matrix-1} and Fact \ref{fact:partition-matrix-2} exactly, and Lemma~\ref{lem:concentration-PS} approximately up to some error.

\begin{definition}
    Suppose $H$ is not bipartite. The \emph{level-$D$ Path Statistics Algorithm with error tolerance $\delta > 0$} for distinguishing between $\mathcal{L}_m(H)$ and $\mathcal{G}(n, d)$ takes as input a $d$-regular graph $G_0$ and solves the following feasibility SDP: find $\Tilde{P}$ such that
    \begin{enumerate}
        \item $\Tilde{P}_{u,u} = 1$ for every $u \in V(G_0)$,
        \item $\langle \Tilde{P}, J\rangle = \frac{n^2}{k}$,
        \item $\langle \Tilde{P}, A_{G_0}^{(s)} \rangle \in \left(\frac{1}{k}\sum_{i=1}^k q_s(\lambda_i)\right) n + [-\delta n, \delta n] \text{ for all } 0 \le s \le D$, and
        \item $\Tilde{P} \succeq \frac{1}{k}J$.
    \end{enumerate}
\end{definition}

\begin{definition}
    Suppose $H$ is bipartite. Let $G_0$ be $d$-regular, connected and bipartite. Arrange its vertices so that all the edges are between $\{1, 2, \dots, \frac{n}{2}\}$ and $\{\frac{n}{2} +1, \frac{n}{2} + 2, \dots, n\}$.

    The \emph{level-$D$ Path Statistics Algorithm with error tolerance $\delta > 0$} for distinguishing between $\mathcal{L}_m(H)$ and $\mathcal{G}\left(\left(\frac{n}{2}, \frac{n}{2}\right), d\right)$ takes as input a bipartite $d$-regular graph $G_0$ and solves the following feasibility SDP: find $\Tilde{P}$ such that
    \begin{enumerate}
        \item $\Tilde{P}_{u,u} = 1$ for all $u \in V(G_0)$,
        \item $\langle \Tilde{P}, J\rangle = \frac{n^2}{k}$,
        \item $\Tilde{P}_{u,v} = \Tilde{P}_{v,u} = 0 \text{ for all } u \in \{1, 2, \dots, \frac{n}{2}\}$ and $v\in \{\frac{n}{2} +1, \frac{n}{2} + 2, \dots, n\}$,
        \item $\langle \Tilde{P}, A_{G_0}^{(s)} \rangle \in \left(\frac{1}{k}\sum_{i=1}^k q_s(\lambda_i)\right) n + [-\delta n, \delta n] \text{ for all } 0 \le s\le D$,
        \item $\Tilde{P} \succeq \frac{1}{k}J_n$ and $\Tilde{P} \succeq \frac{1}{k} \begin{bmatrix}
            J_{n/2} & -J_{n/2}\\
            -J_{n/2} & J_{n/2}
        \end{bmatrix}.$
    \end{enumerate}
\end{definition}

\begin{remark}
    \label{rem:bipartite-connected}
    In the bipartite version, we assume that the input $G_0$ is connected and bipartite. For a graph drawn from $\mathcal{L}_m(H)$, it is always bipartite, and with high probability it is connected if $H$ is connected (see Theorem 1 of \cite{ALMR-2001-RandomLifts}). Similarly, for a graph drawn from $\mathcal{G}\left(\left(\frac{n}{2}, \frac{n}{2}\right), d\right)$, it is always bipartite, and with high probability it is connected (as follows from Proposition~\ref{prop:eigenvalues-random-regular}, since a disconnected graph has trivial eigenvalue $d$ with multiplicity greater than 1).
\end{remark}

Now let us state versions of our main theorem for the Path Statistics SDP (rather than the Local Statistics SDP, which we will proceed to afterwards).
\begin{definition}
    The \emph{spectral radius} of a connected $d$-regular graph $G$ is $\rho(G) = \max_{i: |\lambda_i| < d} |\lambda_i|$.
\end{definition}

\begin{theorem}\label{thm:PS-main}
    Suppose $H$ is $d$-regular, connected and non-bipartite. If $\rho(H) > 2\sqrt{d-1}$, then there exist $D$ and error tolerance $\delta > 0$ at which the level-$D$ Path Statistics SDP can w.h.p.~distinguish $\mathcal{L}_m(H)$ from $\mathcal{G}(n,d)$. If $\rho(H) \le 2\sqrt{d-1}$, then no such $D$ and $\delta$ exist.

    Suppose $H$ is $d$-regular, connected and bipartite. If $\rho(H) > 2\sqrt{d-1}$, then there exist $D$ and error tolerance $\delta > 0$ at which the level-$D$ Path Statistics SDP can w.h.p.~distinguish $\mathcal{L}_m(H)$ and $\mathcal{G}\left(\left(\frac{n}{2}, \frac{n}{2}\right), d\right)$. If $\rho(H) \le 2\sqrt{d-1}$, then no such $D$ and $\delta$ exist.
\end{theorem}

\begin{remark}
    In Theorem \ref{thm:PS-main}, we assume the base graph $H$ is connected. If $H$ is not connected, then any lift of $H$ is also disconnected, which makes it very simple to distinguish $\mathcal{L}_m(H)$ from $\mathcal{G}(n,d)$ or $\mathcal{G}\left(\left(\frac{n}{2}, \frac{n}{2}\right), d\right)$.
\end{remark}

\begin{proof}[Proof of Theorem \ref{thm:PS-main}]
    Recall that we use $\mathbb{P}$ to denote the ``planted" distribution of random lifts of $H$, and $\mathbb{Q}$ to denote the ``null" distribution of random $d$-regular graphs/random bipartite $d$-regular graphs. Let the spectrum of $H$ be the multiset $\{\{\lambda_1, \dots, \lambda_k\}\}$.

    By Fact \ref{fact:partition-matrix-1}, Fact \ref{fact:partition-matrix-2}, and Lemma~\ref{lem:concentration-PS}, the Path Statistics SDP is w.h.p.~feasible on $G$ drawn from $\mathbb{P}$. We will show the Path Statistics SDP w.h.p.\ is:
    \begin{itemize}
        \item infeasible on $G$ drawn from $\mathbb{Q}$ if $\rho(H) > 2\sqrt{d-1}$,
        \item feasible on $G$ drawn from $\mathbb{Q}$ if $\rho(H) \le 2\sqrt{d-1}$.
    \end{itemize} Since the analysis for the non-bipartite case and the bipartite case is almost identical, we will first consider the case of non-bipartite $H$ and in the end we will explain how to adapt the analysis to bipartite $H$.

    \paragraph{(1a) $H$ is not bipartite and $\rho(H) > 2\sqrt{d-1}$:}
    Let $\lambda := \rho(H)$. We will construct a degree~$D$ polynomial $f$ that is strictly positive on $[-2\sqrt{d-1}, 2\sqrt{d-1}]$ and satisfies $\sum_{i: |\lambda_i| < d} f(\lambda_i) < 0$. We may set $f$ to be a transformed Chebyshev polynomial of the first kind. Let us recall some properties of these polynomials, $T_s(x)$ for $s \in \NN$. We have $\deg(T_s) = s$, $|T_s(x)| \leq 1$ for $|x| \leq 1$, and, for every $\varepsilon > 0$ and large constant $B > 0$, there exists $S = S(B, \varepsilon)$ such that for every $T_s(x)$ of even degree $s \ge S$, $T_s(x) \ge B$ for $|x| \ge 1 + \varepsilon$.
    The last fact follows from the following identity (see, e.g., 18.5.1 of \cite{Olver-2010-NISTHandbook}): when $|x| > 1$, then
    \begin{equation}
        T_s(x) = \frac{1}{2}\left(\left(x - \sqrt{x^2 - 1}\right)^s + \left(x + \sqrt{x^2 - 1}\right)^s\right)
    \end{equation}
    In particular, for $s$ even, we have
    \begin{equation}
        T_s(x) \geq \frac{1}{2}\left(x + \sqrt{x^2 - 1}\right)^s \geq \frac{1}{2}x^s.
    \end{equation}

    Therefore, to produce the polynomial $f$ we wanted above, we can set $f(x) = -T_S(2\sqrt{d-1} \cdot x) + 2$, where $S$ is chosen to be an even degree large enough such that $\sum_{i: |\lambda_i| < d} f(\lambda_i) < 0$. Clearly, this construction satisfies $f$ being strictly positive on $[-2\sqrt{d-1}, 2\sqrt{d-1}]$.

    By Proposition \ref{prop:eigenvalues-random-regular}, when $G$ is drawn from $\mathbb{Q}$, w.h.p.\ $A_G - d J/n$ has eigenvalues in $[-2\sqrt{d-1} - o(1), 2\sqrt{d-1} + o(1)]$. Thus, $f(A_G - dJ/n) = f(A_G) - f(d)J/n \succeq 0$ by the assumption that $f$ is strictly positive on $[-2\sqrt{d-1}, 2\sqrt{d-1}]$, and thus nonnegative on $[-2\sqrt{d-1} - o(1), 2\sqrt{d-1} + o(1)]$ for sufficiently large $n$.

    Now, using this $f$, we show that the Local Statistics SDP is w.h.p.~infeasible for $G$ drawn from $\mathbb{Q}$. Suppose for the sake of contradiction that $\Tilde{P}$ is a feasible solution for the degree-$D$ Path Statistics SDP on input $G$ drawn from $\mathbb{Q}$, where $D \geq \deg(f)$. Recall that the non-backtracking polynomials $q_s$ are mutually orthogonal with respect to the inner product $\langle p, q\rangle_{\mathrm{KM}(d)} \colonequals \mathbb{E}_{\lambda \sim \mathrm{KM}(d)}[p(\lambda)q(\lambda)]$. Then, by rewriting \[ f = \sum_{s = 0}^D \frac{\langle f, q_s\rangle_{\mathrm{KM}(d)}}{\|q_s\|^2_{\mathrm{KM}(d)}}q_s \] in the basis of non-backtracking polynomials and applying the constraints $\Tilde{P}$ satisfies in the Path Statistics SDP, we have
    \begin{align*}
        0 &\le \left\langle \Tilde{P}, f(A_G) - \frac{f(d)}{n} J \right\rangle\\
        &\le \sum_{s = 0}^D \frac{\langle f, q_s\rangle_{\mathrm{KM}(d)}}{\|q_s\|_{\mathrm{KM}(d)}^2} \langle \Tilde{P}, q_s(A_G) \rangle - f(d) \frac{n}{k} + \delta |f(d)| n\\
        &\le \sum_{s=0}^S \frac{\langle f, q_s\rangle_{\mathrm{KM}(d)}}{\|q_s\|_{\mathrm{KM}(d)}^2} \left[\left(\frac{1}{k}\sum_{i=1}^k q_s(\lambda_i)\right) n + \delta n\right] - f(d) \frac{n}{k} + \delta |f(d)| n\\
        &= \left[\frac{1}{k} \sum_{i: |\lambda_i| < d} f(\lambda_i) + \delta \|f\|_{\mathrm{KM}(d)}^2  + \delta |f(d)|\right] n\\
        &< 0,
    \end{align*}
    when $\delta > 0$ is set to be small enough, since $\sum_{i: |\lambda_i| < d} f(\lambda_i) < 0$. This gives a contradiction and completes the proof.

    \paragraph{(1b) $H$ is not bipartite and $\rho(H) \le 2\sqrt{d-1}$:}
    Now we turn to the case when $\rho(H) \le 2\sqrt{d-1}$, and aim to show that the Path Statistics SDP is feasible on $G$ drawn from $\mathbb{Q}$. We set $\Tilde{P} \colonequals g(A_G) + J/k$ for an appropriately chosen polynomial $g$. We would like to construct a polynomial $g$ such that $g$ is strictly positive in $[-2\sqrt{d-1}, 2\sqrt{d-1}]$ and $\langle g, q_s\rangle_{\mathrm{KM}(d)} \in \frac{1}{k}\sum_{i: |\lambda_i| \le 2\sqrt{d-1}} q_s(\lambda_i) + [-\delta/4 , \delta/4 ]$ for $s = 0, 1, \dots, D$.

    By Theorem 5.6 in \cite{BBKMW-2020-SpectralPlantingColoring}, for any $\lambda \in [-2\sqrt{d-1}, 2\sqrt{d-1}]$, constant $D$, and $\delta > 0$, there exists an even\footnote{An even polynomial $g$ has nonzero coefficients only on even degrees, and thus is an even function. While this property is not used here, it is needed in part \textbf{(2)}.} polynomial $g_\lambda$ such that
    \begin{align*}
        \langle g_\lambda, q_s\rangle_{\mathrm{KM}(d)} &\in q_s(\lambda) + [-\delta/4 , \delta/4 ] \text{ for } s = 1, \dots, D,\\
        \langle g_\lambda, q_0\rangle_{\mathrm{KM}(d)} &= 1.
    \end{align*}
    Moreover, for any fixed $\delta > 0$, these polynomials have constant degrees that depend uniformly on $D$, satisfy $g_\lambda \ge 0$, and in particulalr $g_\lambda(d) \ge 0$. Thus, we may simply set $g = \frac{1}{k} \sum_{i:|\lambda_i| \le 2\sqrt{d-1}} g_{\lambda_i}$.

    By Lemma 5.16 in \cite{BBKMW-2020-SpectralPlantingColoring} and Lemma 4.5 in \cite{banks2021local}, when $G$ is drawn from $\mathbb{Q}$, w.h.p.
    \begin{align*}
        \langle g(A_G) , A_G^{(s)}\rangle = (n + O(\log n)) \langle g, q_s\rangle_{\mathrm{KM}(d)}.
    \end{align*}
    Thus, for every $0 \le s \le D$,
    \begin{align*}
        \left\langle g(A_G) + \frac{1}{k} J , A_G^{(s)}\right\rangle &= (n + O(\log n)) \langle g, q_s\rangle_{\mathrm{KM}(d)} + \frac{1}{k}\langle J, q_s(d) J/n\rangle\\
        &= \left(\frac{1}{k} \sum_{i: |\lambda_i| < d} q_s(\lambda_i)\right)n + \frac{1}{k} q_s(d)n + [-\delta n/4, \delta n/4] + o(n)
    \end{align*}
    So, $g(A_G) + J/k$ would satisfy the moment constraints w.h.p. We still need to verify the inner product constraint against $J$ and the diagonal constraints. A simple calculation verifies
    \begin{align*}
        \left\langle  g(A_G) + \frac{1}{k} J, J \right\rangle = \frac{n^2}{k} + g(d) n.
    \end{align*}
    For the diagonal entries,
    \begin{align*}
        g(A_G)_{u,u} &= \sum_{s = 0}^D \frac{\langle g, q_s\rangle_{\mathrm{KM}(d)}}{\|q_s\|_{\mathrm{KM}(d)}^2} (q_s(A_G))_{u,u}\\
        &= \sum_{s = 0}^D \frac{\langle g, q_s\rangle_{\mathrm{KM}(d)}}{\|q_s\|_{\mathrm{KM}(d)}^2} (A_G^{(s)})_{u,u} \\
        &= \frac{\frac{1}{k} \sum_{i: |\lambda_i| \le 2\sqrt{d-1}} 1 }{\|q_0\|_{\mathrm{KM}(d)}^2} (A_G^{(0)})_{u,u} \\
        &= \frac{k-1}{k},
    \end{align*}
    for all vertices $u$ that are not bad vertices as defined in Definition \ref{def:bad}, since for a good vertex $u$, $(A_G^{(0)})_{u,u} = 1$ and $(A_G^{(s)})_{u,u} = 0$ for $s \ge 1$. Thus, the diagonal constraints and inner product constraint against $J$ are not exactly satisfied. Let  \[ \Tilde{Y} = \frac{\frac{k-1}{k}}{\frac{k-1}{k} - \frac{g(d)}{n}}\left(g(A_G) - \frac{g(d)}{n}J\right). \] Note that this makes sure that, for all good vertices $u$, \[ \Tilde{Y}_{u,u} = \frac{\frac{k-1}{k}}{\frac{k-1}{k} - \frac{g(d)}{n}}\left(\frac{k-1}{k} - \frac{g(d)}{n}\right) = \frac{k-1}{k}. \] Moreover, $\Tilde{Y}$ satisfies

    \begin{align*}
        \langle \Tilde{Y}, A_G^{(s)} \rangle &= \frac{\frac{k-1}{k}}{\frac{k-1}{k} - \frac{g(d)}{n}} \left(\langle g(A_G), A_G^{(s)}\rangle- \frac{g(d)}{n} \langle g(A_G), J\rangle\right)\\
        &\in \frac{\frac{k-1}{k}}{\frac{k-1}{k} - \frac{g(d)}{n}} \left(\left(\frac{1}{k} \sum_{i: |\lambda_i| < d} q_s(\lambda_i)\right)n + [-\delta n/4, \delta n/4] + o(n)- g(d)^2\right)\\
        &\in \left(\frac{1}{k} \sum_{i: |\lambda_i| < d} q_s(\lambda_i)\right)n + [-\delta n/4, \delta n/4] + o(n),\\
        \langle \Tilde{Y}, J\rangle &= \frac{\frac{k-1}{k}}{\frac{k-1}{k} - \frac{g(d)}{n}} \left(\langle g(A_G), J \rangle - \frac{g(d)}{n} \langle J, J\rangle\right)\\
        &= \frac{\frac{k-1}{k}}{\frac{k-1}{k} - \frac{g(d)}{n}} \left(g(d)n - g(d)n\right)\\
        &= 0.
    \end{align*}

    Now, using the same trick as in Proposition 4.8 in \cite{banks2021local}, we may further modify $\Tilde{Y}$ into a desirable $Y$ to fix the diagonal entries of the bad vertices, while $Y$ still approximately satisfies the other constraints. Let the spectral decomposition of $A_G$ be $\sum_{i} \lambda_i v_iv_i^\top$. We note that due to the nonnegativity of $g$,
    \begin{align*}
        g(A_G) - \frac{g(d)}{n}J &= g(d) \frac{J}{n} + \sum_{i: |\lambda_i| < d} g(\lambda_i) v_iv_i^\top - \frac{g(d)}{n}J\\
        &= \sum_{i: |\lambda_i| \le 2 \sqrt{d-1}} g(\lambda_i) v_iv_i^\top
    \end{align*}
    has nonnegative eigenvalues and thus $\Tilde{Y} \succeq 0$.  We write $\Tilde{Y}$ as the Gram matrix of vectors $\alpha_1, \dots, \alpha_n \in \mathbb{R}^n$. Since $\Tilde{Y}_{u,u} = \frac{k-1}{k}$ for good vertices $u$, $\|\alpha_u\|^2 = \frac{k-1}{k}$. For the bad vertices, note that every entry of $g(A_G)$ is at most a constant depending on $g$, and thus $\|\alpha_u\| = O(1)$ for the rest of bad vertices. Recall from Proposition \ref{prop:bad-v} that w.h.p.~there are at most $O(\log n)$ bad vertices. Since $\langle \Tilde{Y}, J\rangle = 0$, we have $\sum_{u} \alpha_u = 0$. Let $\Gamma \subseteq [n]$ be the set of good vertices in $G$. Then,
    \begin{align*}
        \left\|\sum_{u \not\in \Gamma} \alpha_u\right\| = O(\log n)
    \end{align*}

    By removing at most $O(\log n)$ vertices from $\Gamma$, we may form a subset $\Gamma'$ of $\Gamma$, and choose a collection of vectors $\beta_u$ for $u\not\in \Gamma'$ such that
    \begin{align*}
        \|\beta_u\|^2 &= \frac{k-1}{k}\\
        \sum_{u \not\in \Gamma'} \beta_u &= \sum_{u \not\in \Gamma'} \alpha_u.
    \end{align*}
    Consider the new Gram matrix $Y$ formed by $\alpha_u$ for $u \in \Gamma'$ and $\beta_u$ for $u \not\in \Gamma'$. Clearly, all the diagonal entries of $Y$ are equal to $\frac{k-1}{k}$. Moreover, we have
    \begin{align*}
        \langle Y, J\rangle &= \left\|\sum_{u \in \Gamma'} \alpha_u + \sum_{u \not\in \Gamma'} \beta_u\right\|^2\\
        &= \left\|\sum_{u } \alpha_u\right\|^2\\
        &= 0, \\
        \langle Y, A_G^{(s)} \rangle &= \langle \Tilde{Y}, A_G^{(s)} \rangle + \langle Y -\Tilde{Y}, A_G^{(s)} \rangle \\
        &= \langle \Tilde{Y}, A_G^{(s)} \rangle + \left(2\sum_{u\not\in \Gamma', v\in \Gamma'} {A_G^{(s)}}_{u,v} \alpha_u^\top (\alpha_v - \beta_v) + \sum_{u \not\in \Gamma'}  {A_G^{(s)}}_{u,u} (\|\alpha_u\| - \|\beta_u\|)\right)\\
        &= \langle \Tilde{Y}, A_G^{(s)} \rangle + O(\log n),
    \end{align*}
    since w.h.p.~there are at most $O(\log n)$ vertices outside $\Gamma'$, there are at most a constant number of nonzero entries in each row of $A_G^{(s)}$, each of which has magnitude bounded by a constant, and all the vectors $\alpha_u, \beta_u$ involved have norm at most $1$. Finally, setting $\Tilde{P} = Y + J/k$, $\Tilde{P}$ clearly satisfies the PSD constraint as $Y \succeq 0$ is a Gram matrix. We may easily verify that $\Tilde{P}$ is a feasible solution to the level-$D$ Path Statistics SDP with error tolerance $\delta$:
    \begin{align*}
        \Tilde{P}_{u,u} &= Y_{u,u} + \frac{1}{k}\\
        &= 1,\\
        \langle \Tilde{P}, J \rangle &= \langle Y, J \rangle + \frac{n^2}{k}\\
        &= \frac{n^2}{k},\\
        \langle \Tilde{P}, A_G^{(s)} \rangle &= \langle Y, A_G^{(s)} \rangle + \left\langle \frac{J}{k}, A_G^{(s)} \right\rangle\\
        &= \langle \Tilde{Y}, A_G^{(s)} \rangle + \left\langle \frac{J}{k}, A_G^{(s)} \right\rangle + O(\log n)\\
        &\in \left(\frac{1}{k} \sum_{i: |\lambda_i| < d} q_s(\lambda_i)\right)n + \frac{1}{k}q_s(d)n + [-\delta n/4, \delta n/4] + o(n),
    \end{align*}
    completing the proof.

    \paragraph{(2) $H$ is bipartite:}
    Finally, we address how to adapt the previous analysis to the bipartite case.
    When $\rho(H) > 2\sqrt{d-1}$, we may show that the Path Statistics SDP is w.h.p.~infeasible when $G$ is drawn from $\mathbb{Q}$, the uniform distribution over bipartite $d$-regular graphs. In this case, by Proposition~\ref{prop:eigenvalues-random-regular}, w.h.p.~the nontrivial eigenvalues of $G$ fall in $[-2\sqrt{d-1} - o(1), 2\sqrt{d-1} + o(1)]$. Using the same polynomial $f$ constructed in part \textbf{(1a)}, we may show
            \[f\left(A_G - \frac{d}{n}J + \frac{d}{n}\begin{bmatrix}
            J_{n/2} & -J_{n/2}\\
            -J_{n/2} & J_{n/2}\end{bmatrix}\right) = f(A_G) - \frac{1}{n}f(d)J - \frac{1}{n}f(-d) \begin{bmatrix}
            J_{n/2} & -J_{n/2}\\
            -J_{n/2} & J_{n/2}
            \end{bmatrix} \succeq 0,\]
        where the vertices of $G$ are arranged in the order such that $\frac{1}{n}\begin{bmatrix}
            J_{n/2} & -J_{n/2}\\
            -J_{n/2} & J_{n/2}
            \end{bmatrix}$ is the orthogonal projection to the eigenvector corresponding to eigenvalue $-d$. Suppose $\Tilde{P}$ is a feasible solution to the Path Statistics SDP. Similarly to before, checking the inner product
            \[\left\langle \Tilde{P},  f(A_G) - \frac{f(d)}{n}J - \frac{ f(-d)}{n}\begin{bmatrix}
            J_{n/2} & -J_{n/2}\\
            -J_{n/2} & J_{n/2}
            \end{bmatrix} \right\rangle\]
            leads to a contradiction.

    When $\rho(H) \le 2\sqrt{d-1}$, we may show that the Path Statistics SDP is w.h.p.~feasible when $G$ is drawn from $\mathbb{Q}$, the uniform distribution over bipartite $d$-regular graphs. We may again use the same $g = \frac{1}{k} \sum_{i: |\lambda_i| \le 2\sqrt{d-1}} g_{\lambda_i}$ as constructed in \textbf{(1b)}. Note that $g$ is a sum of $k-2$ polynomials $g_{\lambda_i}$ rather than $k-1$ polynomials as in \textbf{(1b)}, because in the bipartite case the Ramanujan graph $H$ has two trivial eigenvalues $\pm d$. Consider the matrix \[g(A_G)- \frac{g(d)}{n}J - \frac{g(-d)}{n}\begin{bmatrix}
        J_{n/2} & -J_{n/2}\\
        -J_{n/2} & J_{n/2}
        \end{bmatrix}, \] which will be close to the final construction. Using the same argument as earlier, one may show that the diagonal entries are equal to $1$ for all good vertices. Importantly for the bipartite situation, since $g$ is a sum of $g_{\lambda_i}$, of all which are even polynomials, $g$ is again an even polynomial. In particular, $g(A_G)$ only contains even powers of $A_G$, and by the bipartiteness of the graph $G$, the entries $(u,v)$ of $g(A_G)$ corresponding to vertices in different parts of $G$ are equal to $0$.

        To show that the Path Statistics SDP is feasible, we again set $\Tilde{P}$ to be a slight modification of \[\frac{\frac{k-1}{k}}{\frac{k-1}{k} - g(d)n - g(-d)n}\left(g(A_G) - \frac{g(d)}{n}J - \frac{g(-d)}{n}\begin{bmatrix}
        J_{n/2} & -J_{n/2}\\
        -J_{n/2} & J_{n/2}
        \end{bmatrix}\right),\]
        following the analysis in \textbf{(1b)} by swapping out some vectors\footnote{In the bipartite situation, one needs additionally to make sure during this modification that all of the vectors corresponding to the vertices in one part of the bipartite graph $G$ are orthogonal to the vectors corresponding to the vertices in the other part. This makes sure the off-diagonal blocks of the final matrix are $0$.} in the Gram matrix representation of it, plus $\frac{1}{k}J + \frac{1}{k}\begin{bmatrix}
        J_{n/2} & -J_{n/2}\\
        -J_{n/2} & J_{n/2}
        \end{bmatrix}$.
\end{proof}

Lastly, before moving on to the Local Statistics Algorithm, let us define the Symmetric Path Statistics SDP parametrized by $\lambda$, which simply replaces all of the nontrivial eigenvalues that appear in the moment constraints of the Path Statistics SDP by $\lambda$.
\begin{definition}[The Symmetric Path Statistics SDP,  Non-Bipartite Version]
    The \emph{level-$D$ Symmetric Path Statistics Algorithm parametrized by $\lambda$ with error tolerance $\delta > 0$} on input a $d$-regular graph $G_0$ is the following feasibility SDP: find $\Tilde{P}$ such that
    \begin{enumerate}
        \item $\Tilde{P}_{u,u} = 1$ for every $u \in V(G_0)$,
        \item $\langle \Tilde{P}, J\rangle = \frac{n^2}{k}$,
        \item $\langle \Tilde{P}, A_{G_0}^{(s)} \rangle \in \left(\frac{k-1}{k}q_s(\lambda) + \frac{1}{k}q_s(d)\right) n + [-\delta n, \delta n] \text{ for all } 0 \le s \le D$,
        \item $\Tilde{P} \succeq \frac{1}{k}J$.
    \end{enumerate}
\end{definition}

\begin{definition}[The Symmetric Path Statistics SDP, Bipartite Version]
    Let $G_0$ be $d$-regular, connected and bipartite. Arrange its vertices so that all the edges are between $\{1, 2, \dots, \frac{n}{2}\}$ and $\{\frac{n}{2} +1, \frac{n}{2} + 2, \dots, n\}$.

    The \emph{level-$D$ Symmetric Path Statistics Algorithm parametrized by $\lambda$ with error tolerance $\delta > 0$} on input $G_0$ is the following feasibility SDP: find $\Tilde{P}$ such that
    \begin{enumerate}
        \item $\Tilde{P}_{u,u} = 1$ for every $u \in V(G_0)$,
        \item $\langle \Tilde{P}, J\rangle = \frac{n^2}{k}$,
        \item $\Tilde{P}_{u,v} = \Tilde{P}_{v,u} = 0 \qquad \forall u \in \{1, 2, \dots, \frac{n}{2}\}, v\in \{\frac{n}{2} +1, \frac{n}{2} + 2, \dots, n\}$,
        \item $\langle \Tilde{P}, A_{G_0}^{(s)} \rangle \in \left(\frac{k-2}{k}q_s(\lambda) + \frac{1}{k}q_s(d) + \frac{1}{k}q_s(-d)\right) n + [-\delta n, \delta n] \text{ for all } 0 \le s \le D$,
        \item $\Tilde{P} \succeq \frac{1}{k}J$ and $ \Tilde{P} \succeq \frac{1}{k} \begin{bmatrix}
            J_{n/2} & -J_{n/2}\\
            -J_{n/2} & J_{n/2}
        \end{bmatrix}.$
    \end{enumerate}
\end{definition}

It is straightforward to infer the following fact about the Symmetric Path Statistics SDP from the proof of Theorem \ref{thm:PS-main} above.

\begin{theorem}\label{thm:SPS-main}
    The following hold:
    \begin{enumerate}
    \item If $\lambda > 2\sqrt{d-1}$, then there exist $D$ and error tolerance $\delta > 0$ at which the level-$D$ (non-bipartite) Symmetric Path Statistics SDP parametrized by $\lambda$ is w.h.p.~infeasible for $G$ drawn from $\mathcal{G}(n,d)$. If $\lambda \le 2\sqrt{d-1}$, then for any $D$ and $\delta > 0$, the level-$D$ (non-bipartite) Symmetric Path Statistics SDP parametrized by $\lambda$ is w.h.p.~feasible for $G$ drawn from $\mathcal{G}(n,d)$.
    \item If $\lambda > 2\sqrt{d-1}$, then there exist $D$ and error tolerance $\delta > 0$ at which the level-$D$ (bipartite) Symmetric Path Statistics SDP parametrized by $\lambda$ is w.h.p.~infeasible for $G$ drawn from $\mathcal{G}\left(\left(\frac{n}{2}, \frac{n}{2}\right),d\right)$. If $\lambda \le 2\sqrt{d-1}$, then for any $D$ and $\delta > 0$, the level-$D$ (bipartite) Symmetric Path Statistics SDP parametrized by $\lambda$ is w.h.p.~feasible for $G$ drawn from $\mathcal{G}\left(\left(\frac{n}{2}, \frac{n}{2}\right),d\right)$.
    \end{enumerate}
\end{theorem}

\subsection{Local Statistics SDP}

Recall that we have already defined the Local Statistics SDP with informal moment constraints in Definition \ref{def:local-stats-informal}.
Now we will make formal the moment constraints with a suitable notion of error tolerance.
Note that one goal in the definition of the Local Statistics SDP is to make sure that, for a graph $G$ drawn from the planted distribution $\mathbb{P}$, w.h.p.~the Local Statistics SDP is feasible.
Let us define a basis for the symmetric polynomials that appear in the moment constraints in the Local Statistics SDP.
For technical reasons, we will treat the bipartite case differently from the non-bipartite case.
We first start with the non-bipartite case, for which we can use existing results from \cite{BBKMW-2020-SpectralPlantingColoring}.

\subsubsection{Basis of Symmetric Polynomials: Non-Bipartite Case}

\begin{definition}
    A partially labelled graph $(\alpha, S, \tau )$ is a triple of a graph $\alpha$, a subset of distinguished vertices $S \subseteq V (\alpha)$, and a labelling $\tau : S \to [k]$. We say that a graph is fully labelled if $S = V (\alpha)$, and in this case write $(\alpha, \tau )$ for short. A homomorhism from $(\alpha, S, \tau )$ into a fully labelled graph $(G, \sigma)$ is a map $\phi : V (\alpha) \to V (G)$ that is a graph homormorphism from $\alpha$ to $G$ and agrees on labels for the subset $S$ of distinguished vertices. An occurrence of $(\alpha, S, \tau )$ in $(G, \sigma)$ is an injective homomorphism.

    For every partially labelled graph $(\alpha, S, \tau)$, there exists an associated polynomial in $\mathbb{R}[x, G]$ that counts the number of occurrences of it in a labelled graph $(G, \sigma)$:
    \begin{align*}
        p_{(\alpha, S, \tau)}(x, G) = \sum_{\phi: V(\alpha) \hookrightarrow V(G)} \prod_{\{u,v\} \in E(\alpha) } G_{\phi(u), \phi(v) } \prod_{i \in S} x_{\phi(i), \tau(i)}.
    \end{align*}
    Note that $\deg_x(p_{(\alpha, S, \tau)}) = |S|$ and $\deg_G(p_{(\alpha, S, \tau)}) = |E(\alpha)|$.
\end{definition}

Let $\mathbb{S}[x,G] \subseteq \mathbb{R}[x,G]$ be the space of polynomials that are fixed under the action of $S_n$, and similarly let $\mathbb{S}[x,G]_{\le D_x, D_G} \subseteq \mathbb{R}[x,G]_{\le D_x, D_G}$ be the space of corresponding polynomials with degree in $x$ and $G$ bounded by $D_x$ and $D_G$ respectively. The following lemma is easy to verify.

\begin{lemma}[\cite{BBKMW-2020-SpectralPlantingColoring}]\label{lem:basis}
    The set of polynomials $\{p_{(\alpha,S,\tau)}: |S| \le D_x, |E(\alpha)| \le D_G\}$ form a basis for $\mathbb{S}[x, G]_{\le D_x, D_G}$.
\end{lemma}

Now that we have a basis for $\mathbb{S}[x,G]$ the space of polynomials that are fixed under the action of $S_n$, the moment constraints in the Local Statistics SDP are equivalent to the moment constraints restricted to the basis elements, i.e.,
\begin{align*}
    \Tilde{\mathbb{E}}[p_{(\alpha, S, \tau)}(x, G_0)] \approx \mathbb{E}_{(x, G) \sim \mathbb{P}}[p_{(\alpha, S, \tau)}(x, G)] \,\, \text{ for all } \,\, (\alpha, S, \tau): |S|\le D_x, |E(\alpha)| \le D_G.
\end{align*}

We will then use the following results proved in \cite{BBKMW-2020-SpectralPlantingColoring} about the expectations of these basis polynomials. Recall we use $M$ to denote the $k \times k$ adjacency matrix of the base graph $H$. Following \cite{BBKMW-2020-SpectralPlantingColoring}, for a partially labelled graph $(\alpha,S,\tau)$ we define
\begin{align*}
    (M)_{\tau}^{(\alpha, S)} = \sum_{ \substack{\hat{\tau}: V(\alpha) \to [k]\\ \hat{\tau}|_{S} = \tau } } \frac{\prod_{v \in V(\alpha)} \prod_{i \in [k]} M_{\hat{\tau}(v), i}(M_{\hat{\tau}(v), i} - 1) \cdots (M_{\hat{\tau}(v), i} - \deg_{\hat{\tau}, i}(v) + 1) }{\prod_{\{u,v\}\in E(\alpha) } M_{\hat{\tau}(u), \hat{\tau}(v) } },
\end{align*}
where $\deg_{\hat{\tau}, i}(v)$ is the number of neighbors of $v$ in $\alpha$ that are mapped to group $i$ under the labelling $\hat{\tau}$. Note that this quantity is multiplicative on disjoint union $(\alpha_1 \sqcup \alpha_2, S_1 \sqcup S_2, \tau_1 \sqcup \tau_2)$. Let $\zeta(\alpha) = |V(\alpha)| - |E(\alpha)|$, and write $\text{cc}(\alpha)$ for the number of connected components of $\alpha$.

\begin{lemma}[\cite{BBKMW-2020-SpectralPlantingColoring}]\label{lem:planted-expectation}
    Let $(x, G)$ be drawn from the configuration model\footnote{See \cite[Section 5.4]{BBKMW-2020-SpectralPlantingColoring} for this caveat of working with the configuration model. The upshot is that high probability events transfer from one model to the other.} corresponding to random lift of the base graph $H$ whose adjacency matrix is $M$. For any partially labelled graph $(\alpha,S,\tau)$,
    \begin{align*}
        \mathbb{E}[p_{(\alpha, S, \tau)}(x,G)] = (n/k)^{\zeta(\alpha)} (M)_{\tau}^{(\alpha,S)} + O(n^{\zeta(\alpha) - 1}).
    \end{align*}
\end{lemma}

\begin{lemma}[\cite{BBKMW-2020-SpectralPlantingColoring}] \label{lem:expected-cycle}
    Let $(\alpha,S,\tau)$ be a partially labelled graph with $O(1)$ edges and at least $1$ cycle. Let $(x, G)$ be drawn from the random lift of the base graph $H$ whose adjacency matrix is $M$. Then, for any function $f(n) > 0$,
    \begin{align*}
        \mathbb{P}[p_{(\alpha,S,\tau)}(x, G) > f(n)] = O\left(\frac{n^{\text{cc}(\alpha) - 1}}{f(n)}\right)
    \end{align*}
\end{lemma}

\begin{lemma}[\cite{BBKMW-2020-SpectralPlantingColoring}] \label{lem:variance-partially-labelled-graph}
    Let $(\alpha, S, \tau) = \bigsqcup_{t \in \text{cc}(H)} (\alpha_t, S_t, \tau_t)$ be a partially labelled forest with $O(1)$ edges. Let $(x,G)$ be drawn from the random lift of the base graph $H$ whose adjacency matrix is $M$. Then, for any function $f(n) > 0$,
    \begin{align*}
        \mathbb{P}\left[\left|p_{(\alpha,S,\tau) }(x,G) - (n/k)^{\text{cc}(\alpha)} \prod_{t\in \text{cc}(\alpha)}(M)_{\tau_t}^{(\alpha_t,S_t)}\right| > f(n) \right] = O\left(\frac{n^{2\text{cc}(\alpha) - 1}}{f(n)^2}\right)
    \end{align*}
\end{lemma}

We moreover need the following simple result about the polynomials corresponding to edgeless partially labelled graphs.

\begin{lemma}\label{lem:planted-edgeless}
    Let $(x, G)$ encode any lift of the base graph $H$ whose adjacency matrix is $M$. For any partially labelled graph $(\alpha,S,\tau)$ that is edgeless, i.e., if $\alpha$ is a collection of isolated vertices,
    \begin{align*}
        p_{(\alpha, S, \tau)}(x,G) = N_{(\alpha, S, \tau)},
    \end{align*}
    where $N_{(\alpha, S, \tau)} \colonequals \left[\prod_{i \in [k]} \frac{n}{k}\left(\frac{n}{k} - 1\right) \cdots \left(\frac{n}{k} -|\tau^{-1}(i)| + 1\right)\right](n - |S|)(n - |S| - 1) \cdots (n - |V(\alpha)| + 1)$.
\end{lemma}

In particular, the statement of Lemma~\ref{lem:planted-edgeless} does not involve the expectation under $\mathbb{P}$, but rather holds for an arbitrary lift of $H$, since these polynomials are only related to the $x$ variables and do not use the $G$ variables. We will later on enforce these hard constraints and will refer to them as \emph{label constraints}.

\subsubsection{Basis of Symmetric Polynomials: Bipartite Case}

Next let us consider the bipartite case. Recall in this case $\mathbb{P} = \mathcal{L}_m(H)$ and $\mathbb{Q} = \mathcal{G}\left(\left(\frac{n}{2}, \frac{n}{2}\right), d\right)$. Note that these two distributions are still invariant under the action of $S_n$, so one may still want to work with the basis we defined in the previous section. However, in this case, we can extract more information from a graph $G$ drawn from $\mathbb{P}$ and $\mathbb{Q}$ since w.h.p.~$G$ is bipartite and connected (see Remark~\ref{rem:bipartite-connected}), and it would be more convenient for the analysis to work with a different basis (together with slightly modified $\mathbb{P}$ and $\mathbb{Q}$). In particular, for a connected $d$-regular bipartite graph $G$, there exists a unique balanced bipartition of the vertex set $V(G)$ into $A \sqcup B$ such that all the edges of $G$ are between $A$ and $B$, and this bipartition can moreover be found in linear time.

Therefore, we will work with modified $\mathbb{P}$ and $\mathbb{Q}$ such that each distribution is conditioned on the event that $G$ is connected, and has all edges with one endpoint in $A \colonequals \{1, \dots, \frac{n}{2}\}$ and the other in $B \colonequals \{\frac{n}{2}+1, \dots, n\}$. The reason we assume that a fixed bipartition is that, once we recover the unique balanced bipartition of $G$, we may reindex its vertex set so that the edges respect this fixed bipartition. It is easy to check that this procedure results in the modified $\mathbb{P}$ and $\mathbb{Q}$ above. Denote the modified distributions as $\Tilde{\mathbb{P}}$ and $\Tilde{\mathbb{Q}}$. If one can distinguish $\mathbb{P}$ and $\mathbb{Q}$ w.h.p., one can also do so for $\Tilde{\mathbb{P}}$ and $\Tilde{\mathbb{Q}}$, and vice-versa.

Now, observe that the distribution of $(x,G)$ drawn from $\Tilde{\mathbb{P}}$ is invariant under the simultaneous action of the wreath product $S_{n/2} \wr \mathbb{Z}_2$: $((\xi_1, \xi_2), a) \in (S_{n/2} \times S_{n/2}) \times \mathbb{Z}_2$, which acts on $x$ and $G$ by permuting the vertices such that $\xi_1$ acts on the vertices in $A = \{1, \dots, \frac{n}{2}\}$, $\xi_2$ acts on the vertices in $B = \{\frac{n}{2} +1, \dots, n\}$, and $a \in \mathbb{Z}_2$ either swaps $A$ and $B$ if $a = 1$ or does nothing if $a = 0$.

By the same argument as before, we only need to focus on the polynomials that are fixed under this action of $S_{n/2} \wr \mathbb{Z}_2$. To this end, we slightly modify the definition of partially labelled graphs in the previous section, and define partially labelled bipartite graphs. We also assume that the vertices of the base graph $H$ are arranged so that all of the edges in $H$ are between $\{1, \dots, \frac{k}{2}\}$ and $\{\frac{k}{2} +1, \dots, k\}$.

\begin{definition}
    A partially labelled bipartite graph $(\alpha = ((V_1(\alpha), V_2(\alpha)), E(\alpha)), S = S_1 \sqcup S_2, \tau )$ is a triple of a bipartite graph $\alpha$, a subset of distinguished vertices $S$ with $S_1 \subseteq V_1 (\alpha)$ and $S_2 \subseteq V_2 (\alpha)$, and a labelling $\tau : S \to [k]$ such that $\tau$ maps vertices in $S_1$ to $\{1, \dots, \frac{k}{2}\}$ and maps vertices in $S_2$ to $\{\frac{k}{2} + 1, \dots, k\}$. We say that such a graph is fully labelled if $S = V (\alpha)$, and in this case write $(\alpha, \tau )$ for short. Let $G = ((V_1, V_2), E)$ be a bipartite graph. We say a map $f: V(\alpha) \to V(G)$ is partition-preserving if $f(V_1(\alpha)) \subseteq V_1$ and $f(V_2(\alpha)) \subseteq V_2$, or $f(V_1(\alpha)) \subseteq V_2$ and $f(V_2(\alpha)) \subseteq V_1$. A homomorhism from $(\alpha, S, \tau )$ into a fully labelled bipartite graph $(G = ((V_1, V_2), E), \sigma)$ is a map $\phi : V (\alpha) \to V (G)$ that is a partition-preserving graph homormorphism from $\alpha$ to $G$, and agrees on labels for the subset $S$ of distinguished vertices. An occurrence of $(\alpha, S, \tau )$ in $(G, \sigma)$ is an injective homomorphism.

    For every partially labelled bipartite graph $(\alpha, S, \tau)$, there exists an associated polynomial in $\mathbb{R}[x, G]$ that counts the number of occurrences of it in a labelled bipartite graph $(G = ((V_1, V_2), E), \sigma)$, where $(V_1, V_2) = (A, B)$ is the fixed partition $\{1, \dots, \frac{n}{2}\} \sqcup \{\frac{n}{2} + 1, \dots, n\}$ described above:
    \begin{align*}
        p_{(\alpha, S, \tau)}(x, G) = \sum_{\substack{\phi: V(\alpha) \hookrightarrow V(G)\\
        \text{partition-preserving}}} \prod_{\{u,v\} \in E(\alpha) } G_{\phi(u), \phi(v) } \prod_{i \in S} x_{\phi(i), \tau(i)}.
    \end{align*}
    Note that $\deg_x(p_{(\alpha, S, \tau)}) = |S|$ and $\deg_G(p_{(\alpha, S, \tau)}) = |E(\alpha)|$.
\end{definition}

Let $\Tilde{\mathbb{S}}[x,G] \subseteq \mathbb{R}[x,G]$ be the space of polynomials that are fixed under the action of $S_{n/2} \wr \mathbb{Z}_2$, and similarly let $\Tilde{\mathbb{S}}[x,G]_{\le D_x, D_G} \subseteq \mathbb{R}[x,G]_{\le D_x, D_G}$ be the space of corresponding polynomials with degree in $x$ and $G$ bounded by $D_x$ and $D_G$ respectively. Similar to Lemma~\ref{lem:basis}, one can show these polynomials form a basis for $\Tilde{\mathbb{S}}[x,G]_{\le D_x, D_G}$.

\begin{lemma}
    The set of polynomials $\{p_{(\alpha,S,\tau)}: \alpha \text{ bipartite}, |S| \le D_x, |E(\alpha)| \le D_G\}$ form a basis for $\Tilde{\mathbb{S}}[x, G]_{\le D_x, D_G}$.
\end{lemma}

Now that we have a basis for $\Tilde{\mathbb{S}}[x,G]$ the space of polynomials that are fixed under the action of $S_{n/2} \wr \mathbb{Z}_2$, the moment constraints in the Local Statistics SDP for distinguishing $\Tilde{\mathbb{P}}$ and $\Tilde{\mathbb{Q} }$ are equivalent to the moment constraints restricted to the basis elements, i.e.,
\begin{align*}
    \Tilde{\mathbb{E}}[p_{(\alpha, S, \tau)}(x, G_0)] \approx \mathbb{E}_{(x, G) \sim \mathbb{P}}[p_{(\alpha, S, \tau)}(x, G)] \,\, \text{ for all } \,\, (\alpha, S, \tau): |S|\le D_x, |E(\alpha)| \le D_G.
\end{align*}

Recall we use $M$ to denote the $k \times k$ adjacency matrix of the base graph $H$, whose vertices are arranged so that all the edges are between $\{1, \dots, \frac{k}{2}\}$ and $\{\frac{k}{2} +1, \dots, k\}$. We say a labelling $\tau: V(\alpha) \to [k]$ of a bipartite graph $\alpha = ((V_1(\alpha), V_2(\alpha)), E(\alpha))$ is \emph{partition-respecting} if $\tau(V_1(\alpha)) \subseteq \{1, \dots, \frac{k}{2}\}$ and $\tau(V_2(\alpha)) \subseteq \{\frac{k}{2} +1, \dots, k\}$, or $\tau(V_2(\alpha)) \subseteq \{1, \dots, \frac{k}{2}\}$ and $\tau(V_1(\alpha)) \subseteq \{\frac{k}{2} +1, \dots, k\}$. For a partially labelled bipartite graph $(\alpha,S,\tau)$ we define
\begin{align*}
    (M)_{\tau}^{(\alpha, S)} = \sum_{ \substack{\hat{\tau}: V(\alpha) \to [k]\\ \hat{\tau}|_{S} = \tau\\
    \text{partition-respecting}} } \frac{\prod_{v \in V(\alpha)} \prod_{i \in [k]} M_{\hat{\tau}(v), i}(M_{\hat{\tau}(v), i} - 1) \cdots (M_{\hat{\tau}(v), i} - \deg_{\hat{\tau}, i}(v) + 1) }{\prod_{\{u,v\}\in E(\alpha) } M_{\hat{\tau}(u), \hat{\tau}(v) } },
\end{align*}
where $\deg_{\hat{\tau}, i}(v)$ is the number of neighbors of $v$ in $\alpha$ that are mapped to group $i$ under the labelling $\hat{\tau}$. By analogous arguments to those in \cite{BBKMW-2020-SpectralPlantingColoring}, we can show the statements in Lemma~\ref{lem:planted-expectation}, Lemma~\ref{lem:expected-cycle}, and Lemma~\ref{lem:variance-partially-labelled-graph} still hold for $\Tilde{\mathbb{P}}$, since we may reuse the calculations from Lemma~5.13 in \cite{BBKMW-2020-SpectralPlantingColoring} and compute the expectation of $p_{(\alpha, S, \tau)}$ under the configuration model corresponding to $\Tilde{\mathbb{P}}$ by only summing over partition-preserving injections and partition-respecting labellings. We omit the proofs for Lemma~\ref{lem:expected-cycle} and Lemma~\ref{lem:variance-partially-labelled-graph} for $\Tilde{\mathbb{P}}$, which are identical to those for the nonbipartite case given in \cite{BBKMW-2020-SpectralPlantingColoring} and are simple applications of Markov's inequality. We give the adapted proof of Lemma~\ref{lem:planted-expectation} for $\Tilde{\mathbb{P}}$ below:
\begin{proof}[Proof of Lemma~\ref{lem:planted-expectation} for $\Tilde{\mathbb{P}}$]
    Instead of working with the distribution $\Tilde{\mathbb{P}}$ over random lifts of $H$ directly, we will work with configuration model corresponding to $\Tilde{\mathbb{P}}$, which has the properties that the probability that the drawn graph is simple is bounded away from $0$ as $n\to \infty$, and that the conditional distribution on the event of the drawn graph being simple is identical to $\Tilde{\mathbb{P}}$. See \cite[Section 5.4, Claim 5.11]{BBKMW-2020-SpectralPlantingColoring} for the definition of configuration model and this simple argument. In particular, if some event holds with high probability for $\Tilde{\mathbb{P}}$, it also holds with high probability for the corresponding configuration model, and vice versa.

    Fix a balanced labelling $\sigma: [n] \to [k]$ such that $\sigma^{-1}(i)$ have the same size and moreoever, $\sigma(\{1, \dots, \frac{n}{2}\}) = \{1, \dots, \frac{k}{2}\}$ and $\sigma(\{\frac{n}{2} + 1, \dots, n\}) = \{\frac{k}{2} + 1, \dots, k\}$. Let $G$ be drawn from the configuration model corresponding to random lifts of $H$ with $\sigma$ being the underlying fixed partition.

    Fix a partition-respecting $\hat{\tau}: V(\alpha) \to [k]$ that is an extension of $\tau: S \to [k]$. If $\phi: V(\alpha) \hookrightarrow [n]$ is a partition-preserving injection that agrees with the labels, i.e., $\hat{\tau}(u) = \sigma(\phi(u))$ for all $u\in V(\alpha)$, then, using standard calculations for the configuration model \cite[Lemma 5.12]{BBKMW-2020-SpectralPlantingColoring}, we may compute the probability that $\phi$ is an occurrence of $(\alpha, S, \tau)$ in $G$ as
    \begin{align*}
        &\mathbb{P}[\phi \text{ is an occurrence of } (\alpha, S, \tau)]\\
        &= \prod_{i \le \frac{k}{2} < j} \left(\frac{(M_{i,j} \cdot \frac{n}{k} - E_{\hat{\tau}(i), \hat{\tau}(j)}(\alpha))!}{(M_{i,j} \cdot \frac{n}{k})!}\prod_{\substack{v \in V(\alpha):\\ \hat{\tau}(v) = i}} \frac{M_{i,j}!}{(M_{i,j} - \deg_{\hat{\tau}, j}(v))!} \right)\\
        &= \left(k/n\right)^{|E(\alpha)|} \frac{\prod_{v \in V(\alpha)} \prod_{i \in [k]} M_{\hat{\tau}(v), i}(M_{\hat{\tau}(v), i} - 1) \cdots (M_{\hat{\tau}(v), i} - \deg_{\hat{\tau}, i}(v) + 1)}{\prod_{\{u,v\}\in E(\alpha)} M_{\hat{\tau}(u), \hat{\tau}(v)}} + O(n^{-|E(\alpha)|-1}).
    \end{align*}
    Furthermore, we notice that $\phi$ is automatically partition-preserving since it agrees with the labels of a partition-respecting extension $\hat{\tau}$ of $\tau$. Moreover, whenever $\hat{\tau}$ is not partition-respecting, the probability above is $0$ for any injection $\phi$ that agrees with the labels of $\hat{\tau}$.

    Finally, for a fixed partition-respecting $\hat{\tau}$, the number of partition-preserving injections $\phi: V(\alpha) \hookrightarrow [n]$ that agrees with the labels is $(n/k)^{|V(\alpha)|} + O(n^{|V(\alpha)| - 1})$. Lastly, we observe that the analysis above does not depend on the specific choice $\sigma$. Thus, we conclude that for $(x, G)$ drawn from configuration model corresponding to $\Tilde{\mathbb{P}}$,
    \begin{align*}
        &\mathbb{E}[p_{(\alpha, S, \tau)}(x,G)]\\
        &= \mathbb{E}_{\sigma} \mathbb{E}[p_{(\alpha, S, \tau)}(x,G) \vert \sigma]\\
        &= \sum_{\substack{\hat{\tau}: V(\alpha) \to [k]\\ \hat{\tau}|_S = \tau \\ \text{partition-respecting} }} (n/k)^{|V(\alpha)| - |E(\alpha)|}\frac{\prod_{v \in V(\alpha)} \prod_{i \in [k]} M_{\hat{\tau}(v), i} \cdots (M_{\hat{\tau}(v), i} - \deg_{\hat{\tau}, i}(v) + 1)}{\prod_{\{u,v\}\in E(\alpha)} M_{\hat{\tau}(u), \hat{\tau}(v)}} \\
        &\hspace{1cm} + O(n^{|V(\alpha)|-|E(\alpha)|-1})\\
        &= (n/k)^{\zeta(\alpha)} M_{\tau}^{(\alpha, S)} + O(n^{\zeta(\alpha) - 1})
    \end{align*}
    as desired.
\end{proof}

Similar to Lemma~\ref{lem:planted-edgeless}, we have the following result about the polynomials corresponding to edgeless partially labelled bipartite graphs.

\begin{lemma}\label{lem:planted-edgeless-bipartite}
    Let $(x, G)$ encode any lift of the bipartite base graph $H$ whose adjacency matrix is $M$. For any partially labelled bipartite graph $(\alpha = ((V_1(\alpha), V_2(\alpha)), E(\alpha)),S = S_1 \sqcup S_2,\tau)$ that is edgeless,
    \begin{align*}
        p_{(\alpha, S, \tau)}(x,G) = N_{(\alpha, S, \tau)},
    \end{align*}
    where
    \begin{align*}
        N_{(\alpha, S, \tau)} &\colonequals \left[\prod_{i \in [k]} \frac{n}{k}\left(\frac{n}{k} - 1\right) \cdots \left(\frac{n}{k} -|\tau^{-1}(i)| + 1\right)\right] \cdot \\
        &\hspace{1cm} \prod_{j=1}^2\left(\frac{n}{2} - |S_j|\right)\left(\frac{n}{2} - |S_j| - 1\right)\cdots\left(\frac{n}{2} - |V_j(\alpha)| + 1\right).
    \end{align*}
\end{lemma}

Last but not the least, we add to the ideal $\mathcal{I}_k$ generated by our set of hard constraints the following additional constraints that are always satisfied by $(x,G)$ encoding the lift of a bipartite $H$, where the vertices of $G$ are arranged so that $G$ is bipartite on $A = \{1, \dots, \frac{n}{2}\}$ and $B = \{\frac{n}{2} +1, \dots, n\}$, and the vertices of $H$ are arranged so that $H$ is bipartite on $\{1, \dots, \frac{k}{2}\}$ and $\{\frac{k}{2} +1, \dots, k\}$:

\begin{align*}
    x_{u,i}x_{v,j} &= 0 \text{ for all } u\in \{1, \dots, \frac{n}{2}\}, v\in \{\frac{n}{2}+1, \dots, n\}, \text{and } i,j\in \{1, \dots, \frac{k}{2}\} \\
    x_{u,i}x_{v,j} &= 0 \text{ for all } u\in \{1, \dots, \frac{n}{2}\}, v\in \{\frac{n}{2}+1, \dots, n\}, \text{and } i,j\in \{\frac{k}{2}+1, \dots, k\},
\end{align*}
which encodes the condition that vertices from different parts of the lifted graph $G$ do not get assigned labels from the same part of $H$.

\subsubsection{Defining Local Statistics SDP}

Now that we have defined the suitable basis of symmetric polynomials to work with for both the non-bipartite case and the bipartite case, let us define the Local Statistics SDP.

For any partially labelled forest $(\alpha,S,\tau)$, we have $\zeta(\alpha) = |V(\alpha)| - |E(\alpha)| = \text{cc}(\alpha)$. Therefore, by Lemma~\ref{lem:planted-expectation} and Lemma~\ref{lem:variance-partially-labelled-graph}, for any partially labelled forest $(\alpha, S, \tau)$, w.h.p.~the value of $p_{(\alpha,S,\tau)}(x,G)$ on input $(x,G)$ drawn from $\mathbb{P}$ is close to $(n/k)^{\zeta(\alpha)} (M)_{\tau}^{(\alpha,S)}$, up to an error of $\delta n^{\text{cc}(\alpha)}$ for any constant $\delta > 0$. For any partially labelled graph $(\alpha,S,\tau)$ with $O(1)$ edges and at least one cycle, $\zeta(\alpha) = |V(\alpha)| - |E(\alpha)| < \text{cc}(\alpha)$, and w.h.p.~the value of $p_{(\alpha,S,\tau)}(x,G)$ on input $(x,G)$ drawn from $\mathbb{P}$ is at most $O(n^{\text{cc}(\alpha) - 1/2})$ by Lemma~\ref{lem:expected-cycle}. Lastly, we note that the polynomials $p_{(\alpha, S, \tau)}$ corresponding to edgeless partially labelled graphs do not involve the $G$ variables and only involve the $x$ variables, and by Lemma~\ref{lem:planted-edgeless} and Lemma~\ref{lem:planted-edgeless-bipartite} they equal $N_{(\alpha, S, \tau)}$. This can be thought of as reflecting that the $x$ variables encode the underlying balanced partition in any lifted graph of $H$. Putting everything together, we give the precise definition of Local Statistics SDP with error tolerance $\delta > 0$ in the moment constraints as follows:

\begin{definition}
    The \emph{degree-$(D_x, D_G)$ Local Statistics Algorithm with error tolerance $\delta$} is the following feasibility SDP: given an input graph $G_0$, find $\Tilde{\mathbb{E}}: \mathbb{R}[x]_{\le D_x} \to \mathbb{R}$ such that
    \begin{enumerate}
        \item (Positivity) $\Tilde{\mathbb{E}}[p(x)^2] \ge 0$ whenever $\deg p(x)^2 \le D_x$.
        \item (Hard Constraints) $\Tilde{\mathbb{E}}[p(x, G_0)] = 0$ for every $p \in \mathcal{I}_k$.
        \item (Label Constraints) For every edgeless $(\alpha, S, \tau)$ with at most $D_x$ distinguished vertices,
        \begin{align*}
            \Tilde{\mathbb{E}}[p_{(\alpha, S, \tau)}(x, G_0)] = N_{(\alpha, S, \tau)}.
        \end{align*}
        \item (Moment Constraints) For every $(\alpha, S, \tau )$ with at most $D_x$ distinguished vertices and at most $D_G$ edges,
        \begin{align*}
            \Tilde{\mathbb{E}}[p_{(\alpha,S,\tau)}(x, G_0)] \in (n/k)^{\zeta(\alpha)}(M)_{\tau}^{(\alpha,S)} + [-\delta n^{\text{cc}(\alpha)}, \delta n^{\text{cc}(\alpha)}].
        \end{align*}
    \end{enumerate}
\end{definition}
We remark that, with the scalings of $n^{\zeta(H)}$ and $n^{\text{cc}(H)}$ in the moment constraints, for any $\delta > 0$, this SDP is w.h.p.~feasible on input $G$ drawn from $\mathbb{P}$ by our argument above:

\begin{lemma}\label{lem:planted-feasible}
    Fix $D_x, D_G$ constant, and $\delta > 0$. With high probability, the degree-$(D_x, D_G)$ Local Statistics Algorithm with error tolerance $\delta$ is feasible on input $G \sim \mathbb{P}$.
\end{lemma}

\subsection{Proof of Local Statistics Hardness}

Now we restate our result of local statistics hardness for distinguishing $\mathbb{P}$ and $\mathbb{Q}$, in the absence of noise $\mathcal{S}_{\varepsilon}$.
\begin{theorem}[Local statistics hardness without noise]
    \label{thm:local-statistics-noiseless}
    Let $H$ be a d-regular multigraph on $k$ vertices. Consider the infinite sequence of $n = km, m\in \mathbb{N}$ for which the relevant distributions are well-defined.
    Suppose $H$ is not bipartite. Then, the following hold:
    \begin{itemize}
        \item If $H$ is Ramanujan, then for any constant $D$ and $\delta > 0$, the degree $(2,D)$ Local Statistics algorithm with error tolerance $\delta$ cannot distinguish $\mathcal{L}_m(H)$ and $\mathcal{G}(n, d)$ for any $\epsilon > 0$.
        \item If $H$ is not Ramanujan, then for sufficiently small noise parameter $\epsilon > 0$, there exists a constant $D$ and $\delta > 0$ such that the degree $(2,D)$ Local Statistics algorithm with error tolerance $\delta$ can distinguish $\mathcal{L}_m(H)$ and $\mathcal{G}(n, d)$.
    \end{itemize}
    Suppose $H$ is bipartite. Then, the following hold:
    \begin{itemize}
        \item If $H$ is Ramanujan, then for any constant $D$ and $\delta > 0$, the degree $(2,D)$ Local Statistics algorithm with error tolerance $\delta$ cannot distinguish $\mathcal{L}_m(H)$ and $\mathcal{G}\left(\left(\frac{n}{2},\frac{n}{2} \right), d\right)$ for any $\epsilon > 0$.
        \item If $H$ is not Ramanujan, then for sufficiently small noise parameter $\epsilon > 0$, there exists a constant $D$ and $\delta > 0$ such that the degree $(2,D)$ Local Statistics algorithm with error tolerance $\delta$ can distinguish $\mathcal{L}_m(H)$ and $\mathcal{G}\left(\left(\frac{n}{2},\frac{n}{2} \right), d\right)$.
    \end{itemize}
\end{theorem}

The following lemma will be useful for proving for proving Theorem \ref{thm:local-statistics-noiseless}, which holds under both the non-bipartite case and the bipartite case.
\begin{lemma}[\cite{BBKMW-2020-SpectralPlantingColoring}]\label{lem:path-M}
    Let $(P_s, \{0,s\}, \{i,j\})$ be a partially labelled path of length $s$, with one endpoint labelled $i$ and the other labelled $j$. Then,
    \begin{align*}
        (M)_{\{i,j\}}^{(P_s, \{0,s\})} = q_s(M)_{i,j}.
    \end{align*}
\end{lemma}

\begin{proof}[Proof of Upper Bounds in Theorem \ref{thm:local-statistics-noiseless}]
    Note that by Lemma~\ref{lem:planted-feasible}, the Local Statistics SDP is w.h.p.~feasible on $G_0$ drawn from $\mathbb{P}$. To show that Local Statistics SDP does distinguish $\mathbb{P}$ and $\mathbb{Q}$, we just need to show that it is w.h.p.~infeasible on $G_0$ drawn from $\mathbb{Q}$.

    \paragraph{(a) $H$ is not bipartite:}
    We first consider the non-bipartite case. We will show that when $G_0$ is drawn from $\mathbb{Q}$, w.h.p.~the Local Statistics is infeasible when $\rho(H) > 2\sqrt{d-1}$. We will do so by a reduction to the Path Statistics SDP.

    Suppose for the sake of contradiction that $\Tilde{\mathbb{E}}$ is a feasible pseudoexpectation to the degree-$(2, D)$ Local Statistics Algorithm with error tolerance $\delta$ on input $G_0$.
    We would like to show this implies feasibility of the Path Statistics SDP.

    Recall that a degree $2$ pseudoexpectation $\Tilde{\mathbb{E}}: \mathbb{R}[x]_{\le 2} \to \mathbb{R}$ can be specified by the \emph{pseudomoment matrix}
    \begin{align*}
        \begin{bmatrix}
            1 & \ell^\top\\
            \ell & Q
        \end{bmatrix} = \begin{bmatrix}
            1 & \ell_1^\top & \cdots & \ell_k^\top\\
            \ell_1 & Q_{1,1} & \cdots & Q_{1,k}\\
            \vdots & \vdots & \ddots & \vdots\\
            \ell_k & Q_{k,1} & \cdots & Q_{k,k}
        \end{bmatrix},
    \end{align*}
    where $(\ell_i)_u = \Tilde{\mathbb{E}}[x_{u,i}]$ for $u\in [n], i\in [k]$, and $(Q_{i,j})_{u,v} = \Tilde{\mathbb{E}}[x_{u,i}x_{v,j}]$ for $u,v\in [n], i,j\in [k]$. Consider $\Tilde{P} = \sum_{i \in [k]} Q_{i,i}$. We will show that $\Tilde{P}$ is a feasible solution to the Path Statistics SDP.

    We verify the diagonal constraints are satisfied:
    \begin{align*}
        \Tilde{P}_{u,u} &= \sum_{i \in [k]} \Tilde{\mathbb{E} }[x_{u,i}^2]\\
        &= \Tilde{\mathbb{E} }\left[ \sum_{i \in [k]} x_{u,i}\right]\\
        &= 1.
    \end{align*}

    For inner product constraint against $J$, let $(S_1, S_1, \{i\})$ and $(S_2, S_2, \{i,i\})$ be partially labelled graphs with $1$ and $2$ isolated vertices respectively such that all vertices receive label $i$. Then,
    \begin{align*}
        \langle \Tilde{P}, J\rangle &= \sum_{i = 1}^k \langle Q_{i,i}, J\rangle\\
        &= \sum_{i=1}^k \sum_{u,v \in [n]} \Tilde{\mathbb{E}}[x_{u,i}x_{v,i}]\\
        &= \sum_{i=1}^k \Tilde{\mathbb{E} }[p_{(S_2, S_2, \{i,i\})}(x, G_0)] + \Tilde{\mathbb{E} }[p_{(S_1, S_1, \{i\})}(x, G_0)]\\
        &= \sum_{i=1}^k N_{(S_2, S_2, \{i,i\})} + \sum_{i=1}^k \sum_{u \in [n]} \Tilde{\mathbb{E} }[x_{u,i}]\\
        &= k\left(\frac{n}{k}\right)\left(\frac{n}{k}-1\right) + n\\
        &= \frac{n^2}{k},
    \end{align*}
    where we use Lemma~\ref{lem:planted-edgeless} and that $p_{(S_1, S_1, \{i\})} = \sum_{u\in [n]} x_{u,i}$ in the third to last line.

    For the moment constraints, let $(P_s, \{0, s\}, \{i,i\})$ be a partially labelled path of length $s$, whose endpoints both receive label $i$, and let $(C_s, \{0\}, \{i\})$ be a partially labelled cycle of length $s$, one of whose vertex receives label $i$. For any $0\le s\le D$, we have
    {\allowdisplaybreaks
    \begin{align*}
        \langle \Tilde{P}, A_{G_0}^{(s)} \rangle &= \sum_{i=1}^k \langle Q_{i,i}, A_{G_0}^{(s)}\rangle\\
        &= \sum_{i=1}^k \left(\langle Q_{i,i}, A_{G_0}^{\langle s\rangle}\rangle + o(n)\right),
        \intertext{where we use that w.h.p.~the self-avoiding matrix $A_G^{\langle s \rangle}$ satisfies $\| A_G^{\langle s \rangle} - A_G^{(s)}\|_{F}^2 \le O(\log n)$ as a corollary of Proposition \ref{prop:bad-v}, and $\Tilde{P} = \sum_{i=1}^k Q_{i,i}$ has its entries bounded by $1$ since it is positive semidefinite with ones on the diagonal,}
        &= \sum_{i=1}^k \sum_{u,v \in [n]} \Tilde{\mathbb{E} }[x_{u,i}x_{v,i}] (A_{G_0}^{\langle s\rangle})_{u,v} + o(n)\\
        &= \sum_{i=1}^k \Tilde{\mathbb{E}}[p_{(P_s, \{0,s\}, \{i,i\})}(x, G_0)] + \sum_{i=1}^k \Tilde{\mathbb{E} }[p_{(C_s, \{0\}, \{i\})}(x, G_0)] + o(n)\\
        &\in \sum_{i=1}^k \left(\frac{n}{k} \left(M\right)_{\{i,i\}}^{(P_s, \{0,s\}) } + [-\delta n, \delta n]\right) + o(n)\\
        &= \sum_{i=1}^k \frac{n}{k} q_s(M)_{i,i} + [-k\delta n, k\delta n] + o(n)\\
        &= \frac{1}{k} \langle q_s(M), I_k\rangle n + [-k\delta n, k\delta n] + o(n)\\
        &= \left(\frac{1}{k} \sum_{i=1}^k q_s(\lambda_i)\right)n + [-k\delta n, k\delta n] + o(n),
    \end{align*}}
    which satisfies the moment constraints with error tolerance $\delta' = (k+1)\delta$.

    Finally, let us verify the PSD constraint. For any $u \in [n]$, we have
    \begin{align*}
        \sum_{i\in [k]}(\ell_i)_u = \sum_{i\in [k]}\Tilde{\mathbb{E}}[x_{u,i}]
        = 1,
    \end{align*}
    so $(\sum_i \ell_i)(\sum_i \ell_i)^\top = J$. Since the pseudomoment matrix is PSD, the following matrix, which is a sum over $k$ principal submatrices of the pseudomoment matrix, is again PSD:
    \begin{align*}
        \begin{bmatrix}
            k & (\sum_{i\in [k]} \ell_i)^\top\\
            \sum_{i\in [k]} \ell_i & \sum_{i\in [k]} Q_{i,i}.
        \end{bmatrix}
    \end{align*}
    By Schur complement's criterion for positive semidefiniteness, we have
    \begin{align*}
        \sum_{i\in [k]}Q_{i,i} - \frac{1}{k}\left(\sum_i \ell_i\right)\left(\sum_i \ell_i\right)^\top
        &= \Tilde{P} - \frac{1}{k} J \succeq 0.
    \end{align*}
    This completes the proof that $\Tilde{P}$ is a feasible solution to the level-$D$ Path Statistics SDP if $\Tilde{\mathbb{E}}$ is a feasible solution to the degree-$(2,D)$ Local Statistics SDP.

    However, by Theorem \ref{thm:PS-main}, w.h.p.~the level-$D$ Path Statistics SDP is infeasible on input drawn from $\mathbb{Q}$ for $D$ large enough. Therefore, we reach a contradiction, and conclude that there exists a constant $D$ and $\delta > 0$ such that the degree-$(2,D)$ Local Statistics Algorithm is w.h.p.~infeasible on $\mathbb{Q}$ when $\rho(M) > 2\sqrt{d-1}$, and thus can be used to distinguish $\mathbb{P}$ and $\mathbb{Q}$.

    \paragraph{(b) $H$ is bipartite:}
    Now we explain how to adapt the previous argument to the bipartite case. We again use the reduction to the Path Statistics SDP.

    We use the same construction from $\textbf{(a)}$. Recall that in the bipartite case we arrange the vertices of $G_0$ into two parts so that all edges are between the vertices $A \colonequals \{1, \dots, \frac{n}{2}\}$ and $B \colonequals \{\frac{n}{2} +1, \dots, n\}$. Consider
    \begin{align*}
        \Tilde{P} = \sum_{i\in [k]}\begin{bmatrix}
        Q_{i,i}^{(1,1)} & Q_{i,i}^{(1,2)}\\
        Q_{i,i}^{(2,1)} & Q_{i,i}^{(2,2)}
    \end{bmatrix},
    \end{align*}
    where the matrices $Q_{i,i}$ are written in the block form $\begin{bmatrix}
        Q_{i,i}^{(1,1)} & Q_{i,i}^{(1,2)}\\
        Q_{i,i}^{(2,1)} & Q_{i,i}^{(2,2)}
    \end{bmatrix}$ partitioned by $A$ and $B$. Clearly, the diagonal constraints $\Tilde{P}_{u,u} = 1$ are still satisfied. Note that the off-diagonal blocks $Q_{i,i}^{(1,2)} = Q_{i,i}^{(2,1)} = 0$, since for any $u\in A, v\in B$, the hard constraints enforce
    \begin{align*}
        (Q_{i,i}^{(1,2)})_{u,v} = \Tilde{\mathbb{E}}[x_{u,i}x_{v,i}]
        = 0.
    \end{align*}
    Thus, we may equivalently express $\Tilde{P}$ as
    \begin{align*}
        \Tilde{P} = \sum_{i\in [k]}\begin{bmatrix}
        Q_{i,i}^{(1,1)} & 0\\
        0 & Q_{i,i}^{(2,2)}
    \end{bmatrix},
    \end{align*}
    and $\Tilde{P}$ satisfies the hard constraints that the off-diagonal blocks are zero.

    Now let us verify the inner product constraints. Let $(S_1, S_1, \{i\})$ and $(S_2, S_2, \{i,i\})$ be partially labelled graphs with isolated vertices as defined before.  We have
    \begin{align*}
        \left\langle \sum_{i\in [k]}\begin{bmatrix}
            Q_{i,i}^{(1,1)} & 0\\
            0 & Q_{i,i}^{(2,2)}
        \end{bmatrix}, J\right\rangle
    &= \sum_{i=1}^k \left(\sum_{u,v \in A}\Tilde{\mathbb{E}}[x_{u,i}x_{v,i}]  + \sum_{u,v\in B}\Tilde{\mathbb{E}}[x_{u,i}x_{v,i}]\right)\\
    &= \sum_{i=1}^k \Tilde{\mathbb{E} }[p_{(S_2, S_2, \{i,i\})}(x, G_0)]  + \Tilde{\mathbb{E} }[p_{(S_1, S_1, \{i\})}(x, G_0)]\\
    &= \sum_{i=1}^k N_{(S_2, S_2, \{i,i\})} + \sum_{i=1}^k \sum_{u \in [n]} \Tilde{\mathbb{E}}[x_{u,i}]\\
    &= k\left(\frac{n}{k}\right)\left(\frac{n}{k}-1\right) + n\\
    &= \frac{n^2}{k},
    \end{align*}
    where we use Lemma~\ref{lem:planted-edgeless-bipartite} and that $p_{(S_1, S_1, \{i\})} = \sum_{u\in [n]} x_{u,i}$ in the third to last line.

    For the moment constraints, let $(P_s, \{0, s\}, \{i,i\})$ be a partially labelled even path of length $s$ that is even, whose endpoints both receive label $i$, and let $(C_s, \{0\}, \{i\})$ be a partially labelled even cycle of length $s$ that is even, one of whose vertex receives label $i$. For even $s$, we again have
    \begin{align*}
        \langle \Tilde{P}, A_{G_0}^{(s)} \rangle &= \sum_{i=1}^k \left\langle \begin{bmatrix}
            Q_{i,i}^{(1,1)} & 0\\
            0 & Q_{i,i}^{(2,2)}
        \end{bmatrix}, A_{G_0}^{(s)}\right\rangle\\
        &= \sum_{i=1}^k \left(\left\langle \begin{bmatrix}
            Q_{i,i}^{(1,1)} & 0\\
            0 & Q_{i,i}^{(2,2)}
        \end{bmatrix}, A_{G_0}^{\langle s\rangle}\right\rangle + o(n)\right),\\
        &= \sum_{i=1}^k \left(\sum_{u,v \in A} \Tilde{\mathbb{E} }[x_{u,i}x_{v,i}] (A_{G_0}^{\langle s\rangle})_{u,v}+ \sum_{u,v\in B} \Tilde{\mathbb{E} }[x_{u,i}x_{v,i}] (A_{G_0}^{\langle s\rangle})_{u,v}\right) + o(n)\\
        &= \sum_{i=1}^k \Tilde{\mathbb{E}}[p_{(P_s, \{0,s\}, \{i,i\})}(x, G_0)] + \sum_{i=1}^k \Tilde{\mathbb{E} }[p_{(C_s, \{0\}, \{i\})}(x, G_0)] + o(n)\\
        &\in \sum_{i=1}^k \left(\frac{n}{k} \left(M\right)_{\{i,i\}}^{(P_s, \{0,s\}) } [-\delta n, \delta n]\right) + o(n)\\
        &= \sum_{i=1}^k \frac{n}{k} q_s(M)_{i,i} + [-k\delta n, k\delta n] + o(n)\\
        &= \frac{1}{k} \langle q_s(M), I_k\rangle n + [-k\delta n, k\delta n]  + o(n)\\
        &= \left(\frac{1}{k} \sum_{i=1}^k q_s(\lambda_i)\right)n + [-k\delta n, k\delta n] + o(n),
    \end{align*}
    which satisfies the moment constraints with error tolerance $\delta' = (k+1)\delta$. Note that for bipartite $H$ and odd $s$, $(\frac{1}{k}\sum_{i=1}^k q_s(\lambda_i))n = 0$, since the spectrum of a bipartite graph is symmetric and $q_s$ is an odd function for odd $s$. We thus do not need to worry about path and cycles of odd lengths in this case, because the matrix $\Tilde{P}$ is block diagonal and the nonzero entries of $A_G^{(s)}$ are all supported on the off-diagonal blocks for odd $s$, so the inner product $\langle \Tilde{P}, A_G^{(s)}\rangle$ is automatically $0$.

    Finally, we verify the PSD constraints. By the same argument as in part \textbf{(a)}, the following matrix is PSD:
    \begin{align*}
        X_1 = \begin{bmatrix}
            k & (\sum_{i\in [k]} \ell_i^{(1)})^\top & (\sum_{i\in [k]} \ell_i^{(2)})^\top\\
            \sum_{i\in [k]} \ell_i^{(1)} & \sum_{i\in [k]} Q_{i,i}^{(1,1)} & \sum_{i\in [k]} Q_{i,i}^{(1,2)}\\
            \sum_{i\in [k]} \ell_i^{(2)} & \sum_{i\in [k]} Q_{i,i}^{(2,1)} & \sum_{i\in [k]} Q_{i,i}^{(2,2)}
        \end{bmatrix} = \begin{bmatrix}
            k & \boldsymbol{1}_{n/2}^\top & \boldsymbol{1}_{n/2}^\top\\
            \boldsymbol{1}_{n/2} & \sum_{i\in [k]} Q_{i,i}^{(1,1)} & 0\\
            \boldsymbol{1}_{n/2} & 0 & \sum_{i\in [k]} Q_{i,i}^{(2,2)}
        \end{bmatrix},
    \end{align*}
    where $\ell_i = \begin{bmatrix}
        \ell_i^{(1)}\\
        \ell_i^{(2)}
    \end{bmatrix}$ is again in the block form partitioned by $A$ and $B$ and we use that for any $u \in [n]$, $(\sum_{i\in [k]} \ell_i)_u = 1$. We also know the following matrix is again PSD since it is obtained from $X_1$ by conjugating with a specific diagonal matrix:
    \begin{align*}
        X_2 &= \begin{bmatrix}
            k & -\boldsymbol{1}_{n/2}^\top & \boldsymbol{1}_{n/2}^\top\\
            -\boldsymbol{1}_{n/2} & \sum_{i\in [k]} Q_{i,i}^{(1,1)} & 0\\
            \boldsymbol{1}_{n/2} & 0 & \sum_{i\in [k]} Q_{i,i}^{(2,2)}
        \end{bmatrix}.
    \end{align*}
    By the Schur complement criterion applied to $X_1$, we get as in part \textbf{(a)} that
    $\Tilde{P} - \frac{1}{k}J \succeq 0$. By Schur complement's criterion applied to $X_2$, we get $\Tilde{P} - \frac{1}{k}\begin{bmatrix}
        J_{n/2} & -J_{n/2}\\
        -J_{n/2} & J_{n/2}
        \end{bmatrix} \succeq 0$. This completes the proof of the PSD constraints.

    So, $\Tilde{P}$ is a feasible solution to the level-$D$ Path Statistics SDP if $\Tilde{\mathbb{E}}$ is a feasible solution to the degree-$(2,D)$ Local Statistics SDP for the bipartite situation. However, by Theorem \ref{thm:PS-main}, w.h.p.~the level-$D$ Path Statistics SDP is infeasible on input drawn from $\mathbb{Q}$ for $D$ large enough. Therefore, we conclude that there exists a constant $D$ and $\delta > 0$ such that the degree-$(2,D)$ Local Statistics Algorithm is w.h.p.~infeasible on $\mathbb{Q}$ when $\rho(M) > 2\sqrt{d-1}$, and thus can be used to distinguish $\mathbb{P}$ and $\mathbb{Q}$.
\end{proof}

We now turn to the proofs of the hardness results in Theorem \ref{thm:local-statistics-noiseless}. Recall that by Lemma~\ref{lem:planted-feasible}, the Local Statistics SDP is w.h.p.~feasible on $G_0$ drawn from $\mathbb{P}$. To show that Local Statistics SDP fails to distinguish $\mathbb{P}$, and $\mathbb{Q}$, we just need to show the Local Statistics SDP is again w.h.p.~feasible on input drawn from $\mathbb{Q}$.

Before we set off to prove the lower bounds, we state the following result that will simplify the analysis of the moment constraints.

\begin{definition}
    Let $(\alpha, S, \tau)$ be a partially labelled forest. We say $(\alpha, S, \tau)$ is \emph{pruned} if all the leaves of the forest belong to the set $S$ of distinguished vertices.
\end{definition}

\begin{lemma}[{\cite[Lemmas 5.19 and 5.23]{BBKMW-2020-SpectralPlantingColoring}}]\label{lem:minimal-graphs}
    Let $G \sim \mathbb{Q}$ and $\Tilde{\mathbb{E}}$ be a degree-$D_x$ pseudoexpectation that depends on $G$. If $\Tilde{\mathbb{E}}$ w.h.p.~satisfies the hard constraints, the positive semidefiniteness constraint, and the moment constraints for all pruned partially labelled forests with at most $D_G$ edges and $D_x$ distinguished vertices, then w.h.p.~it satisfies all the moment constraints for the degree-$(D_x, D_G)$ Local Statistics SDP.
\end{lemma}

\begin{remark}
    We note that although \cite{BBKMW-2020-SpectralPlantingColoring} only prove the statement of Lemma~\ref{lem:minimal-graphs} for the case when $\mathbb{Q} = \mathcal{G}(n, d)$, it is easy to see that the same analysis immediately extends to the bipartite case $\Tilde{\mathbb{Q}}$ associated to $\mathbb{Q} = \mathcal{G}\left(\left(\frac{n}{2}, \frac{n}{2}\right), d\right)$.
\end{remark}

\begin{proof}[Proof of Lower Bounds in Theorem \ref{thm:local-statistics-noiseless}]
    Let $M$ be the $k \times k$ adjacency matrix of $H$. We will show that if $\rho(M) \le 2\sqrt{d-1}$, we may construct a pseudoexpectation $\Tilde{\mathbb{E}}$ for any degree-$(2, D)$ Local Statistics Algorithm with error tolerance $\delta > 0$ that is w.h.p.~feasible on input $G_0$ drawn from $\mathbb{Q}$. Let us first consider the non-bipartite case.

    \paragraph{(1) $H$ is not bipartite:}
    Let $\Tilde{P}^{(\lambda)}$ be a pseudo-partition matrix to the level-$D$ Symmetric Path Statistics SDP parametrized by $\lambda$ with $|\lambda| \le 2\sqrt{d-1}$, whose asymptotic almost sure existence on $G_0 \sim \mathbb{Q}$ is guaranteed by Theorem \ref{thm:SPS-main}. Recall this matrix satisfies:
    \begin{enumerate}
        \item $\Tilde{P}^{(\lambda)}_{u,u} = 1$ for every $u \in V(G_0)$,
        \item $\langle \Tilde{P}^{(\lambda)}, J\rangle = \frac{n^2}{k}$,
        \item $\langle \Tilde{P}^{(\lambda)}, A_G^{(s)} \rangle \in \frac{k-1}{k}q_s(\lambda)n + \frac{1}{k}q_s(d)n + [-\delta n, \delta n]$,
        \item $\Tilde{P}^{(\lambda)} \succeq \frac{1}{k}J$.
    \end{enumerate}

    As before, we will use that a degree-$2$ pseudoexpectation $\Tilde{\mathbb{E}}$ for the Local Statistics SDP may be specified by a PSD degree-$2$ pseudomoment matrix
    \begin{align*}
        \begin{bmatrix}
            1 & \ell^\top\\
            \ell & Q
        \end{bmatrix} = \begin{bmatrix}
            1 & \ell_1^\top & \cdots & \ell_n^\top\\
            \ell_1 & Q_{1,1} & \cdots & Q_{1,n}\\
            \vdots & \vdots & \ddots & \vdots\\
            \ell_n & Q_{n,1} & \cdots & Q_{n,n}
        \end{bmatrix},
    \end{align*}
    where $(\ell_u)_i = \Tilde{\mathbb{E}}[x_{u,i}]$ and $(Q_{u,v})_{i,j} = \Tilde{\mathbb{E}}[x_{u,i}x_{v,j}]$. Note that here we use the other block representation of $Q$, where $Q_{u,v}$ are indexed by vertices $u,v \in [n]$ of $G_0$.

    Let the spectral decomposition of $M$ be $M = \sum_{i=1}^k \lambda_i v_iv_i^\top = d J_k/k + \sum_{i: |\lambda_i| < d} \lambda_i v_i v_i^\top$. Our construction sets $(\ell_u)_i = \frac{1}{k}$ for every $u$ and $i$, and the remaining block to
     \[ Q = \frac{1}{k-1} \sum_{i: |\lambda_i| < d} \Tilde{P}^{(\lambda_i)} \otimes v_iv_i^\top + \frac{1}{k(k-1)}J_n \otimes (J_k - I_k). \]

    First we verify the PSD constraint:
    \begin{align*}
        & Q - \ell \ell^\top\\
        &= \frac{1}{k-1} \sum_{i: |\lambda_i| < d} \Tilde{P}^{(\lambda_i)} \otimes v_iv_i^\top + \frac{1}{k(k-1)}J_n \otimes (J_k - I_k) - \frac{1}{k^2} J_n \otimes J_k\\
        &= \frac{1}{k-1} \sum_{i: |\lambda_i| < d} \Tilde{P}^{(\lambda_i)} \otimes v_iv_i^\top + \frac{1}{k}J_n \otimes \left(\frac{1}{k(k-1)}J_k - \frac{1}{k-1}I_k\right)\\
        &= \frac{1}{k-1} \sum_{i: |\lambda_i| < d} \left(\Tilde{P}^{(\lambda_i)} - \frac{1}{k}J_n\right) \otimes v_iv_i^\top + \frac{1}{k(k-1)} J_n \otimes \left(\sum_{i:|\lambda_i| < d} v_iv_i^\top\right) \\
        &\hspace{1cm} + \frac{1}{k(k-1)}J_n \otimes \left(\frac{1}{k}J_k - I_k\right)\\
        &=  \frac{1}{k-1} \sum_{i: |\lambda_i| < d} \left(\Tilde{P}^{(\lambda_i)} - \frac{1}{k}J_n \right) \otimes v_iv_i^\top + \frac{1}{k(k-1)}J_n \otimes \left(I_k - \frac{1}{k}J_k\right) +\frac{1}{k(k-1)} J_n \otimes \left(\frac{1}{k}J_k - I_k\right)\\
        &= \frac{1}{k-1} \sum_{i: |\lambda_i| < d} \left(\Tilde{P}^{(\lambda_i)} - \frac{1}{k}J_n \right) \otimes v_iv_i^\top \succeq 0,
    \end{align*}
    where we use $\Tilde{P}^{(\lambda_i)} - \frac{1}{k}J_n \succeq 0$, $v_iv_i^\top \succeq 0$, and that tensor product of positive semidefinite matrices is again positive semidefinite. By the Schur complement criterion, we conclude that the pseudomoment matrix is PSD.

    Next we verify the hard constraints:
    \begin{align*}
        \Tilde{\mathbb{E}}[x_{v,i}^2] &= (Q_{v,v})_{i,i}\\
        &= \frac{1}{k-1} \sum_{k: |\lambda_k| < d} \Tilde{P}^{(\lambda_k)}_{v,v} (v_kv_k^\top)_{i,i}\\
        &= \frac{1}{k-1} (I - J_k/k)_{i,i}\\
        &= \frac{1}{k}\\
        &= \Tilde{\mathbb{E}}[x_{v,i}]
    \end{align*}
    and
    \begin{align*}
        \Tilde{\mathbb{E} }[(x_{u,1} + \cdots + x_{u,k}) x_{v,i}] &= \sum_{j = 1}^k \frac{1}{k-1} \left(\sum_{r: |\lambda_r| < d} \Tilde{P}^{(\lambda_r)}_{u,v} \left(v_rv_r^\top\right)_{j,i}  + \frac{1}{k} \One\{j \ne i\}\right)\\
        &= \frac{1}{k} + \frac{1}{k-1}\sum_{r: |\lambda_r| < d} \Tilde{P}^{(\lambda_r)}_{u,v} \sum_{j=1}^k \left(v_rv_r^\top\right)_{j,i}\\
        &= \frac{1}{k} + \frac{1}{k-1}\sum_{r: |\lambda_r| < d} \Tilde{P}^{(\lambda_r)}_{u,v} (v_r)_i\sum_{j=1}^k (v_r)_j\\
        &= \frac{1}{k}\\
        &= \Tilde{\mathbb{E} }[x_{v,i}],
    \end{align*}
    where we use the fact that $\sum_{j} (v_r)_j = \langle v_r, (1, \dots, 1)^\top \rangle = 0$ in the second to last line, since $v_r$ is an eigenvector orthogonal to the constant vector, the eigenvector of the eigenvalue $d$.

    We also verify the label constraints, i.e., the hard moment constraints on the edgeless labelled graphs. We have
    {\allowdisplaybreaks
    \begin{align*}
        \sum_{u \in [n]} \Tilde{\mathbb{E}}[x_{u,i}] &= \frac{n}{k},\\
        \sum_{u, v\in [n]: u\ne v} \Tilde{\mathbb{E}}[x_{u,i}x_{v,j}] &= \sum_{u, v\in [n]: u\ne v} \left(\frac{1}{k-1} \left(\sum_{r: |\lambda_r| < d} \Tilde{P}^{(\lambda_r)}_{u,v} (v_rv_r^\top)_{i,j} + \frac{1}{k} \One\{i \ne j\}\right)\right)\\
        &= \frac{n(n-1)}{k(k-1)}\One\{i \ne j\}+ \frac{1}{k-1} \sum_{r: |\lambda_r| < d} \left( (v_rv_r^\top)_{i,j} \sum_{u, v\in [n]: u\ne v}\Tilde{P}^{(\lambda_r)}_{u,v}\right)\\
        &= \frac{n(n-1)}{k(k-1)} \One\{i \ne j\} + \frac{1}{k-1} \sum_{r: |\lambda_r| < d} (v_rv_r^\top)_{i,j} \left(\langle \Tilde{P}^{(\lambda_r)}, J_n \rangle - \sum_{u \in [n]} \Tilde{P}^{(\lambda_r)}_{u,u}\right)\\
        &= \frac{n(n-1)}{k(k-1)} \One\{i \ne j\} + \frac{1}{k-1} \sum_{r: |\lambda_r| < d} (v_rv_r^\top)_{i,j} \left(\frac{n^2}{k} - n\right)\\
        &= \frac{n(n-1)}{k(k-1)} \One\{i \ne j\} + \frac{1}{k-1} \left(\frac{n^2}{k} - n\right) \left(I_k - \frac{J_k}{k}\right)_{i,j}\\
        &= \begin{cases}
            \frac{n}{k}\left(\frac{n}{k} -1\right) & \quad \text{ if } i = j,\\
            \frac{n^2}{k^2} & \quad \text{ if } i\ne j,
        \end{cases}
    \end{align*}}
    as desired.

    Finally, we verify the moment constraints. By Lemma~\ref{lem:minimal-graphs}, we only need verify the moment constraints on the class of pruned partially labelled forests. So long as these constraints are satisfied, the remaining moment constraints are automatically satisfied with high probability.

    A minimal partially labelled forest with $1$ distinguished vertex is just a single labelled vertex, which we have already dealt with in the label constraints.

    A minimally partially labelled forest with $2$ distinguished vertices is either two isolated labelled vertices, or a path of length $s$ with endpoints labelled by $i, j$. The case of $2$ isolated labelled vertices is again dealt with in the label constraints, so we only need to verify the moment constraints imposed by the paths with labelled endpoints. Recall $A_{G_0}^{\langle s\rangle}$ is the self-avoiding matrix of $G_0$, which is useful for writing down the polynomial associated with partially labelled paths. For a path of length $s$ with two endpoints labelled $i$ and $j$, we verify
    {\allowdisplaybreaks
    \begin{align*}
        &\Tilde{\mathbb{E}}[p_{(P_s, \{0, s\}, \{i,j\})}]\\
        &=\langle Q_{i,j}, A_{G_0}^{\langle s \rangle}\rangle\\
        &= \langle Q_{i,j}, A_{G_0}^{(s)}\rangle + o(n)
        \intertext{where we use that w.h.p.~the self-avoiding matrix $A_{G_0}^{\langle s \rangle}$ satisfies $\| A_{G_0}^{\langle s \rangle} - A_{G_0}^{(s)}\|_{F}^2 \le O(\log n)$ as a corollary of Proposition \ref{prop:bad-v}, and $Q_{i,j}$ has its entries bounded by $\frac{1}{k}$ since $Q$ is PSD with $\frac{1}{k}$ on the diagonal,}
        &= \left\langle \frac{1}{k-1}\left(\sum_{r: |\lambda_r| < d} \Tilde{P}^{(\lambda_r)} \cdot \left(v_rv_r^\top\right)_{i,j} + \frac{1}{k} J_n \cdot (J_k - I_k)_{i,j}\right), A_{G_0}^{(s)} \right\rangle + o(n)\\
        &= \frac{1}{k-1}\sum_{r: |\lambda_r| < d} \left(v_rv_r^\top\right)_{i,j} \langle \Tilde{P}^{(\lambda_r)}, A_G^{(s)} \rangle + \frac{1}{k(k-1)}(J_k - I_k)_{i,j} \langle J_n, A_G^{(s)}\rangle + o(n) \\
        &= \sum_{r: |\lambda_r| < d} \left(v_rv_r^\top\right)_{i,j} \left( \frac{1}{k} q_s(\lambda_r)n + \frac{1}{k(k-1)} q_s(d) n \pm \frac{1}{k-1}\delta n \right) + \frac{1}{k(k-1)} (J_k - I_k)_{i,j} q_s(d) n + o(n)
        \intertext{note that $|v_rv_r^\top|_{i,j} \le 1$ since $v_r$ are unit vectors, and thus the total error term is bounded by $\pm \delta n$,}
        &\in \frac{1}{k}q_s\left(M - d\frac{J_k}{k}\right)_{i,j}n + \frac{1}{k(k-1)}q_s(d)n \cdot \left(I_k - \frac{J_k}{k} + J_k - I_k\right)_{i,j} + [-\delta n, \delta n]  + o(n)\\
        &=  \frac{1}{k}q_s\left(M - d\frac{J_k}{k}\right)_{i,j}n + \frac{1}{k}q_s\left(d \frac{J_k}{k}\right)_{i,j} n + [-\delta n, \delta n] + o(n)\\
        &= \frac{n}{k}q_s(M)_{i,j} + [-\delta n, \delta n] + o(n),
    \end{align*}}
    which satisfies the moment constraints with error tolerance $2\delta$.

    Thus, using the solutions from the level-$D$ Symmetric Path Statistics SDP, we have constructed a degree-$2$ pseudomoment $\Tilde{\mathbb{E}}$ for the degree-$(2, D)$ Local Statistics SDP that is w.h.p.~feasible on input $G_0 \sim \mathbb{Q}$.

    \paragraph{(2) $H$ is bipartite:}
    Now we address the bipartite case. Again let $\Tilde{P}^{(\lambda)}$ be a pseudo-partition matrix to the level-$D$ Bipartite Symmetric Path Statistics SDP parametrized by $\lambda$ with $|\lambda| \le 2\sqrt{d-1}$, whose asymptotic almost sure existence on $G_0 \sim \Tilde{\mathbb{Q}}$ is guaranteed by Theorem \ref{thm:SPS-main}. Recall that this matrix satisfies:
    \begin{enumerate}
        \item $\Tilde{P}^{(\lambda)}_{u,u} = 1$ for every $u \in V(G_0)$,
        \item $\langle \Tilde{P}^{(\lambda)}, J\rangle = \frac{n^2}{k}$,
        \item $\Tilde{P}^{(\lambda)}_{u,v} = \Tilde{P}^{(\lambda)}_{v,u} = 0 \text{ for all } u \in \{1, 2, \dots, \frac{n}{2}\}$ and $v\in \{\frac{n}{2} +1, \frac{n}{2} + 2, \dots, n\}$,
        \item $\langle \Tilde{P}^{(\lambda)}, A_G^{(s)} \rangle \in \frac{k-2}{k}q_s(\lambda)n + \frac{1}{k}q_s(d)n + \frac{1}{k}q_s(-d) + [-\delta n, \delta n]$,
        \item $\Tilde{P}^{(\lambda)} \succeq \frac{1}{k}J$, and $\Tilde{P}^{(\lambda)} \succeq \frac{1}{k} \begin{bmatrix}
            J_{n/2} & -J_{n/2}\\
            -J_{n/2} & J_{n/2}
        \end{bmatrix}.$
    \end{enumerate}
    We further note that Conditions 2, 3, and 5 together imply that \[ \Tilde{P}^{(\lambda)} \succeq \frac{1}{k}J + \frac{1}{k} \begin{bmatrix}
            J_{n/2} & -J_{n/2}\\
            -J_{n/2} & J_{n/2}
        \end{bmatrix}, \] since these conditions enforce that $(1, \dots, 1, 1, \dots,  1)^\top$ and $(1, \dots, 1, -1, \dots, -1)^\top$ are two orthogonal eigenvectors of $\Tilde{P}^{(\lambda)}$ with eigenvalue $\frac{n}{k}$.

    Let the spectral decomposition of $M$, the adjacency matrix of $H$, be \[M = \sum_{i=1}^k \lambda_i v_iv_i^\top = \frac{d}{k} J_k - \frac{d}{k} \begin{bmatrix}
            J_{k/2} & -J_{k/2}\\
            -J_{k/2} & J_{k/2}
        \end{bmatrix} + \sum_{i: |\lambda_i| < d} \lambda_i v_iv_i^\top.\]
        We similarly construct the degree-$2$ pseudomoment matrix that represents a degree-$2$ pseudoexpectation $\Tilde{\mathbb{E}}$
        \begin{align*}
        \begin{bmatrix}
            1 & \ell^\top\\
            \ell & Q
        \end{bmatrix} = \begin{bmatrix}
            1 & \ell_1^\top & \cdots & \ell_n^\top\\
            \ell_1 & Q_{1,1} & \cdots & Q_{1,n}\\
            \vdots & \vdots & \ddots & \vdots\\
            \ell_n & Q_{n,1} & \cdots & Q_{n,n}
        \end{bmatrix},
    \end{align*}
    by setting $(\ell_u)_i = \frac{1}{k}$ for every $u$ and $i$, and
    \begin{align*}
        Q &= \frac{1}{k-2} \Bigg(\sum_{i: |\lambda_i| < d} \Tilde{P}^{(\lambda_i)} \otimes v_iv_i^\top + \frac{1}{k} J_n \otimes \left(\frac{k-1}{k}J_k + \frac{1}{k}\begin{bmatrix}
            J_{k/2} & -J_{k/2}\\
            -J_{k/2} & J_{k/2}
        \end{bmatrix} - I_k\right)\\
        &\hspace{2cm} + \frac{1}{k}\begin{bmatrix}
            J_{n/2} & -J_{n/2}\\
            -J_{n/2} & J_{n/2}
        \end{bmatrix} \otimes \left(\frac{k-1}{k} \begin{bmatrix}
            J_{k/2} & -J_{k/2}\\
            -J_{k/2} & J_{k/2}
        \end{bmatrix} + \frac{1}{k} J_k - I_k\right)\Bigg) .
    \end{align*}

    First we verify the positive semidefiniteness constraint:
    {\allowdisplaybreaks
    \begin{align*}
        &Q - \ell \ell^\top\\
        &= \frac{1}{k-2} \sum_{i: |\lambda_i| < d} \Tilde{P}^{(\lambda_i)} \otimes v_iv_i^\top + \frac{1}{k(k-2)}J_n \otimes \left(\frac{k-1}{k}J_k + \frac{1}{k}\begin{bmatrix}
            J_{k/2} & -J_{k/2}\\
            -J_{k/2} & J_{k/2}
        \end{bmatrix} - I_k\right)\\
        &\quad + \frac{1}{k(k-2)}\begin{bmatrix}
            J_{n/2} & -J_{n/2}\\
            -J_{n/2} & J_{n/2}
        \end{bmatrix} \otimes \left(\frac{k-1}{k} \begin{bmatrix}
            J_{k/2} & -J_{k/2}\\
            -J_{k/2} & J_{k/2}
        \end{bmatrix} + \frac{1}{k} J_k - I_k\right)- \frac{1}{k^2} J_n \otimes J_k\\
        &= \frac{1}{k-2} \sum_{i: |\lambda_i| < d}\Tilde{P}^{(\lambda_i)} \otimes v_iv_i^\top + \frac{1}{k(k-2)}J_n \otimes \left(\frac{k-1}{k}J_k - \frac{k-2}{k} J_k + \frac{1}{k}\begin{bmatrix}
            J_{k/2} & -J_{k/2}\\
            -J_{k/2} & J_{k/2}
        \end{bmatrix} - I_k\right)\\
        &\quad + \frac{1}{k(k-2)}\begin{bmatrix}
            J_{n/2} & -J_{n/2}\\
            -J_{n/2} & J_{n/2}
        \end{bmatrix} \otimes \left(\frac{k-1}{k} \begin{bmatrix}
            J_{k/2} & -J_{k/2}\\
            -J_{k/2} & J_{k/2}
        \end{bmatrix} + \frac{1}{k} J_k - I_k\right)\\
        &= \frac{1}{k-2} \sum_{i: |\lambda_i| < d}\left(\Tilde{P}^{(\lambda_i)} - \frac{1}{k}J_n - \frac{1}{k}\begin{bmatrix}
            J_{n/2} & -J_{n/2}\\
            -J_{n/2} & J_{n/2}
        \end{bmatrix} \right) \otimes v_iv_i^\top\\
        &\quad + \frac{1}{k(k-2)} \left(J_n + \begin{bmatrix}
            J_{n/2} & -J_{n/2}\\
            -J_{n/2} & J_{n/2}
        \end{bmatrix}\right) \otimes \left(\sum_{i:|\lambda_i| < d} v_iv_i^\top\right)\\
        &\quad + \frac{1}{k(k-2)}J_n \otimes \left(\frac{1}{k}J_k + \frac{1}{k}\begin{bmatrix}
            J_{k/2} & -J_{k/2}\\
            -J_{k/2} & J_{k/2}
        \end{bmatrix} - I_k\right)\\
        &\quad + \frac{1}{k(k-2)}\begin{bmatrix}
            J_{n/2} & -J_{n/2}\\
            -J_{n/2} & J_{n/2}
        \end{bmatrix} \otimes \left(\frac{k-1}{k} \begin{bmatrix}
            J_{k/2} & -J_{k/2}\\
            -J_{k/2} & J_{k/2}
        \end{bmatrix} + \frac{1}{k} J_k - I_k\right)\\
        &= \frac{1}{k-2} \sum_{i: |\lambda_i| < d}\left(\Tilde{P}^{(\lambda_i)} - \frac{1}{k}J_n - \frac{1}{k}\begin{bmatrix}
            J_{n/2} & -J_{n/2}\\
            -J_{n/2} & J_{n/2}
        \end{bmatrix} \right) \otimes v_iv_i^\top\\
        &\quad + \frac{1}{k(k-2)} \left(J_n + \begin{bmatrix}
            J_{n/2} & -J_{n/2}\\
            -J_{n/2} & J_{n/2}
        \end{bmatrix}\right) \otimes \left(I_k - \frac{1}{k}J_k - \frac{1}{k}\begin{bmatrix}
            J_{k/2} & -J_{k/2}\\
            -J_{k/2} & J_{k/2}
        \end{bmatrix} \right)\\
        &\quad + \frac{1}{k(k-2)}J_n \otimes \left(\frac{1}{k}J_k + \frac{1}{k}\begin{bmatrix}
            J_{k/2} & -J_{k/2}\\
            -J_{k/2} & J_{k/2}
        \end{bmatrix} - I_k\right)\\
        &\quad + \frac{1}{k(k-2)}\begin{bmatrix}
            J_{n/2} & -J_{n/2}\\
            -J_{n/2} & J_{n/2}
        \end{bmatrix} \otimes \left(\frac{k-1}{k} \begin{bmatrix}
            J_{k/2} & -J_{k/2}\\
            -J_{k/2} & J_{k/2}
        \end{bmatrix} + \frac{1}{k} J_k - I_k\right)\\
        &= \frac{1}{k-2} \sum_{i: |\lambda_i| < d}\left(\Tilde{P}^{(\lambda_i)} - \frac{1}{k}J_n - \frac{1}{k}\begin{bmatrix}
            J_{n/2} & -J_{n/2}\\
            -J_{n/2} & J_{n/2}
        \end{bmatrix} \right) \otimes v_iv_i^\top\\
        &\quad + \frac{1}{k^2}\begin{bmatrix}
            J_{n/2} & -J_{n/2}\\
            -J_{n/2} & J_{n/2}
        \end{bmatrix} \otimes \begin{bmatrix}
            J_{k/2} & -J_{k/2}\\
            -J_{k/2} & J_{k/2}
        \end{bmatrix} \\
        &\succeq 0,
    \end{align*}}
    where we use the properties
    \begin{align*}
        \Tilde{P}^{(\lambda_i)} - \frac{1}{k}J_n - \frac{1}{k}\begin{bmatrix}
            J_{k/2} & -J_{k/2}\\
            -J_{k/2} & J_{k/2}
        \end{bmatrix} &\succeq 0, \\
        v_iv_i^\top &\succeq 0, \\
        \begin{bmatrix}
            J_{n/2} & -J_{n/2}\\
            -J_{n/2} & J_{n/2}
        \end{bmatrix} &\succeq 0, \\
        \begin{bmatrix}
            J_{k/2} & -J_{k/2}\\
            -J_{k/2} & J_{k/2}
        \end{bmatrix} &\succeq 0,
    \end{align*}
    and that tensor product of positive semidefinite matrices is again positive semidefinite. By the Schur complement criterion, we conclude that the pseudomoment matrix is PSD.

    Next we verify the hard constraints:
    \begin{align*}
        \Tilde{\mathbb{E}}[x_{v,i}^2] &= (Q_{v,v})_{i,i}\\
        &= \frac{1}{k-2} \sum_{k: |\lambda_k| < d} \Tilde{P}^{(\lambda_k)}_{v,v} (v_kv_k^\top)_{i,i}\\
        &= \frac{1}{k-2} \left(I - \frac{1}{k}J_k - \frac{1}{k} \begin{bmatrix}
            J_{k/2} & -J_{k/2}\\
            -J_{k/2} & J_{k/2}
        \end{bmatrix}\right)_{i,i}\\
        &= \frac{1}{k}\\
        &= \Tilde{\mathbb{E}}[x_{v,i}]
    \end{align*}
    and
    \begin{align*}
        &\Tilde{\mathbb{E} }[(x_{u,1} + \cdots + x_{u,k}) x_{v,i}] \\
        &= \sum_{j = 1}^k \frac{1}{k-2} \Bigg(\sum_{r: |\lambda_r| < d} \Tilde{P}^{(\lambda_r)}_{u,v} \left(v_rv_r^\top\right)_{j,i}  + \frac{1}{k} \left(\frac{k-1}{k}J_k + \frac{1}{k}\begin{bmatrix}
            J_{k/2} & -J_{k/2}\\
            -J_{k/2} & J_{k/2}
        \end{bmatrix} - I_k\right)_{i,j}\\
        &\quad + \frac{1}{k} \begin{bmatrix}
            J_{k/2} & -J_{k/2}\\
            -J_{k/2} & J_{k/2}
        \end{bmatrix}_{u,v} \left(\frac{k-1}{k}\begin{bmatrix}
            J_{k/2} & -J_{k/2}\\
            -J_{k/2} & J_{k/2}
        \end{bmatrix} + \frac{1}{k}J_k - I_k\right)_{i,j}\Bigg)
        \intertext{notice that $\frac{k-1}{k}J_k + \frac{1}{k}\begin{bmatrix}
            J_{k/2} & -J_{k/2}\\
            -J_{k/2} & J_{k/2}
        \end{bmatrix} - I_k$ has row sum equal to $k-2$, and $\frac{k-1}{k}\begin{bmatrix}
            J_{k/2} & -J_{k/2}\\
            -J_{k/2} & J_{k/2}
        \end{bmatrix} + \frac{1}{k}J_k - I_k$ has row sum equal to $0$, and thus}
        &= \frac{1}{k} + \frac{1}{k-2}\sum_{r: |\lambda_r| < d} \Tilde{P}^{(\lambda_r)}_{u,v} \sum_{j=1}^k \left(v_rv_r^\top\right)_{j,i}\\
        &= \frac{1}{k} + \frac{1}{k-2}\sum_{r: |\lambda_r| < d} \Tilde{P}^{(\lambda_r)}_{u,v} (v_r)_i\sum_{j=1}^k (v_r)_j\\
        &= \frac{1}{k}\\
        &= \Tilde{\mathbb{E} }[x_{v,i}],
    \end{align*}
    where we use the fact that $\sum_{j} (v_r)_j = \langle v_r, (1, \dots, 1)^\top \rangle = 0$ in the second to last line, since $v_r$ is an eigenvector orthogonal to the constant vector, the eigenvector of eigenvalue of $d$.

    We moreover verify the label constraints, i.e., the hard moment constraints on the edgeless labelled graphs. For two labels $i,j \in [k]$, we say they are in the same group if either $i,j \in \{1, \dots, \frac{k}{2}\}$ or $i,j \in \{\frac{k}{2} + 1, \dots, k\}$, and we say they are in different groups otherwise. We have
    \[ \sum_{u \in [n]} \Tilde{\mathbb{E}}[x_{u,i}] = \frac{n}{k} \]
    and
    {\allowdisplaybreaks
    \begin{align*}
        &\sum_{u, v\in [n]: u\ne v} \Tilde{\mathbb{E}}[x_{u,i}x_{v,j}] \\
        &= \sum_{u, v\in [n]: u\ne v} \frac{1}{k-2} \sum_{r: |\lambda_r| < d} \Tilde{P}^{(\lambda_r)}_{u,v} (v_rv_r^\top)_{i,j}\\
        &\quad +\sum_{u,v\in [n]: u\ne v} \frac{1}{k(k-2)}\left(\frac{k-1}{k}J_k + \frac{1}{k}\begin{bmatrix}
            J_{k/2} & -J_{k/2}\\
            -J_{k/2} & J_{k/2}
        \end{bmatrix} - I_k\right)_{i,j}\\
        &\quad + \sum_{u,v\in [n]: u\ne v} \frac{1}{k(k-2)}\begin{bmatrix}
            J_{n/2} & -J_{n/2}\\
            -J_{n/2} & J_{n/2}
        \end{bmatrix}_{u,v}\left(\frac{k-1}{k}\begin{bmatrix}
            J_{k/2} & -J_{k/2}\\
            -J_{k/2} & J_{k/2}
        \end{bmatrix} + \frac{1}{k}J_k - I_k\right)_{i,j}\\
        &= \frac{n(n-1)}{k(k-2)}\left
        \{ \begin{array}{ll}
             0 & \quad \text{ if } i=j,\\
            1 & \quad \text{ if } i\ne j \text{ in the same group}\\
            \frac{k-2}{k} & \quad \text{ if } i\ne j \text{ in different groups}
        \end{array}\right\}\\
        &\quad + \frac{-n}{k(k-2)}\left
        \{ \begin{array}{ll}
             0 & \quad \text{ if } i=j,\\
            1 & \quad \text{ if } i\ne j \text{ in the same group}\\
            -\frac{k-2}{k} & \quad \text{ if } i\ne j \text{ in different groups}
        \end{array}\right\}\\
        &\quad + \frac{1}{k-2} \sum_{r: |\lambda_r| < d} \left( (v_rv_r^\top)_{i,j} \sum_{u, v\in [n]: u\ne v}\Tilde{P}^{(\lambda_r)}_{u,v}\right)\\
        &= \left
        \{ \begin{array}{ll}
             0 & \quad \text{ if } i=j,\\
            \frac{n(n-2)}{k(k-2)} & \quad \text{ if } i\ne j \text{ in the same group}\\
            \frac{n^2}{k^2} & \quad \text{ if } i\ne j \text{ in different groups}
        \end{array}\right\}\\
        &\quad + \frac{1}{k-2} \sum_{r: |\lambda_r| < d}  (v_rv_r^\top)_{i,j} \left(\langle \Tilde{P}^{(\lambda_r)}, J_n \rangle - \sum_{u \in [n]} \Tilde{P}^{(\lambda_r)}_{u,u}\right)\\
        &= \left
        \{ \begin{array}{ll}
             0 & \quad \text{ if } i=j,\\
            \frac{n(n-2)}{k(k-2)} & \quad \text{ if } i\ne j \text{ in the same group}\\
            \frac{n^2}{k^2} & \quad \text{ if } i\ne j \text{ in different groups}
        \end{array}\right\}\\
        &\quad + \frac{1}{k-2} \left(I_k - \frac{1}{k}J_k - \frac{1}{k}\begin{bmatrix}
            J_{k/2} & -J_{k/2}\\
            -J_{k/2} & J_{k/2}
        \end{bmatrix}\right)_{i,j} \left(\frac{n^2}{k} - n\right)\\
        &= \left
        \{ \begin{array}{ll}
             0 & \quad \text{ if } i=j,\\
            \frac{n(n-2)}{k(k-2)} & \quad \text{ if } i\ne j \text{ in the same group}\\
            \frac{n^2}{k^2} & \quad \text{ if } i\ne j \text{ in different groups}
        \end{array}\right\} + \left
        \{ \begin{array}{ll}
             \frac{n}{k}\left(\frac{n}{k} - 1\right) & \quad \text{ if } i=j,\\
            -\frac{2n(n-k)}{k^2(k-2)} & \quad \text{ if } i\ne j \text{ in the same group}\\
            0 & \quad \text{ if } i\ne j \text{ in different groups}
        \end{array}\right\}\\
        &= \left\{\begin{array}{ll}
            \frac{n}{k}\left(\frac{n}{k} -1\right) & \quad \text{ if } i = j,\\
            \frac{n^2}{k^2} & \quad \text{ if } i\ne j
        \end{array}\right\},
    \end{align*}}
    as desired.

    Finally, we verify the moment constraints. By Lemma~\ref{lem:minimal-graphs}, we only need verify the moment constraints on the class of pruned partially labelled forests. So long as these constraints are satisfied, the remaining moment constraints are automatically satisfied with high probability.

    As before, the moment constraints of pruned partially labelled forests with isolated distinguished vertices are subsumed by the label constraints.

    Now the only other pruned partially labelled forests with two distinguished vertices are paths with two labelled endpoints. It remains to verify the moment constraints imposed by these paths. For a path of length $s$ with two endpoints labelled $i$ and $j$, we verify
    {\allowdisplaybreaks
    \begin{align*}
        &\quad \Tilde{\mathbb{E}}[p_{(P_s, \{0, s\}, \{i,j\})}]\\
        &=\langle Q_{i,j}, A_{G_0}^{\langle s \rangle}\rangle\\
        &= \langle Q_{i,j}, A_{G_0}^{(s)}\rangle + o(n)
        \intertext{again by the same reason that $A_{G_0}^{\langle s \rangle}$ and $A_{G_0}^{(s)}$ are close and $Q_{i,j}$ has its entries bounded by $\frac{1}{k}$,}
        &= \left\langle \frac{1}{k-2}\sum_{r: |\lambda_r| < d} \Tilde{P}^{(\lambda_r)} \cdot \left(v_rv_r^\top\right)_{i,j}, A_{G_0}^{(s)} \right\rangle\\
        &\quad+ \left\langle \frac{1}{k(k-2)} J_n \cdot \left(\frac{k-1}{k}J_k + \frac{1}{k}\begin{bmatrix}
            J_{k/2} & -J_{k/2}\\
            -J_{k/2} & J_{k/2}
        \end{bmatrix} - I_k\right)_{i,j} , A_{G_0}^{(s)}\right\rangle\\
        &\quad + \left\langle \frac{1}{k(k-2)} \begin{bmatrix}
            J_{n/2} & -J_{n/2}\\
            -J_{n/2} & J_{n/2}
        \end{bmatrix} \cdot \left(\frac{k-1}{k}\begin{bmatrix}
            J_{k/2} & -J_{k/2}\\
            -J_{k/2} & J_{k/2}
        \end{bmatrix} + \frac{1}{k}J_k - I_k\right)_{i,j} , A_{G_0}^{(s)}\right\rangle +  o(n)\\
        &= \frac{1}{k-2}\sum_{r: |\lambda_r| < d} \left(v_rv_r^\top\right)_{i,j} \langle \Tilde{P}^{(\lambda_r)}, A_{G_0}^{(s)} \rangle \\
        &\quad + \frac{1}{k(k-2)}\left(\frac{k-1}{k}J_k + \frac{1}{k}\begin{bmatrix}
            J_{k/2} & -J_{k/2}\\
            -J_{k/2} & J_{k/2}
        \end{bmatrix} - I_k\right)_{i,j} \langle J_n, A_{G_0}^{(s)}\rangle\\
        &\quad + \frac{1}{k(k-2)}\left(\frac{k-1}{k}\begin{bmatrix}
            J_{k/2} & -J_{k/2}\\
            -J_{k/2} & J_{k/2}
        \end{bmatrix} + \frac{1}{k}J_k - I_k\right)_{i,j} \left\langle\begin{bmatrix}
            J_{n/2} & -J_{n/2}\\
            -J_{n/2} & J_{n/2}
        \end{bmatrix}, A_{G_0}^{(s)} \right\rangle+ o(n) \\
        &= \sum_{r: |\lambda_r| < d} \left(v_rv_r^\top\right)_{i,j} \left( \frac{1}{k} q_s(\lambda_r)n + \frac{1}{k(k-2)} q_s(d) n + \frac{1}{k(k-2)} q_s(-d)n \pm \frac{1}{k-2}\delta n \right)\\
        &\quad + \frac{1}{k(k-2)}\left(\frac{k-1}{k}J_k + \frac{1}{k}\begin{bmatrix}
            J_{k/2} & -J_{k/2}\\
            -J_{k/2} & J_{k/2}
        \end{bmatrix} - I_k\right)_{i,j} q_s(d)n\\
        &\quad + \frac{1}{k(k-2)} \left(\frac{k-1}{k}\begin{bmatrix}
            J_{k/2} & -J_{k/2}\\
            -J_{k/2} & J_{k/2}
        \end{bmatrix} + \frac{1}{k}J_k - I_k\right)_{i,j} q_s(-d) n + o(n)
        \intertext{note that $|v_rv_r^\top|_{i,j} \le 1$ since $v_r$ are unit vectors, and thus the total error term is bounded by $\pm \delta n$,}
        &\in \frac{n}{k} \left(\sum_{r: |\lambda_r| < d} q_s(\lambda_r) v_rv_r^\top\right)_{i,j} \\
        &\quad + \frac{n}{k(k-2)} q_s(d)\cdot \left(\sum_{r: |\lambda_r| < d} v_rv_r^\top + \frac{k-1}{k}J_k + \frac{1}{k}\begin{bmatrix}
            J_{k/2} & -J_{k/2}\\
            -J_{k/2} & J_{k/2}
        \end{bmatrix} - I_k \right)_{i,j} \\
        &\quad + \frac{n}{k(k-2)}q_s(-d)\cdot\left(\sum_{r: |\lambda_r| < d} v_rv_r^\top + \frac{k-1}{k}\begin{bmatrix}
            J_{k/2} & -J_{k/2}\\
            -J_{k/2} & J_{k/2}
        \end{bmatrix} + \frac{1}{k}J_k - I_k \right)_{i,j} + [-\delta n, \delta n] + o(n) \\
        &= \frac{n}{k}q_s\left(M - d\frac{1}{k}J_k + d\frac{1}{k}\begin{bmatrix}
            J_{k/2} & -J_{k/2}\\
            -J_{k/2} & J_{k/2}
        \end{bmatrix}\right)_{i,j} + [-\delta n, \delta n]  + o(n)\\
        &\quad + \frac{n}{k(k-2)}q_s(d) \cdot \left(I_k - \frac{1}{k}J_k - \frac{1}{k}\begin{bmatrix}
            J_{k/2} & -J_{k/2}\\
            -J_{k/2} & J_{k/2}
        \end{bmatrix} + \frac{k-1}{k}J_k + \frac{1}{k}\begin{bmatrix}
            J_{k/2} & -J_{k/2}\\
            -J_{k/2} & J_{k/2}
        \end{bmatrix} - I_k\right)_{i,j}\\
        &\quad + \frac{n}{k(k-2)}q_s(-d) \cdot \left(I_k - \frac{1}{k}J_k - \frac{1}{k}\begin{bmatrix}
            J_{k/2} & -J_{k/2}\\
            -J_{k/2} & J_{k/2}
        \end{bmatrix} + \frac{k-1}{k}\begin{bmatrix}
            J_{k/2} & -J_{k/2}\\
            -J_{k/2} & J_{k/2}
        \end{bmatrix} + \frac{1}{k}J_k - I_k\right)_{i,j} \\
        &=  \frac{n}{k}\left[q_s\left(M - d\frac{J_k}{k} + d\frac{1}{k}\begin{bmatrix}
            J_{k/2} & -J_{k/2}\\
            -J_{k/2} & J_{k/2}
        \end{bmatrix} \right) +q_s\left(d \frac{1}{k}J_k\right) + q_s\left(-d \frac{1}{k}\begin{bmatrix}
            J_{k/2} & -J_{k/2}\\
            -J_{k/2} & J_{k/2}
        \end{bmatrix}\right)\right]_{i,j} \\
        &\quad + [-\delta n, \delta n] + o(n)\\
        &= \frac{n}{k}q_s(M)_{i,j} \pm \delta n + o(n),
    \end{align*}}
    which satisfies the moment constraints with error tolerance $2\delta$.

    Thus, using the solutions from the level-$D$ Bipartite Symmetric Path Statistics SDP, we have constructed a degree-$2$ pseudomoment $\Tilde{\mathbb{E}}$ for the degree-$(2, D)$ Local Statistics SDP that is w.h.p.~feasible on input $G_0 \sim \Tilde{\mathbb{Q}}$.

\end{proof}

\subsection{Robustness of Local Statistics SDP}

In this section, we show that our results in Theorem \ref{thm:local-statistics-noiseless} is furthermore robust against adversarial noise. More specifically, we finish the proof of Theorem \ref{thm:local-statistics-informal} by analyzing the effect of the adversarial noise $\mathcal{S}_{\varepsilon}$.

Let $G_0$ be a $d$-regular input graph to the Local Statistics SDP, and let $\Tilde{G}_0$ be the $d$-regular graph obtained from $G_0$ after the application of the adversarial noise $\mathcal{S}_{\varepsilon}$. By definition of the noise model in Section \ref{sec:prelim:noise-conjectures} $\mathcal{S}_{\varepsilon}$, the graphs $G_0$ and $\Tilde{G}_0$ differ by a set of at most $\varepsilon n$ edges. In particular, for $\varepsilon$ small enough, the local neighborhoods of most vertices remain unaffected by the noise operator, and this essentially guarantees that the matrices associated to $G_0$ and $\Tilde{G}_0$ of interest in our analysis will remain ``close" after the perturbation.

To show the robustness, we need to show
\begin{enumerate}[(a)]
    \item When $G_0 \sim \mathcal{P}$ with $\rho(H) \le 2\sqrt{d-1}$, the Local Statistics SDP w.h.p.~remains feasible after applying the adversarial noise operator $\mathcal{S}_{\varepsilon}$ to $G$, for a small enough $\varepsilon$.
    \item When $G_0 \sim \mathcal{P}$ with $\rho(H) > 2\sqrt{d-1}$, the Local Statistics SDP w.h.p.~remains infeasible after applying the adversarial noise operator $\mathcal{S}_{\varepsilon}$ to $G$, for a small enough $\varepsilon$.
\end{enumerate}

To prove (a), we note that $A_{G_0}^{\langle s\rangle}$ and $A_{\Tilde{G}_0}^{\langle s\rangle}$ will be close after the perturbation of a small linear number of edges. More precisely, the sum of entrywise absolute value of the difference $A_{G_0}^{\langle s\rangle} - A_{\Tilde{G}_0}^{\langle s\rangle}$ is bounded by $O(\varepsilon n)$ with the hidden constant depending only on $d$ and $s$; we refer the reader to \cite[Section 4.4]{banks2021local} for a detailed argument of this claim.\footnote{\cite[Section 4.4]{banks2021local} proves the analogous claim for the non-backtracking matrices $A_{G_0}^{(s)}$ and $A_{\Tilde{G}_0}^{(s)}$, but it is clear that the same argument would work for the self-avoiding matrices.} Moreover, we note that in the computation of the moment constraints $\langle Q_{i,j}, A_{\Tilde{G}_0}^{\langle s\rangle} \rangle$, the matrix $Q_{i,j}$ have all entries bounded by $1$ in magnitude, since the pseudomoment matrix is PSD whose diagonal entries can be easily verified to be bounded by $1$. Together with the entry-wise bound of $A_{G_0}^{\langle s\rangle} - A_{\Tilde{G}_0}^{\langle s\rangle}$, we conclude that $\langle Q_{i,j}, A_{\Tilde{G}_0}^{\langle s\rangle} \rangle = \langle Q_{i,j}, A_{G_0}^{\langle s\rangle} \rangle + O(\varepsilon n)$, which will still satisfy the moment constraints for an appropriately chosen error tolerance $\delta$.

To prove (b), we again use that $A_{G_0}^{\langle s\rangle}$ and $A_{\Tilde{G}_0}^{\langle s\rangle}$ are close as discussed above. We can similarly repeat the upper bound proof of Theorem \ref{thm:local-statistics-noiseless} by swapping $A_{\Tilde{G}_0}^{\langle s\rangle}$ with $A_{G_0}^{\langle s\rangle}$ to show that the moment constraints are still approximately satisfied after the perturbation of a small linear number of edges, and conclude the infeasibility of the Local Statistics SDP by reducing to the infeasibility of the Path Statistics SDP in Theorem~\ref{thm:PS-main}.

\section{Applications}

\subsection{General Theory}
\label{sec:certification-theory}

We first develop some general statements about the implications that Conjecture~\ref{conj:hardness-formal} (the formal version of Conjecture~\ref{conj:hardness}) can have for certification.
We hope that these will prove useful in future work.

\begin{definition}[Ramanujan supremum]
    Let $f$ be a real-valued function of a $d$-regular graph.
    When $G$ is a lift of a base graph $H$ with adjacency matrix $M$, we say that another graph $G^{\prime}$ on the same vertex set as $G$ \emph{$H$-respects} $G$ if it contains no edges between vertices in fibers of $i, j \in V(H)$ with $M_{ij} = 0$.
    And, we call $G$ and $G^{\prime}$ \emph{cobipartite} if they are both bipartite on the same bipartition.
    We define the following:
     {\allowdisplaybreaks
    \begin{align*}
        M(f) &\colonequals \sup_c\bigg\{\text{there exists $d$-regular Ramanujan } H \text{ such that }
        \min_{\Delta(G, G^{\prime}) \leq \epsilon} f(G^{\prime}) \geq c \\ &\hspace{1.45cm} \text{ with high probability as } m \to \infty \text{ when } G \sim \sL_m(H) \text{ for all } \epsilon \text{ sufficiently small} \bigg\}, \\
        \widetilde{M}(f) &\colonequals \sup_c\bigg\{\text{there exists $d$-regular Ramanujan } H \text{ such that }
        \min_{\substack{\Delta(G, G^{\prime}) \leq \epsilon \\ G^{\prime} \text{ } H\text{-respects } G}} f(G^{\prime}) \geq c \\ &\hspace{1.45cm} \text{ with high probability as } m \to \infty \text{ when } G \sim \sL_m(H) \text{ for all } \epsilon \text{ sufficiently small} \bigg\}, \\
        M_{\bi}(f) &\colonequals \sup_c\bigg\{\text{there exists $d$-regular bipartite Ramanujan } H \text{ such that }
        \min_{\substack{\Delta(G, G^{\prime}) \leq \epsilon \\ G, G^{\prime} \text{ cobipartite}}} f(G^{\prime}) \geq c \\ &\hspace{1.45cm} \text{ with high probability as } m \to \infty \text{ when } G \sim \sL_m(H) \text{ for all } \epsilon \text{ sufficiently small} \bigg\}.\\
        \widetilde{M}_{\bi}(f) &\colonequals \sup_c\bigg\{\text{there exists $d$-regular bipartite Ramanujan } H \text{ such that }
        \min_{\substack{\Delta(G, G^{\prime}) \leq \epsilon \\ G^{\prime} \text{ } H\text{-respects } G}} f(G^{\prime}) \geq c \\ &\hspace{1.45cm} \text{ with high probability as } m \to \infty \text{ when } G \sim \sL_m(H) \text{ for all } \epsilon \text{ sufficiently small} \bigg\}.
    \end{align*}}
    Note in the last two definitionsthat if $H$ is bipartite and $G$ $H$-respects $G^{\prime}$ a lift of $H$, then $G$ and $G^{\prime}$ are also cobipartite.
\end{definition}
\noindent
We state an abstract result on certification using these quantities.
\begin{theorem}
    \label{thm:certification}
    Let $f$ be a real-valued function of a $d$-regular graph and $\delta > 0$.
    \begin{enumerate}
    \item If Conjecture~\ref{conj:hardness-formal} holds for Ramanujan base graphs with (non-bipartite, non-respectful) random noise or adversarial noise, then there is no polynomial-time algorithm that certifies an upper bound of $M(f) - \delta$ on $f(G)$ with high probability when $G \sim \sG(n, d)$.
    \item If Conjecture~\ref{conj:hardness-formal} holds for Ramanujan base graphs with respectful random or adversarial noise, then there is no polynomial-time algorithm that certifies an upper bound of $\widetilde{M}(f) - \delta$  on $f(G)$ with high probability when $G \sim \sG(n, d)$.
    \item If Conjecture~\ref{conj:hardness-formal} holds for bipartite Ramanujan base graphs with bipartite random or adversarial noise, then there no polynomial-time algorithm that certifies an upper bound of $M_{\bi}(f) - \delta$ on $f(G)$ with high probability when $G \sim \sG((\frac{n}{2}, \frac{n}{2}), d)$.
    \item If Conjecture~\ref{conj:hardness-formal} holds for bipartite Ramanujan base graphs with respectful bipartite random or respectful adversarial noise, then there no polynomial-time algorithm that certifies an upper bound of $\widetilde{M}_{\bi}(f) - \delta$ on $f(G)$ with high probability when $G \sim \sG((\frac{n}{2}, \frac{n}{2}), d)$.
    \end{enumerate}
\end{theorem}
\noindent
The proof is a straightforward combination of the Conjecture (in its various forms) with the definitions of the quantities involved.

In practice, what will be more useful for most situations than this abstract statement is the following simpler corollary.

\begin{definition}[Lift-monotone property]
    We say that a function $f(G)$ of a $d$-regular graph $G$ is \emph{lift-monotone} if for any $d$-regular graph $H$, whenever $L$ is a lift of $H$, then $f(L) \geq f(H)$.
\end{definition}

\begin{definition}[Robustly lift-monotone property]
    We say that a function $f(G)$ of a $d$-regular graph $G$ is \emph{robustly lift-monotone} if, for any fixed $d$-regular multigraph $H$, there is some $\delta_H: \RR \to \RR$ such that $\delta_H(\epsilon) \to 0$ as $\epsilon \to 0$, and, whenever $L$ is a lift of $H$ and $\Delta(L^{\prime}, L) \leq \eps$, then $f(L^{\prime}) \geq f(H) - \delta_H(\epsilon)$.
\end{definition}

\begin{proposition}
    \label{prop:ramanujan-bounds}
    Suppose that $f$ satisfies one of the following:
    \begin{enumerate}
    \item $f$ is lift-monotone and $C$-Lipschitz for some constant $C > 0$ with respect to the metric $\Delta(\cdot, \cdot)$.
    \item $f$ is robustly lift-monotone.
    \item Whenever $G$ is a lift of $H$ and $G^{\prime}$ $H$-respects $G$, then $f(G^{\prime}) \geq f(H)$.
    \end{enumerate}
    Then,
    \begin{align*}
        M(f), \widetilde{M}(f) &\geq \sup_{H \text{ $d$-regular Ramanujan}} f(H), \\
        M_{\bi}(f), \widetilde{M}_{\bi}(f) &\geq \sup_{H \text{ $d$-regular bipartite Ramanujan}} f(H).
    \end{align*}
\end{proposition}

\begin{remark}[Beyond lift-monotone properties]
    \label{rem:beyond-lift-monotone}
    As the title of this paper indicates, we are mostly focusing on lift-monotonicity as a tool for proving lower bounds against certification.
    However, this is not at all necessary, and we have intentionally formulated Theorem~\ref{thm:certification} in a more abstract way to emphasize that any means of controlling properties of lifts suffices to execute the quiet planting strategy.
    Relatively little is known about this question for random lifts of general base graphs for the quantities we will be focusing on, but some preliminary results may be found in the early line of work \cite{ALMR-2001-RandomLifts,ALM-2002-RandomLiftsIndependenceChromatic,AL-2006-RandomLiftsExpansion}.
\end{remark}

Finally, we point out the following intriguing phenomenon, which suggests that proving lower bounds against certification for general graphs is harder than for bipartite graphs (a phenomenon reminiscent of how it is easier, per the work of \cite{MSS-2013-InterlacingFamiliesBipartiteRamanujan}, to construct bipartite Ramanujan graphs than general Ramanujan graphs).

\begin{theorem}
    \label{thm:bipartite-unipartite}
    If $H$ is a Ramanujan graph, then there exists $H^{\prime}$ a lift of $H$ that is a bipartite Ramanujan graph.
    Consequently, if $f$ is lift-monotone, then
    \begin{equation}
         \sup_{H \text{ $d$-regular bipartite Ramanujan}} f(H) \geq \sup_{H \text{ $d$-regular Ramanujan}} f(H).
    \end{equation}
\end{theorem}
\begin{proof}
    The ``canonical'' or Kronecker double cover of $H$ gives the requisite lift \cite{BHM-1980-BigraphsDigraphs}, and the other result follows immediately.
\end{proof}
\noindent
As a consequence, hardness results in the style of those below for ordinary graphs drawn from $\sG(n, d)$ translate immediately to the same results for bipartite graphs drawn from $\sG((\frac{n}{2}, \frac{n}{2}), d)$.

\begin{remark}
    Similarly, a statement like $\widetilde{M}_{\bi}(f) \geq \widetilde{M}(f)$ will be true provided that $f$ is robustly lift-monotone, but only with alternative definitions that ask for a large value of $f(G)$ to hold over $G \sim \sL_m(H)$ with probability 1, rather than with high probability.
\end{remark}

\subsection{Maximum $t$-Cut}
\label{sec:max-cut}

We recall that $\MC_t(G)$ is defined as the normalized size of the maximum $t$-cut of $G$:
\begin{equation*}
   \MC_t(G) \colonequals \max_{\kappa : V \to [t]} \frac{|\{\{u, v\} \in E: \kappa(u) \neq \kappa(v)\}|}{|E|} \in [0, 1].
\end{equation*}

\begin{proposition}
    $\MC_t(G)$ is robustly lift-monotone.
\end{proposition}
\begin{proof}
    Let $H$ be a $d$-regular graph, and $L$ be any lift of $H$. Suppose $\kappa: V(H) \to [t]$ encodes the maximum $t$-cut of $H$. Now define $\kappa: V(L) \to [k]$ by assigning $\kappa(v)$ to vertices in the fiber of $v \in V(H)$. It is easy to see that $\kappa_L$ encodes a $t$-cut in $L$ of the same normalized size as $\kappa$.

    Suppose $L'$ is a $d$-regular graph such that $\Delta(L, L') \le \epsilon$. Let $n = |V(L)|$, and $E' \colonequals E(L)\setminus E(L') \subseteq E(L)$ be the set of edges in $L$ but not in $L'$. $|E'| \le \epsilon n$. We note that \[\{\{u,v\} \in E(L): \kappa_L(u) \ne \kappa_L(v)\} \subseteq \{\{u,v\} \in E(L'): \kappa_L(u) \ne \kappa_L(v)\} \cup E',\] since every edge that exists in both graphs receives the same pair of labels given by $\kappa_L$. Then,
    \begin{align*}
        \MC_t(H) &= \frac{\left| \{\{u,v\} \in E(H): \kappa(u) \ne \kappa(v)\}\right|}{|E(H)|}\\
        &= \frac{\left| \{\{u,v\} \in E(L): \kappa_L(u) \ne \kappa_L(v)\}\right|}{|E(L)|}\\
        &\le \frac{\left| \{\{u,v\} \in E(L'): \kappa_L(u) \ne \kappa_L(v)\}\right|}{|E(L')|} + \frac{|E'|}{|E(L')|}\\
        &\le \MC_t(L') + \frac{\epsilon n}{\frac{dn}{2}}\\
        &= \MC_t(L') + \frac{2}{d}\epsilon,
    \end{align*}
    which finishes the proof.
\end{proof}

To prove Theorem~\ref{thm:max-cut}, we then invoke Theorem~\ref{thm:certification} together with Proposition~\ref{prop:ramanujan-bounds}, which together imply that it suffices to produce 3- and 4-regular Ramanujan graphs with suitably large cuts.
These are given in Figures~\ref{fig:d3-k2-example} and \ref{fig:d4-k2-example}, respectively.

\subsection{Chromatic Number}
\label{sec:pf:chromatic-number}

\begin{proposition}
    $-\chi(G)$ satisfies Condition 3 of Proposition~\ref{prop:ramanujan-bounds}.
\end{proposition}
\begin{proof}
    Let $H$ be a $d$-regular graph, and $L$ be any lift of $H$. Suppose $\chi(H) = c$ and $h: V(H) \to [c]$ is a proper coloring of $H$ using $c$ colors. Now consider $g: V(L) \to [c]$ defined by assigning color $h(v)$ to the vertices in the fiber of $v\in V(H)$. It is easy to verify that $g$ is a proper coloring of $L$. Thus,
    \begin{align*}
        - \chi(H) = - c \le -\chi(L).
    \end{align*}
    The same holds for any modification of $L$ that does not introduce edges within the color classes, which is true of any graph that $H$-respects $L$.
\end{proof}

Again, to prove Theorem~\ref{thm:chromatic} we may use Theorem~\ref{thm:certification} and Proposition~\ref{prop:ramanujan-bounds}, whereby it suffices to produce a 3-colorable, 7-regular Ramanujan graph, which we do in Figure~\ref{fig:d7-chrom-example}.
We note in this case that our argument only gives a lower bound on $\widetilde{M}(-\chi)$, so we obtain hardness conditional on the version of Conjecture~\ref{conj:hardness-formal} using the respectful noise operator.

\subsection{Maximum Independent Set}
\label{sec:ind-set}

Recall that we define the \emph{normalized independence number} of a graph $G$ to be $\what{\alpha}(G) \colonequals \alpha(G) / |V(G)|$.

\begin{proposition}
    \label{prop:ind-set-rlm}
    $\what{\alpha}$ is robustly lift-monotone.
\end{proposition}
\begin{proof}
    Let $H$ be a $d$-regular graph, and $L$ be any lift of $H$. Suppose $S \subseteq V(H)$ is the maximum independent set of $H$. Now consider $T \subseteq V(L)$ which is the union of all the fibers of vertices $v \in S$. It is easy to verify that $T$ is an independent set of $L$.

    Now suppose $L'$ is a $d$-regular graph such that $\Delta(L, L') \le \varepsilon$. Let $n = |V(L)|$, and $N \subseteq V(L')$ be the set of vertices in $L'$ whose incident edges differ from those in $L$. Clearly, $|N| \le 2\epsilon n$. We note that $T \setminus N \subseteq V(L')$ is an independent set of $L'$, since the incident edges of any $v \in T \setminus N$ are the same as those in $L$. Thus,
    \begin{align*}
        \what{\alpha}(H) = \frac{|S|}{|V(H)|} = \frac{|T|}{|V(L)|} \le \frac{|T \setminus N|}{|V(L')|} + \frac{|N|}{|V(L')|} \le \what{\alpha}(L') + 2\epsilon,
    \end{align*}
    which finishes the proof.
\end{proof}

\begin{proof}[Proof of Theorem~\ref{thm:ind-set-asymp}]
For this asymptotic result, we consider a family of $d$-regular base multigraphs $H = H^{(k, d)}$ on $k$ vertices, where $(k - 1) \mid d$, having adjacency matrix
\begin{equation}
    A = \left[
    \begin{array}{ccccc}
    0 & \frac{d}{k - 1} & \frac{d}{k - 1} & \cdots & \frac{d}{k - 1} \\
    \frac{d}{k - 1} & 0 & \frac{d}{k - 1} & \cdots & \frac{d}{k - 1} \\
    \frac{d}{k - 1} & \frac{d}{k - 1} & 0 & \cdots & \frac{d}{k - 1} \\
    \vdots & \vdots & \vdots & \ddots & \vdots \\
    \frac{d}{k - 1} & \frac{d}{k - 1} & \frac{d}{k - 1} & \cdots & 0
    \end{array}
    \right] = \frac{d}{k - 1} 1_k1_k^{\top} - \frac{d}{k - 1}I_k.
\end{equation}
All non-trivial eigenvalues of this adjacency matrix are $-\frac{d}{k - 1}$, so $H$ is Ramanujan provided that $k \geq \frac{d}{2\sqrt{d - 1}} + 1$, while $\what{\alpha}(H) = \frac{1}{k}$.
Thus by Proposition~\ref{prop:ramanujan-bounds} we have that $\what{\alpha}$ satisfies
    \begin{equation}
        M(\what{\alpha}) \geq \frac{2}{\sqrt{d}} - O_{d \to \infty}\left(\frac{1}{d}\right).
    \end{equation}
Therefore, by Theorem~\ref{thm:certification}, conditional on Conjecture~\ref{conj:hardness-formal}, we find the stated optimality of the Hoffman bound.
\end{proof}

\begin{proof}[Proof of Theorem~\ref{thm:ind-set}]
    In this case, we will take fuller advantage of the power of Theorem~\ref{thm:certification} instead of merely relying on robust lift-monotonicity.
    Note that, if a base graph $H$ has a self-loop, then the entire associated fiber in $G \sim \sL_m(H)$ is excluded from the lifted independent set that shows $\what{\alpha}(G) \geq \what{\alpha}(H)$.
    This is excessive: if there is only a single loop attached to a vertex in $H$, then the induced subgraph on that fiber in $G$ will be a perfect matching, half of whose vertices can be included in the independent set in $G$.

    Let us formalize this idea.
    Define $H^{\prime}$ to be $H$ with all self-loops removed.
    Let $\ell(v)$ be the number of self-loops that $v$ has in $H$ for each $v \in V(H) = V(H^{\prime})$.
    Define a modified normalized independence number:
    \begin{equation}
        \what{\alpha}^{\prime}(H) \colonequals \max\left\{\#\{v \in I: \ell(v) = 0 \} + \frac{\#\{v \in I: \ell(v) = 1\}}{2}: I \text{ independent in } H^{\prime}\right\} \geq \what{\alpha}(H).
    \end{equation}
    (Of course, one may continue in the same fashion and allow for smaller fractions of fibers coming from vertices with more self-loops, but this simple observation is all we will need.)
    The argument above together with the robustness part of the argument for Proposition~\ref{prop:ind-set-rlm} implies
    \begin{equation}
        M(\what{\alpha}) \geq \sup_{H \text{ $d$-regular Ramanujan}} \what{\alpha}^{\prime}(H).
    \end{equation}

    Finally, we use this observation together with Theorem~\ref{thm:certification} to prove our lower bounds for $d \in \{3, 4\}$.
    As usual, it suffices to exhibit particular Ramanujan graphs.
    For $d = 3$, we use the same graph as for Theorem~\ref{thm:max-cut}, which is depicted in Figure~\ref{fig:d3-k2-example}.
    For $d = 4$, we use a different graph, given in Figure~\ref{fig:d4-ind-example} (this case does not require the treatment of loops above).
\end{proof}

We have seen in the course of the proof that, for the independence number, lift monotonicity alone does not always give a tight characterization of the size of the largest independent set in a random lift.
We observed a rather trivial instance of this phenomenon due to self-loops, but a more ranging version of it occurs as well, where even in simple graphs sometimes large independent sets in a random lift do not typically occupy a small number of entire fibers, but rather a small fraction of many fibers.
See \cite{ALMR-2001-RandomLifts,ALM-2002-RandomLiftsIndependenceChromatic} for initial results characterizing this phenomenon.

\subsection{Minimum Dominating Set}

Recall that we define the \emph{normalized domination number} of a graph $G$ to be     \begin{equation*}
        \dom(G) \colonequals \min_{\substack{S \subseteq V(G)
        \\ \text{a dominating set}}} \frac{|S|}{|V(G)|}.
    \end{equation*}

\begin{proposition}
    $-\dom(G)$ is robustly lift-monotone.
\end{proposition}
\begin{proof}
    Let $H$ be a $d$-regular graph, and $L$ be any lift of $H$. Suppose $S \subseteq V(H)$ is the minimum dominating set of $H$. Now consider $T \subseteq V(L)$ which is the union of all the fibers of vertices $v \in S$. It is easy to verify that $T$ is a dominating set of $L$.

    Suppose $L'$ is a $d$-regular graph such that $\Delta(L, L') \le \epsilon$. Let $n = |V(L)|$, and $N \subseteq V(L')$ be the set of vertices in $L'$ whose incident edges differ from those in $L$. Clearly, $|N| \le 2\epsilon n$. We note that $T \cup N$ is a dominating set of $L'$, since for every vertex $v \in V(L') \setminus N$, its neighborhood is uncorrupted and it is covered by some vertex in $T$, and every vertex $v \in N$ is included in $T \cup N$ by construction. Thus,
    \begin{align*}
        -\dom(H) = -\frac{|S|}{|V(H)|} = -\frac{|T|}{|V(L)|} \le -\frac{|T \cup N|}{|V(L')|} + \frac{|N|}{|V(L')|} \le -\dom(L') + 2\epsilon,
    \end{align*}
    which finishes the proof.
\end{proof}

Recall also, as mentioned in the Introduction, that for any $d$-regular graph $G$ on $n$ vertices, $\frac{n}{d+1}$ is always a trivial lower bound on the size of any dominating set of $G$, since any vertex of a dominating set can only ``cover'' $d+1$ vertices of $G$, namely itself and its $d$ neighbors.
Therefore, for any $d$-regular $G$,
\begin{equation}
   \dom(G) \geq \frac{1}{d + 1}.
\end{equation}
We proceed with our proof that this bound is optimal for each $d$.

\begin{proof}[Proof of Theorem~\ref{thm:dom}]
    We consider the complete $d$-regular graph $H = K_{d+1}$.
    The adjacency matrix of $H$ is $A = J_{d + 1} - I$, so its non-trivial eigenvalues are all $-1$ and $H$ is Ramanujan.
    $H$ also satisfies $\dom(H) = \frac{1}{d + 1}$, since any vertex alone is a dominating set.
    Thus, by Theorem~\ref{thm:certification} and Proposition~\ref{prop:ramanujan-bounds}, we find that the trivial lower bound above is optimal for polynomial-time certification algorithms conditional on Conjecture~\ref{conj:hardness-formal}.
\end{proof}

\subsection{Vertex Expansion}

Since for this problem (and edge expansion below) the definitions are somewhat more involved and were only briefly given in the Introduction, we review the main ingredients below.

\begin{definition}[Vertex boundary]
    Let $G$ be a graph and $S \subseteq V(G)$. The vertex boundary of $S$ is the set of vertices
    \begin{equation*}
        \partial_v S \colonequals \left\{u \in V(G): \{u,v\} \in E(G) \text{ for some } v \in S \right\}.
    \end{equation*}
    Note that the set $\partial_v S$ is not necessarily disjoint from $S$.\footnote{Another definition of vertex boundary that excludes the vertices from $S$ is also common.} When we need to be explicit with the underlying graph $G$, we denote the vertex boundary by $\partial_v^{G} S$.
\end{definition}

Recall that the small-set vertex expansion was then defined as
 \begin{equation}
        \Phi_{\epsilon}^{v}(G) \colonequals \min_{\substack{S \subseteq V \\ 1 \leq |S| \leq \epsilon n}} \frac{|\partial_v S|}{|S|}.
    \end{equation}

\begin{proposition}
    $-\Phi_{\epsilon}^{v}(G)$ is robustly lift-monotone.
\end{proposition}
\begin{proof}
    Let $H$ be a $d$-regular graph, and $L$ be any lift of $H$. Suppose $S \subseteq V(H)$ achieves the minimum vertex expansion among sets of size at most $\epsilon |V(H)|$, i.e., $\Phi_{\epsilon}^{v}(H) = \frac{|\partial_v^H S|}{|S|}$. Now consider $T \subseteq V(L)$ which is the union of all the fibers of vertices $v \in S$. Note that the vertex expansion of $T$ in $L$ is the same as the vertex expansion of $S$ in $H$ and that $T$ satisfies $\frac{|T|}{|V(L')|} = \frac{|S|}{|V(H)|} \le \epsilon$, since the vertex boundary of $T \subseteq V(L)$ is the union of the corresponding fibers of vertices in the vertex boundary of $S \subseteq V(H)$.

    Suppose $L'$ is a $d$-regular graph such that $\Delta(L, L') \le \gamma$. Let $n = |V(L)|$, and $N \subseteq V(L')$ be the set of vertices in $L'$ whose incident edges differ from those in $L$. Clearly, $|N| \le 2\gamma n$. We note that $\partial_v^{L} T \subseteq \partial_v^{L'} T \cup N$, since for every vertex $v \in V(L) \setminus N$, its neighborhood is the same as that in $L'$, and $v \in \partial_v^{L'} T$ if $v \in \partial_v^{L} T$. Finally, note that $|T| = |S|\cdot \frac{|V(L)|}{|V(H)|} = \frac{|S|}{|V(H)|} n$. Then,
    \begin{align*}
        -\Phi_{\epsilon}^{v}(H) = - \frac{|\partial_{v}^H S|}{|S|} = - \frac{|\partial_{v}^L T|}{|T|} \le -\frac{|\partial_{v}^{L'} T|}{|T|} + \frac{|N|}{|T|} \le -\Phi_{\epsilon}^{v}(L') + \frac{2\gamma n}{\frac{|S|}{|V(H)|} n} = -\Phi_{\epsilon}^{v}(L') + \frac{|V(H)|}{|S|}\gamma,
    \end{align*}
    which finishes the proof.
\end{proof}

Recall also that our benchmark certificate is Kahale's bound \cite{Kahale-1995-SpectralBoundExpansion} in terms of $\Tilde{\lambda}(G) \colonequals \max(\lambda_2(G), 2\sqrt{d - 1})$.
The bound says that, for an absolute constant $C > 0$,
    \begin{equation}
        \Phi_{\epsilon}^{v}(G) \geq \frac{d}{2}\left(1 - \sqrt{1 - \frac{4(d - 1)}{\Tilde{\lambda}(G)^2}}\right)\left(1 - C \cdot \frac{\log d}{\log \frac{1}{\epsilon}}\right).
    \end{equation}
    In particular, the algorithm outputting this lower bound with high probability certifies a bound of
    \begin{equation}
        \Phi_{\epsilon}^{v}(G) \geq \frac{d}{2}\left(1 - O_d \left(\frac{1}{\log \frac{1}{\epsilon}}\right)\right) - o_{d; n \to \infty}(1)
    \end{equation}
    when $G \sim \mathcal{G}(n,d)$ or $G \sim \mathcal{G}\left(\left(\frac{n}{2}, \frac{n}{2}\right), d\right)$.
    Further, for every fixed $d$, Kahale's spectral bound in the double limit of first taking $n \to \infty$ and then taking $\varepsilon \to 0$, with high probability certifies a lower bound of $\Phi_{\epsilon}^{v}(G) \ge \frac{d}{2} - o_{d; n \to \infty, \varepsilon \to 0}(1)$.

In our proof of Theorem~\ref{thm:vertex-exp}, we will consider a family of bipartite Ramanujan graphs due to Morgenstern \cite{Morgenstern-1994-RamanujanGraphsPrimePowers} and studied in \cite{kamber2022combinatorics}.

\begin{theorem}[{\cite[Theorem 1.1]{kamber2022combinatorics}}]
    \label{thm:kamber}
    For every prime power $q$, there exists an infinite family of $d = (q+1)$-regular bipartite Ramanujan graphs $G$, such that there exists a subset $S \subseteq V(G)$ satisfying $|S| = O(\sqrt{|V(G)|})$, and $\frac{|\partial_v S|}{|S|} = \frac{d}{2}$.
\end{theorem}

\begin{proof}[Proof of Theorem~\ref{thm:vertex-exp}]
    Using Theorem~\ref{thm:certification} and Proposition~\ref{prop:ramanujan-bounds} together with the bipartite Ramanujan graphs guaranteed by Theorem~\ref{thm:kamber} gives the result immediately, showing in particular that Kahale's bound for vertex expansion is optimal in the limit to leading order..
\end{proof}

\begin{remark}
    We remark here Conjecture \ref{conj:hardness} leaves open the hardness of certifying a bound at the critical value $\Phi_{\varepsilon}^v(G) \ge \frac{d}{2}$ (with no error terms) for the values of $d = q+1$ when $G \sim \mathcal{G}\left(\left(\frac{n}{2},\frac{n}{2}\right), d\right)$. We also leave the hardness of certifying vertex expansion for $G \sim \mathcal{G}(n,d)$ as a future direction; see Theorem~\ref{thm:bipartite-unipartite} and the surrounding discussion for an indication of why this is harder than the bipartite problem, at least using our tools.
\end{remark}

\subsection{Edge Expansion}

We review the versions of the above quantities, bounds, and constructions for edge rather than vertex expansion.

\begin{definition}[Edge boundary]
    Let $G$ be a graph and $S \subseteq V(G)$. The \emph{edge boundary} of $S$ is the set of edges
    \begin{equation*}
        \partial_e S \colonequals \left\{\{u,v\} \in E(G): u\in S, v\not\in S\right\}.
    \end{equation*}
    We denote the edge boundary as $\partial_e^{G} S$ when needing to specify the underlying $G$.
\end{definition}
\noindent
Recall that the edge expansion is then defined as
    \begin{equation}
        \Phi_{\epsilon}^{e}(G) \colonequals \min_{\substack{S \subseteq V \\ 1 \leq |S| \leq \epsilon n}} \frac{|\partial_e S|}{|S|}.
    \end{equation}

\begin{proposition}
    $-\Phi_{\epsilon}^{e}(G)$ is robustly lift-monotone.
\end{proposition}
\begin{proof}
    Let $H$ be a $d$-regular graph, and $L$ be any lift of $H$. Suppose $S \subseteq V(H)$ achieves the minimum edge expansion among sets of size at most $\epsilon |V(H)|$, i.e., $\Phi_{\epsilon}^{e}(H) = \frac{|\partial_e^H S|}{|S|}$. Now consider $T \subseteq V(L)$ which is the union of all the fibers of vertices $v \in S$. Note that the edge expansion of $T$ in $L$ is the same as the edge expansion of $S$ in $H$ and that $T$ satisfies $\frac{|T|}{|V(L')|} = \frac{|S|}{|V(H)|} \le \epsilon$, since the edge boundary of $T \subseteq V(L)$ is the union of the edges between the corresponding fibers of edges in the edge boundary of $S \subseteq V(H)$.

    Suppose $L'$ is a $d$-regular graph such that $\Delta(L, L') \le \gamma$. Let $n = |V(L)|$, and $E' \colonequals E(L)\setminus E(L') \subseteq E(L)$ be the set of edges in $L$ but not in $L'$. $|E'| \le \gamma n$. We note that $\partial_e^{L} T \subseteq \partial_e^{L'} T \cup E'$, and $|T| = |S|\cdot \frac{|V(L)|}{|V(H)|} = \frac{|S|}{|V(H)|} n$. Then,
    \begin{align*}
        -\Phi_{\epsilon}^{e}(H) = - \frac{|\partial_{e}^H S|}{|S|} = - \frac{|\partial_{e}^L T|}{|T|} \le -\frac{|\partial_{e}^{L'} T|}{|T|} + \frac{|E'|}{|T|} \le -\Phi_{\epsilon}^{e}(L') + \frac{\gamma n}{\frac{|S|}{|V(H)|} n} = -\Phi_{\epsilon}^{e}(L') + \frac{|V(H)|}{|S|}\gamma,
    \end{align*}
    which finishes the proof.
\end{proof}

Our benchmark certificate here is a variant of Kahale's bound for vertex expansion.
The bound says that, for an absolute constant $C > 0$,
    \begin{equation}
        \Phi_{\epsilon}^{e}(G) \geq d - \left(1 + \frac{\Tilde{\lambda}(G)}{2} + \sqrt{\frac{\Tilde{\lambda}(G)^2}{4} - (d-1)}\right)\left(1 + C \cdot \frac{\log d}{\log \frac{1}{\epsilon}}\right)
    \end{equation}
    In particular, the algorithm outputting this lower bound with high probability certifies a bound of
    \begin{equation}
        \Phi_{\epsilon}^{e}(G) \geq d - (\sqrt{d-1} + 1)\left(1 + O_d \left(\frac{1}{\log \frac{1}{\epsilon}}\right)\right) - o_{d; n \to \infty}(1)
    \end{equation}
    when $G \sim \mathcal{G}(n,d)$ or $G \sim \mathcal{G}\left(\left(\frac{n}{2}, \frac{n}{2}\right), d\right)$.
    Further, for every fixed $d$, Kahale's spectral bound in the double limit of first taking $n \to \infty$ and then taking $\varepsilon \to 0$ with high probability certifies a lower bound of $\Phi_{\epsilon}^{e}(G) \ge d -1-\sqrt{d-1} - o_{d; n \to \infty, \varepsilon \to 0}(1)$.

The same references as in the previous section again provide the Ramanujan graphs that we will use in our argument.
\begin{theorem}[{\cite[Theorem 1.7]{kamber2022combinatorics}}]
    \label{thm:kamber-2}
    For every prime power $q$, there exists an infinite family of $d = (q^2+1)$-regular bipartite Ramanujan graphs $G$, such that there exists a subset $S \subseteq V(G)$ satisfying $|S| = O(\sqrt{|V(G)|})$, and $\frac{|\partial_e S|}{|S|} = d - 1 - \sqrt{d-1}$.
\end{theorem}

\begin{proof}[Proof of Theorem~\ref{thm:edge-exp}]
    The proof is again immediate using Theorem~\ref{thm:certification}, Proposition~\ref{prop:ramanujan-bounds}, and the graphs provided above by Theorem~\ref{thm:kamber-2}.
    In particular, conditional on Conjecture~\ref{conj:hardness-formal}, Kahale's bound for edge expansion in the limit of $\varepsilon \to 0$ is optimal to leading order.
\end{proof}

\begin{remark}
    As for vertex expansion, our argument leaves open the hardness of certifying a bound at the critical value $\Phi_{\varepsilon}^e(G) \ge d - 1 - \sqrt{d-1}$ for the values of $d = q^2+1$ when $G \sim \mathcal{G}\left(\left(\frac{n}{2},\frac{n}{2}\right), d\right)$.
    The hardness of certifying edge expansion for $G \sim \mathcal{G}(n,d)$ is again an interesting question and, per Theorem~\ref{thm:bipartite-unipartite}, likely harder to establish theoretical lower bounds for than the bipartite case we have treated here.
\end{remark}

\clearpage

\begin{figure}[!h]
    \begin{center}
        \includegraphics[scale=0.65]{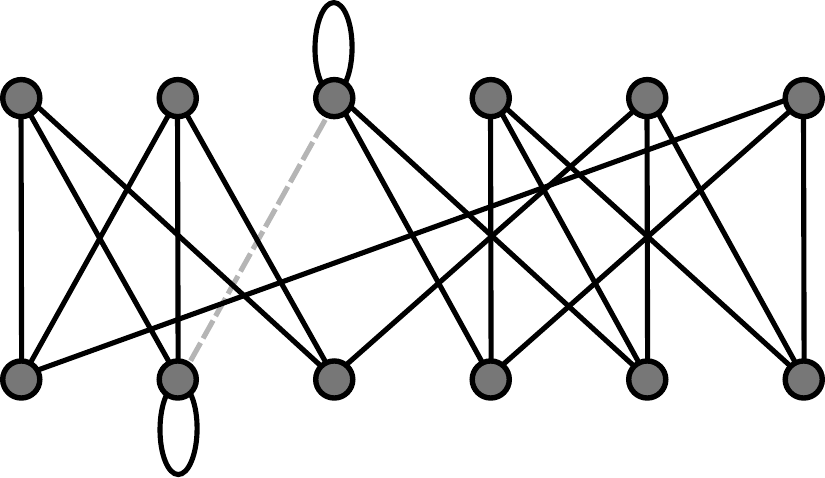}
    \end{center}

    \vspace{1em}

    \[ \mathlarger{\left[ \begin{array}{cccccccccccc}
1 & 0 & 0 & 0 & 0 & 0 & 0 & 1 & 1 & 0 & 0 & 0 \\
0 & 0 & 0 & 0 & 0 & 0 & 1 & 0 & 0 & 1 & 1 & 0 \\
0 & 0 & 0 & 0 & 0 & 0 & 1 & 0 & 0 & 1 & 0 & 1 \\
0 & 0 & 0 & 0 & 0 & 0 & 0 & 1 & 1 & 0 & 1 & 0 \\
0 & 0 & 0 & 0 & 0 & 0 & 0 & 1 & 0 & 1 & 0 & 1 \\
0 & 0 & 0 & 0 & 0 & 0 & 0 & 0 & 1 & 0 & 1 & 1 \\
0 & 1 & 1 & 0 & 0 & 0 & 1 & 0 & 0 & 0 & 0 & 0 \\
1 & 0 & 0 & 1 & 1 & 0 & 0 & 0 & 0 & 0 & 0 & 0 \\
1 & 0 & 0 & 1 & 0 & 1 & 0 & 0 & 0 & 0 & 0 & 0 \\
0 & 1 & 1 & 0 & 1 & 0 & 0 & 0 & 0 & 0 & 0 & 0 \\
0 & 1 & 0 & 1 & 0 & 1 & 0 & 0 & 0 & 0 & 0 & 0 \\
0 & 0 & 1 & 0 & 1 & 1 & 0 & 0 & 0 & 0 & 0 & 0
\end{array} \right]} \]

    \vspace{1em}

    \caption{The 3-regular Ramanujan graph used in the proof of Theorems~\ref{thm:max-cut} and \ref{thm:ind-set} on the maximum cut and maximum independent set of 3-regular graphs, respectively. The graph is formed from a bipartite 3-regular graphs by replacing one edge (shown dashed in gray) with a pair of loops. This graph has $d = 3, |V(H)| = 12$, $\max\{|\lambda_2(H)|, |\lambda_n(H)|\} \approx \num{2.825} < \num{2.828} \approx 2\sqrt{2}$, $\MC_2(H) = \frac{17}{18} \approx \num{0.944}$, and $\what{\alpha}^{\prime}(H) = \frac{5 + 1/2}{12} \approx \num{0.458}$, where $\what{\alpha}^{\prime}$ is the modified independence number discussed in the proof in Section~\ref{sec:ind-set}.}
    \label{fig:d3-k2-example}
\end{figure}

\begin{figure}[!h]
    \begin{center}
        \includegraphics[scale=0.65]{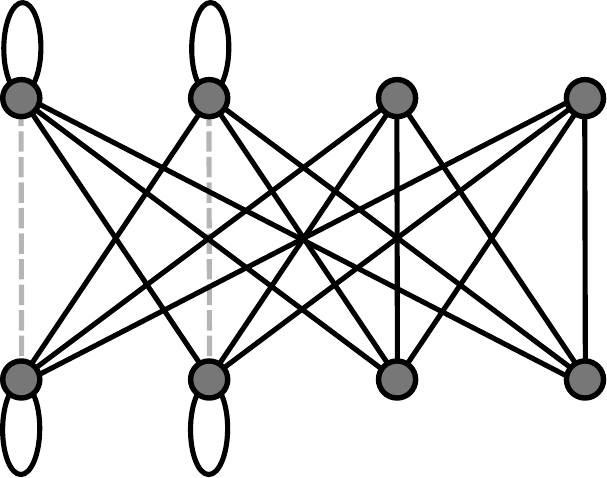}
    \end{center}

    \vspace{1em}
    \[
    \mathlarger{\left[
    \begin{array}{cccccccc}
1 & 0 & 0 & 0 & 0 & 1 & 1 & 1 \\
0 & 1 & 0 & 0 & 1 & 0 & 1 & 1 \\
0 & 0 & 0 & 0 & 1 & 1 & 1 & 1 \\
0 & 0 & 0 & 0 & 1 & 1 & 1 & 1 \\
0 & 1 & 1 & 1 & 1 & 0 & 0 & 0 \\
1 & 0 & 1 & 1 & 0 & 1 & 0 & 0 \\
1 & 1 & 1 & 1 & 0 & 0 & 0 & 0 \\
1 & 1 & 1 & 1 & 0 & 0 & 0 & 0
\end{array}
    \right]}
    \]

    \vspace{1em}

    \caption{The 4-regular Ramanujan graph used in the proof of Theorem~\ref{thm:max-cut} on the maximum cut of 4-regular graphs. The graph is formed from the complete bipartite graph $K_{4, 4}$ by replacing two edges (shown dashed in gray) with pairs of loops. This graph has $d = 4, |V(H)| = 8, \max\{|\lambda_2(H)|, |\lambda_n(H)|\} \approx \num{3.236} < \num{3.464} \approx 2\sqrt{3}$, and $\MC_2(H) = \frac{14}{16} = \num{0.875}$.}
    \label{fig:d4-k2-example}
\end{figure}

\begin{figure}[!h]
    \begin{center}
        \includegraphics[scale=0.65]{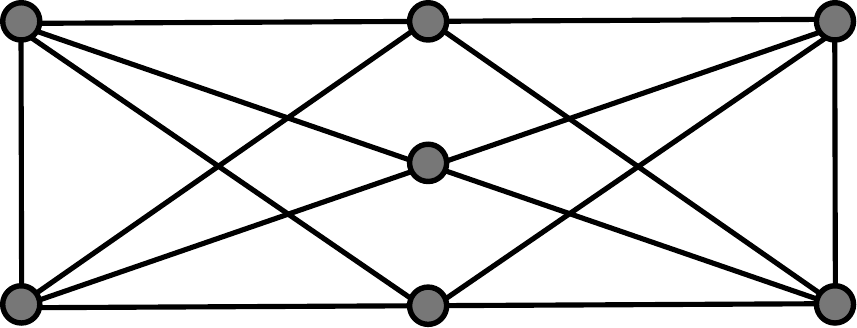}
    \end{center}

    \vspace{1em}

    \[ \mathlarger{\left[\begin{array}{ccccccc}
0 & 1 & 1 & 1 & 1 & 0 & 0 \\
1 & 0 & 1 & 1 & 1 & 0 & 0 \\
1 & 1 & 0 & 0 & 0 & 1 & 1 \\
1 & 1 & 0 & 0 & 0 & 1 & 1 \\
1 & 1 & 0 & 0 & 0 & 1 & 1 \\
0 & 0 & 1 & 1 & 1 & 0 & 1 \\
0 & 0 & 1 & 1 & 1 & 1 & 0
\end{array} \right]} \]

    \vspace{1em}

    \caption{The 4-regular Ramanujan graph used in the proof of Theorem~\ref{thm:ind-set} on the maximum independent set of 4-regular graphs. This graph appears in Figure 3 of \cite{JHSWZ-2024-4RegularGraphsRigidity} in a different context. This graph has $d = 4$, $|V(H)| = 7$, $\max\{|\lambda_2(H)|, |\lambda_n(H)|\} = 3 < \num{3.464} \approx 2\sqrt{3}$, and $\what{\alpha}(H) = \frac{3}{7} \approx \num{0.428}$.}
    \label{fig:d4-ind-example}
\end{figure}

\begin{figure}[!h]
    \begin{center}
        \includegraphics[scale=0.65]{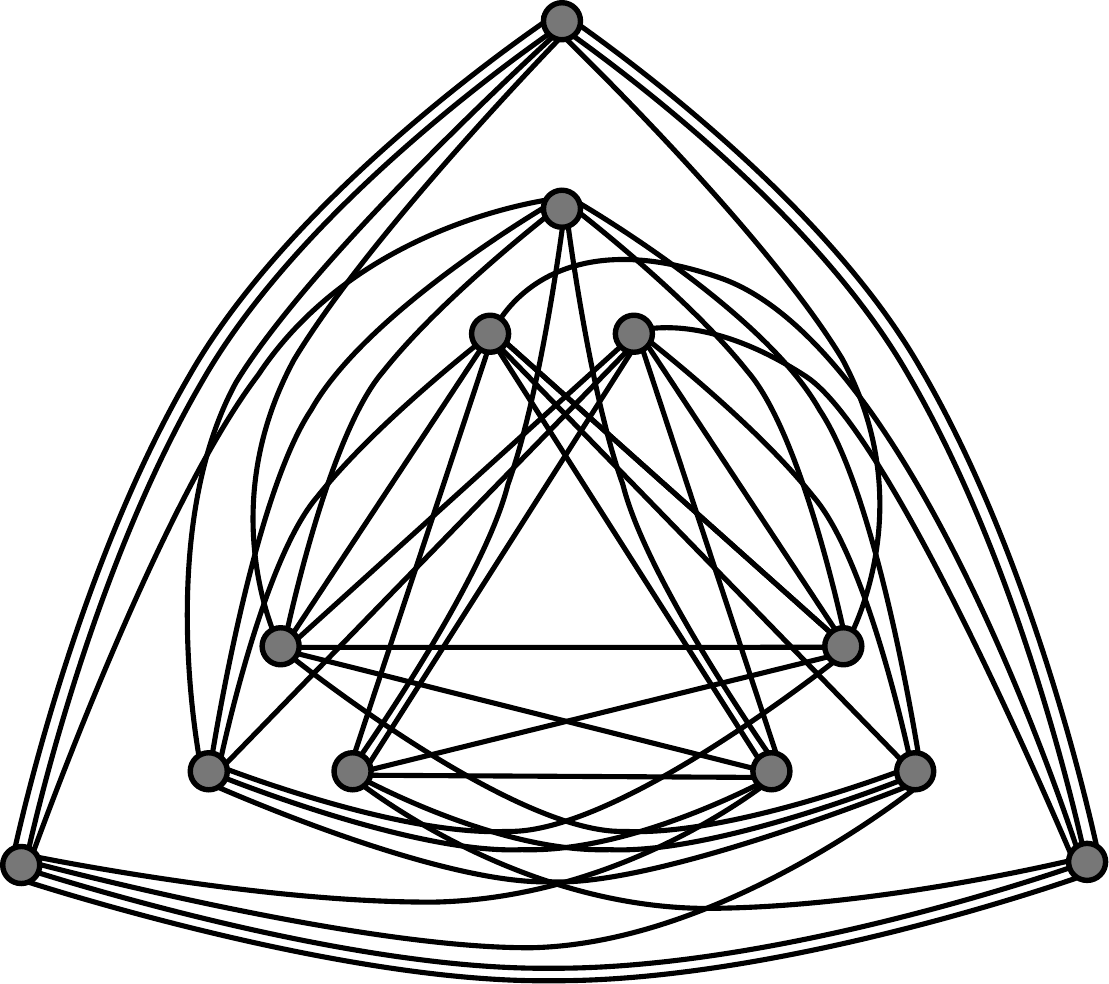}
    \end{center}

    \vspace{1em}

    \[ \mathlarger{\left[\begin{array}{cccccccccccc}
0 & 0 & 0 & 0 & 1 & 1 & 1 & 0 & 1 & 1 & 1 & 1 \\
0 & 0 & 0 & 0 & 1 & 1 & 1 & 1 & 1 & 1 & 1 & 0 \\
0 & 0 & 0 & 0 & 1 & 1 & 1 & 0 & 1 & 1 & 1 & 1 \\
0 & 0 & 0 & 0 & 0 & 1 & 1 & 2 & 1 & 0 & 0 & 2 \\
1 & 1 & 1 & 0 & 0 & 0 & 0 & 0 & 1 & 1 & 1 & 1 \\
1 & 1 & 1 & 1 & 0 & 0 & 0 & 0 & 1 & 1 & 1 & 0 \\
1 & 1 & 1 & 1 & 0 & 0 & 0 & 0 & 1 & 1 & 1 & 0 \\
0 & 1 & 0 & 2 & 0 & 0 & 0 & 0 & 0 & 1 & 1 & 2 \\
1 & 1 & 1 & 1 & 1 & 1 & 1 & 0 & 0 & 0 & 0 & 0 \\
1 & 1 & 1 & 0 & 1 & 1 & 1 & 1 & 0 & 0 & 0 & 0 \\
1 & 1 & 1 & 0 & 1 & 1 & 1 & 1 & 0 & 0 & 0 & 0 \\
1 & 0 & 1 & 2 & 1 & 0 & 0 & 2 & 0 & 0 & 0 & 0
\end{array}\right]} \]

    \vspace{1em}

    \caption{The 7-regular Ramanujan graph used in the proof of Theorem~\ref{thm:chromatic} on the chromatic number. The graph is formed by modifying the complete multipartite graph $K_{3, 3, 3}$ (a Ramanujan 6-regular graph), which appears as the induced subgraph on the ``inner'' nine vertices. Note that there are three pairs of repeated edges, forming a triangle on the ``outer'' three vertices. This graph has $d = 7, |V(H)| = 12, \max\{|\lambda_2(H)|, |\lambda_n(H)|\} \approx \num{3.791} < \num{4.898} \approx 2\sqrt{6}$, and $\chi(H) = 3$.}
    \label{fig:d7-chrom-example}
\end{figure}

\clearpage

\section*{Acknowledgments}
\addcontentsline{toc}{section}{Acknowledgments}

We thank Dan Spielman for helpful discussions.
We also thank the authors of \cite{BBKMW-2020-SpectralPlantingColoring}, Jess Banks, Alex Wein, Afonso Bandeira, and Cris Moore; it is in collaboration with them that the idea for Conjecture~\ref{conj:hardness} and its uses in proving lower bounds for certification came about.
We have explored this strategy and its consequences with their gracious permission to continue this investigation independently and to publicize the Conjecture.

\addcontentsline{toc}{section}{References}
\bibliographystyle{alpha}
\bibliography{main}

\newcommand{\etalchar}[1]{$^{#1}$}
\begin{thebibliography}{COLMS22}

\bibitem[ACT16]{ACT-2016-SpectralBoundsKIndependence}
Aida Abiad, Sebastian~M Cioab{\u{a}}, and Michael Tait.
\newblock Spectral bounds for the $k$-independence number of a graph.
\newblock {\em Linear Algebra and its Applications}, 510:160--170, 2016.

\bibitem[AL06]{AL-2006-RandomLiftsExpansion}
Alon Amit and Nathan Linial.
\newblock Random lifts of graphs: edge expansion.
\newblock {\em Combinatorics, Probability and Computing}, 15(3):317--332, 2006.

\bibitem[ALM02]{ALM-2002-RandomLiftsIndependenceChromatic}
Alon Amit, Nathan Linial, and Ji{\v{r}}{\'\i} Matou{\v{s}}ek.
\newblock Random lifts of graphs: independence and chromatic number.
\newblock {\em Random Structures \& Algorithms}, 20(1):1--22, 2002.

\bibitem[ALMR01]{ALMR-2001-RandomLifts}
Alon Amit, Nathan Linial, Ji{\v{r}}{\'\i} Matou{\v{s}}ek, and Eyal Rozenman.
\newblock Random lifts of graphs.
\newblock In {\em Proceedings of the twelfth annual ACM-SIAM symposium on
  Discrete algorithms}, pages 883--894, 2001.

\bibitem[AM04]{AM-2004-ColoringRandomRegular}
Dimitris Achlioptas and Cristopher Moore.
\newblock The chromatic number of random regular graphs.
\newblock In {\em Approximation, Randomization, and Combinatorial Optimization.
  Algorithms and Techniques}, pages 219--228. Springer, 2004.

\bibitem[AW10]{alon2010high}
Noga Alon and Nicholas Wormald.
\newblock High degree graphs contain large-star factors.
\newblock In {\em Fete of Combinatorics and Computer Science}, pages 9--21.
  Springer, 2010.

\bibitem[BBK{\etalchar{+}}21]{BBKMW-2020-SpectralPlantingColoring}
Afonso~S Bandeira, Jess Banks, Dmitriy Kunisky, Cristopher Moore, and
  Alexander~S Wein.
\newblock Spectral planting and the hardness of refuting cuts, colorability,
  and communities in random graphs.
\newblock In {\em 34th Annual Conference on Learning Theory (COLT 2021)}, pages
  410--473. PMLR, 2021.

\bibitem[BC19]{BC-2019-EigenvaluesRandomLifts}
Charles Bordenave and Beno{\^\i}t Collins.
\newblock Eigenvalues of random lifts and polynomials of random permutation
  matrices.
\newblock {\em Annals of Mathematics}, 190(3):811--875, 2019.

\bibitem[BDH22]{BDH-2022-SpectralGapRandomBipartiteBiregular}
Gerandy Brito, Ioana Dumitriu, and Kameron~Decker Harris.
\newblock Spectral gap in random bipartite biregular graphs and applications.
\newblock {\em Combinatorics, Probability and Computing}, 31(2):229--267, 2022.

\bibitem[BHK{\etalchar{+}}19]{BHKKMP-2019-PlantedClique}
Boaz Barak, Samuel~B Hopkins, Jonathan Kelner, Pravesh~K Kothari, Ankur Moitra,
  and Aaron Potechin.
\newblock A nearly tight sum-of-squares lower bound for the planted clique
  problem.
\newblock {\em SIAM Journal on Computing}, 48(2):687--735, 2019.

\bibitem[BHM80]{BHM-1980-BigraphsDigraphs}
Richard~A Brualdi, Frank Harary, and Zevi Miller.
\newblock Bigraphs versus digraphs via matrices.
\newblock {\em Journal of Graph Theory}, 4(1):51--73, 1980.

\bibitem[Bil06]{Bilu-2006-ExtensionsHoffmanBound}
Yonatan Bilu.
\newblock Tales of {Hoffman}: Three extensions of {Hoffman's} bound on the
  graph chromatic number.
\newblock {\em Journal of Combinatorial Theory, Series B}, 96(4):608--613,
  2006.

\bibitem[BKW20]{BKW-2019-ConstrainedPCA}
Afonso~S Bandeira, Dmitriy Kunisky, and Alexander~S Wein.
\newblock Computational hardness of certifying bounds on constrained {PCA}
  problems.
\newblock In {\em 11th Innovations in Theoretical Computer Science Conference
  ({ITCS} 2020)}, volume 151, pages 78:1--78:29, 2020.

\bibitem[BKW22]{BKW-2020-PositivePCA}
Afonso~S Bandeira, Dmitriy Kunisky, and Alexander~S Wein.
\newblock Average-case integrality gap for non-negative principal component
  analysis.
\newblock In {\em Mathematical and Scientific Machine Learning}, pages
  153--171. PMLR, 2022.

\bibitem[BMR21]{banks2021local}
Jess Banks, Sidhanth Mohanty, and Prasad Raghavendra.
\newblock Local statistics, semidefinite programming, and community detection.
\newblock In {\em Proceedings of the 2021 ACM-SIAM Symposium on Discrete
  Algorithms (SODA)}, pages 1298--1316. SIAM, 2021.

\bibitem[Bor15]{Bordenave-2015-Friedman}
Charles Bordenave.
\newblock A new proof of {Friedman's} second eigenvalue theorem and its
  extension to random lifts.
\newblock {\em arXiv preprint arXiv:1502.04482}, 2015.

\bibitem[COLMS22]{COLMS-2022-MaxCutRandomRegularIsing}
Amin Coja-Oghlan, Philipp Loick, Bal{\'a}zs~F Mezei, and Gregory~B Sorkin.
\newblock The ising antiferromagnet and max cut on random regular graphs.
\newblock {\em SIAM Journal on Discrete Mathematics}, 36(2):1306--1342, 2022.

\bibitem[Cs{\'o}16]{Csoka-2016-IndependentSetsCutsLargeGirthRegular}
Endre Cs{\'o}ka.
\newblock Independent sets and cuts in large-girth regular graphs.
\newblock {\em arXiv preprint arXiv:1602.02747}, 2016.

\bibitem[DDSW03]{DDSW-2003-MaxMinBisection34Regular}
Josep D{\i}az, Norman Do, Maria~J Serna, and Nicholas~C Wormald.
\newblock Bounds on the max and min bisection of random cubic and random
  4-regular graphs.
\newblock {\em Theoretical computer science}, 307(3):531--547, 2003.

\bibitem[DSC96]{DSC-1996-LSIFiniteMarkovChains}
Persi Diaconis and Laurent Saloff-Coste.
\newblock Logarithmic {Sobolev} inequalities for finite {Markov} chains.
\newblock {\em The Annals of Applied Probability}, 6(3):695--750, 1996.

\bibitem[DW02]{DW-2002-IndependentDominatingSetCubic}
William Duckworth and Nicholas~C Wormald.
\newblock Minimum independent dominating sets of random cubic graphs.
\newblock {\em Random Structures \& Algorithms}, 21(2):147--161, 2002.

\bibitem[FF23]{FF-2023-SOSProofsLSI}
Ois{\'\i}n Faust and Hamza Fawzi.
\newblock Sum-of-squares proofs of logarithmic {Sobolev} inequalities on finite
  {Markov} chains.
\newblock {\em IEEE Transactions on Information Theory}, 2023.

\bibitem[F{\L}92]{FL-1992-IndependenceChromaticRandomRegular}
Alan~M Frieze and Tomasz {\L}uczak.
\newblock On the independence and chromatic numbers of random regular graphs.
\newblock {\em Journal of Combinatorial Theory, Series B}, 54(1):123--132,
  1992.

\bibitem[FL96]{FL-1996-SpectraHypergraphs}
Keqin Feng and Wen-Ch'ing~Winnie Li.
\newblock Spectra of hypergraphs and applications.
\newblock {\em Journal of number theory}, 60(1):1--22, 1996.

\bibitem[Fri03]{Friedman-2003-AlonConjecture}
Joel Friedman.
\newblock A proof of {Alon's} second eigenvalue conjecture.
\newblock In {\em 35th Annual ACM Symposium on Theory of Computing (STOC
  2003)}, pages 720--724. ACM, 2003.

\bibitem[GJJ{\etalchar{+}}20]{GJJPR-2020-SK}
Mrinalkanti Ghosh, Fernando~Granha Jeronimo, Chris Jones, Aaron Potechin, and
  Goutham Rajendran.
\newblock Sum-of-squares lower bounds for {Sherrington-Kirkpatrick} via planted
  affine planes.
\newblock In {\em 61st Annual Symposium on Foundations of Computer Science
  (FOCS 2020)}, pages 954--965. IEEE, 2020.

\bibitem[GL18]{GL-2018-MaxCutSparseRandomGraphs}
David Gamarnik and Quan Li.
\newblock On the max-cut of sparse random graphs.
\newblock {\em Random Structures \& Algorithms}, 52(2):219--262, 2018.

\bibitem[GLS12]{GLS-2012-GeometricAlgorithmsCombinatorialOptimization}
Martin Gr{\"o}tschel, L{\'a}szl{\'o} Lov{\'a}sz, and Alexander Schrijver.
\newblock {\em Geometric algorithms and combinatorial optimization}, volume~2.
\newblock Springer Science \& Business Media, 2012.

\bibitem[Hae21]{Haemers-2021-HoffmanRatioBound}
Willem~H Haemers.
\newblock Hoffman's ratio bound.
\newblock {\em Linear Algebra and its Applications}, 617:215--219, 2021.

\bibitem[Har23]{Harangi-2023-ReplicaBoundsIndependenceRandomRegular}
Viktor Harangi.
\newblock Improved replica bounds for the independence ratio of random regular
  graphs.
\newblock {\em Journal of Statistical Physics}, 190(3):60, 2023.

\bibitem[HLW06]{hoory2006expander}
Shlomo Hoory, Nathan Linial, and Avi Wigderson.
\newblock Expander graphs and their applications.
\newblock {\em Bulletin of the American Mathematical Society}, 43(4):439--561,
  2006.

\bibitem[Hof70]{Hoffman-1970-Eigenvalues}
Alan~J Hoffman.
\newblock On eigenvalues and colorings of graphs.
\newblock In Bernard Harris, editor, {\em Graph theory and its applications}.
  1970.

\bibitem[Jan95]{Janson-1995-RandomRegularGraphsContiguity}
Svante Janson.
\newblock Random regular graphs: asymptotic distributions and contiguity.
\newblock {\em Combinatorics, Probability and Computing}, 4(4):369--405, 1995.

\bibitem[JHS{\etalchar{+}}24]{JHSWZ-2024-4RegularGraphsRigidity}
Tibor Jord{\'a}n, Robin Huang, Henry Simmons, Kaylee Weatherspoon, and Zeyu
  Zheng.
\newblock Four-regular graphs with extremal rigidity properties.
\newblock {\em Discrete Mathematics}, 347(4):113833, 2024.

\bibitem[JPR{\etalchar{+}}21]{JPRTX-2021-SOSSparseIndependentSet}
Chris Jones, Aaron Potechin, Goutham Rajendran, Madhur Tulsiani, and Jeff Xu.
\newblock Sum-of-squares lower bounds for sparse independent set.
\newblock {\em arXiv preprint arXiv:2111.09250}, 2021.

\bibitem[Kah95]{Kahale-1995-SpectralBoundExpansion}
Nabil Kahale.
\newblock Eigenvalues and expansion of regular graphs.
\newblock {\em Journal of the ACM (JACM)}, 42(5):1091--1106, 1995.

\bibitem[KK22]{kamber2022combinatorics}
Amitay Kamber and Tali Kaufman.
\newblock Combinatorics via closed orbits: number theoretic {Ramanujan} graphs
  are not unique neighbor expanders.
\newblock In {\em Proceedings of the 54th Annual ACM SIGACT Symposium on Theory
  of Computing}, pages 426--435, 2022.

\bibitem[Kun24]{Kunisky-2023-OptimalityGlauberDynamicsIsing}
Dmitriy Kunisky.
\newblock Optimality of {Glauber} dynamics for general-purpose {Ising} model
  sampling and free energy approximation.
\newblock {\em 2024 Annual ACM-SIAM Symposium on Discrete Algorithms (SODA
  2024)}, 2024.

\bibitem[KVWX23]{KVWX-2023-LowDegreeColoringClique}
Pravesh~K Kothari, Santosh~S Vempala, Alexander~S Wein, and Jeff Xu.
\newblock Is planted coloring easier than planted clique?
\newblock In {\em 36th Annual Conference on Learning Theory (COLT 2023)}. PMLR,
  2023.

\bibitem[KWB22]{KWB-2022-LowDegreeNotes}
Dmitriy Kunisky, Alexander~S Wein, and Afonso~S Bandeira.
\newblock Notes on computational hardness of hypothesis testing: Predictions
  using the low-degree likelihood ratio.
\newblock In Paula Cerejeiras and Michael Reissig, editors, {\em Mathematical
  Analysis, its Applications and Computation}, pages 1--50, Cham, 2022.
  Springer International Publishing.

\bibitem[Las01]{Lasserre-2001-GlobalOptimizationMoments}
Jean~B Lasserre.
\newblock Global optimization with polynomials and the problem of moments.
\newblock {\em SIAM Journal on Optimization}, 11(3):796--817, 2001.

\bibitem[Lau09]{Laurent-2009-SOS}
Monique Laurent.
\newblock Sums of squares, moment matrices and optimization over polynomials.
\newblock In {\em Emerging applications of algebraic geometry}, pages 157--270.
  Springer, 2009.

\bibitem[LS96]{LS-1996-SpectraRegularGraphsBipartite}
Wen-Ch'ing~Winnie Li and Patrick Sol{\'e}.
\newblock Spectra of regular graphs and hypergraphs and orthogonal polynomials.
\newblock {\em European Journal of Combinatorics}, 17(5):461--477, 1996.

\bibitem[McK81]{McKay-1981-EigenvalueRegularGraph}
Brendan~D McKay.
\newblock The expected eigenvalue distribution of a large regular graph.
\newblock {\em Linear Algebra and its Applications}, 40:203--216, 1981.

\bibitem[Mor94]{Morgenstern-1994-RamanujanGraphsPrimePowers}
Moshe Morgenstern.
\newblock Existence and explicit constructions of $q+ 1$ regular {Ramanujan}
  graphs for every prime power $q$.
\newblock {\em Journal of Combinatorial Theory, Series B}, 62(1):44--62, 1994.

\bibitem[MSS13]{MSS-2013-InterlacingFamiliesBipartiteRamanujan}
Adam Marcus, Daniel~A Spielman, and Nikhil Srivastava.
\newblock Interlacing families {I}: Bipartite {Ramanujan} graphs of all
  degrees.
\newblock In {\em 2013 IEEE 54th Annual Symposium on Foundations of computer
  science}, pages 529--537. IEEE, 2013.

\bibitem[Nes98]{Nesterov-1998-SemidefiniteQuadratic}
Yurii Nesterov.
\newblock Semidefinite relaxation and nonconvex quadratic optimization.
\newblock {\em Optimization Methods and Software}, 9(1-3):141--160, 1998.

\bibitem[Nil91]{Nilli-1991-AlonBoppanaBound}
Alon Nilli.
\newblock On the second eigenvalue of a graph.
\newblock {\em Discrete Mathematics}, 91(2):207--210, 1991.

\bibitem[O'D17]{ODonnell-2017-SOSNotAutomatizable}
Ryan O'Donnell.
\newblock {SOS} is not obviously automatizable, even approximately.
\newblock In {\em 8th Innovations in Theoretical Computer Science Conference
  (ITCS 2017)}. Schloss Dagstuhl-Leibniz-Zentrum fuer Informatik, 2017.

\bibitem[Olv10]{Olver-2010-NISTHandbook}
Frank~WJ Olver.
\newblock {\em {NIST} handbook of mathematical functions}.
\newblock Cambridge University Press, 2010.

\bibitem[Par00]{Parrilo-2000-Thesis}
Pablo~A Parrilo.
\newblock {\em Structured semidefinite programs and semialgebraic geometry
  methods in robustness and optimization}.
\newblock PhD thesis, California Institute of Technology, 2000.

\bibitem[PR20]{PR-2020-MachinerySOS}
Aaron Potechin and Goutham Rajendran.
\newblock Machinery for proving sum-of-squares lower bounds on certification
  problems.
\newblock {\em arXiv preprint arXiv:2011.04253}, 2020.

\bibitem[RW17]{RW-2017-BitComplexity}
Prasad Raghavendra and Benjamin Weitz.
\newblock On the bit complexity of sum-of-squares proofs.
\newblock In Ioannis Chatzigiannakis, Piotr Indyk, Fabian Kuhn, and Anca
  Muscholl, editors, {\em 44th International Colloquium on Automata, Languages,
  and Programming (ICALP 2017)}, volume~80 of {\em Leibniz International
  Proceedings in Informatics (LIPIcs)}, pages 80:1--80:13, Dagstuhl, Germany,
  2017. Schloss Dagstuhl--Leibniz-Zentrum fuer Informatik.

\bibitem[Sen18]{Sen-2018-OptimizationSparseHypergraph}
Subhabrata Sen.
\newblock Optimization on sparse random hypergraphs and spin glasses.
\newblock {\em Random Structures \& Algorithms}, 53(3):504--536, 2018.

\bibitem[Sho87]{Shor-1987-SumOfSquares}
Naum~Zuselevich Shor.
\newblock An approach to obtaining global extremums in polynomial mathematical
  programming problems.
\newblock {\em Cybernetics}, 23(5):695--700, 1987.

\end{thebibliography}

\end{document}